\documentclass[]{article}
\usepackage[english]{babel}
\usepackage{tikz}
\usepackage{amssymb}
\usepackage[mathscr]{eucal}
\usepackage[T1]{fontenc}
\usepackage[utf8]{inputenc}
\usepackage{calrsfs}
\usepackage{float}
\usepackage{amsmath}
\usepackage{amsthm}
\usepackage{enumerate}
\usepackage{multirow}
\usepackage[titletoc,toc]{appendix}
\usepackage{authblk}
\usepackage{rotating}
\usetikzlibrary[shapes.misc]
\usetikzlibrary[shapes.geometric]
\newtheorem{de}{Definition}[section]
\newtheorem{theo}{Theorem}[section]
\newtheorem{prop}[theo]{Proposition}
\newtheorem{lemma}[theo]{Lemma}
\newtheorem{cor}[theo]{Corollary}

\newtheorem{obs}[theo]{Observation}

\title{Completely Independent Spanning Trees in Some Regular Networks}

\author[1]{Benoit Darties}
\author[1,2]{Nicolas Gastineau\thanks{Author partially supported by the Burgundy Council}}
\author[1]{Olivier Togni}
\affil[1]{\textit{LE2I, UMR CNRS 6306} \\ \textit{Universit\'e de Bourgogne, 21078 Dijon cedex, France} }
\affil[2]{\textit{LIRIS, UMR CNRS 5205} \\ \textit{Université Claude Bernard Lyon 1, Université de Lyon, F-69622 , France} }
\newcommand{\inn}{\text{IN}}

\begin{document}
\maketitle

\begin{abstract}
Let $k\ge 2$ be an integer and $T_1,\ldots, T_k$ be spanning trees of a graph $G$. If for any pair of vertices $(u,v)$ of $V(G)$, the paths from $u$ to $v$ in each $T_i$, $1\le i\le k$, do not contain common edges and common vertices, except the vertices $u$ and $v$, then $T_1,\ldots, T_k$ are completely independent spanning trees in $G$. For $2k$-regular graphs which are $2k$-connected, 
such as the Cartesian product of a complete graph of order $2k-1$ and a cycle and some Cartesian products of three cycles (for $k=3$), the maximum number of completely independent spanning trees contained in these graphs is determined and it turns out that this maximum is not always $k$.
\end{abstract}

\section{Introduction}
Let $k\ge 2$ be an integer and $T_1,\ldots, T_k$ be spanning trees in a graph $G$. The spanning trees $T_1,\ldots, T_k$ are {\em edge-disjoint} if $E(T_1)\cap \cdots \cap E(T_k)=\emptyset$.
For a given tree $T$ and a given pair of vertices $(u,v)$ of $T$, let $P_{T}(u,v)$ be the set of vertices in the unique path between $u$ and $v$ in $T$.
The spanning trees $T_1,\ldots, T_k$ are {\em internally disjoint} if for any pair of vertices $(u,v)$ of $V(G)$, $P_{T_1}(u,v)\cap\cdots\cap P_{T_k}(u,v)=\{u,v\}$.
Finally, the spanning trees $T_1,\ldots, T_k$ are {\em completely independent spanning trees} if they are pairwise edge disjoint and internally disjoint.

Disjoint spanning trees have been extensively studied as they are of practical interest for fault-tolerant broadcasting or load-balancing communication systems in interconnection networks : a spanning-tree is often used in various network operations; computing completely independent spanning-trees guarantees a continuity of service, as each can be immediately used as backup spanning tree if a node or link failure occurs on the current spanning tree. Thus, computing $k$ completely independent spanning trees allows to handle up to $k-1$ simultaneous independent node or link failures. In this context, a network is often modeled by a graph $G$ in which the set of vertices $V(G)$ corresponds to the nodes set and the set of edges $E(G)$ to the set of direct links between nodes.


Completely independent spanning trees were introduced by T. Hasunuma \cite{HA2001} and then have been studied on different classes of graphs, such as underlying graphs of line graphs \cite{HA2001}, maximal planar graphs \cite{HA2002}, Cartesian product of two cycles \cite{HA2012} and complete graphs, complete bipartite and tripartite graphs \cite{PAI2013}.
Moreover, the decision problem that consists in determining if there exist two completely independent spanning trees in a graph $G$ is NP-hard \cite{HA2002}. 

Other works on disjoint spanning trees include independent spanning trees which focus on finding spanning trees $T_1,\ldots, T_k$ rooted at $r$, such that for any vertex $v$ the paths from $r$ to $v$ in $T_1,\ldots, T_k$ are pairwise openly disjoint. the main difference is that $T_1,\ldots, T_k$ are rooted at $r$ and only the paths to $r$ are considered. Thus  $T_1,\ldots, T_k$ may share common edges, which is not admissible with completely independent spanning trees.  Independent spanning trees have been studied in several topologies, including product graphs \cite{OIBI1996}, de Bruijn and Kautz digraphs \cite{GH1997,HN2001}, and chordal rings \cite{IKOI1999}. Related works also include Edge-disjoint spanning trees, i.e. spanning-trees which are pairwise edge disjoint only. Edge-disjoint spanning trees have been studied on many classes of graphs, including hypercubes \cite{BA1999}, Cartesian product of cycles \cite{WA2001} and Cartesian product of two graphs \cite{KU2003}.

We use the following notations : for a tree, a vertex that is not a leaf is called an {\em inner vertex}. For a vertex $u$ of a graph $G$, let $d_G(u)$ be its degree in $G$, i.e. the number of edges of $G$ incident with it.

For clarity, we recall the definition of  the Cartesian product of two graphs : Given two graphs $G$ and $H$, the Cartesian product of $G$ and $H$, denoted $G\square H$, is the graph with vertex set $V(G)\times V(H)$ and edge set $\{(u,u')(v,v')|( u=v \land u'v'\in E(H))\lor ( u'=v' \land uv\in E(G)) \}$.

The following theorem gives an alternative definition \cite{HA2001} of completely independent spanning trees.
\begin{theo}[\cite{HA2001}]
Let $k\ge 2$ be an integer. $T_1,\ldots,T_k$ are completely independent spanning trees in a graph $G$ if and only if they are edge-disjoint spanning trees of $G$ and for any $v\in V(G)$, there is at most one $T_i$ such that $d_{T_i}(v)>1$.
\label{allf}
\end{theo}

It has been conjectured that in any $2k$-connected graph, there are $k$ completely independent spanning trees \cite{HA2002}. This conjecture has been refuted, as there exist $2k$-connected graphs which do not contain two completely independent spanning trees \cite{FE2011}, for any integer $k$. However, the given counterexamples are not $2k$-regular.
\begin{prop}[\cite{FE2011}]
For any $k\ge 2$, there exist $2k$-connected graphs that do not contain two completely independent spanning trees.
\end{prop}

The proof of the previous proposition consists in constructing a $2k$-connected graph with a large proportion of vertices of degree $2k$ adjacent to the same vertices and proving that these vertices of degree $2k$ can not be all adjacent to inner vertices in a fixed tree. 

This article is organized as follows. Section~\ref{sec:necessaryConds} presents necessary conditions on $2r$-regular graphs in order to have $r$ completely independent spanning trees.
Section~\ref{sec:cartprod} presents the maximum number of completely independent spanning trees in $K_{m}\square C_n$, for $n\ge3$ and $m\ge 3$. In particular, we exhibit the first $2r$-regular graphs which are $2r$-connected and which do not contain $r$ completely independent spanning trees.
In Section~\ref{sec:3Dgrid}, we determine three completely independent spanning trees in some Cartesian products of three cycles $C_{n_1}\square C_{n_2} \square C_{n_3}$, for $3\le n_1\le n_2 \le n_3$.

\section{Necessary conditions on $2r$-regular graphs}
\label{sec:necessaryConds}

\begin{prop}
If in a $2r$-regular graph $G$ there exist $r$ completely independent spanning trees, then every spanning tree has maximum degree at most $r+1$.
\label{pr+1}
\end{prop}
\begin{proof}
By Theorem \ref{allf}, every vertex should be of degree 1 in every spanning tree except in one spanning tree. Hence, in a spanning tree, a vertex is either of degree 1 (a leaf) or has degree between 2 and $r+1$ (an inner vertex), as $2r-(r-1)=r+1$.
\end{proof}

Let $\inn(T)$ be the set of inner vertices in a tree $T$.
\begin{prop}
If in a $2r$-regular graph $G$ of order $n$ there exist $r$ completely independent spanning trees, then there exists a spanning tree $T$ among them such that $|\inn(T)|\le \lfloor n/r \rfloor$.
\label{tinf}
\end{prop}
\begin{proof}Let $T_1, \ldots, T_r$ be completely independent spanning trees in $G$ and suppose that $|\inn(T_i)|> \lfloor n/r \rfloor$ for every $i\in\{1,\ldots,r\}$. By Theorem \ref{allf}, we have $\sum\limits_{i=1}^{r} |\inn(T_i)|\le n$. 
With our hypothesis, we have $\sum\limits_{i=1}^{r} |\inn(T_i)|\ge (\lfloor n/r \rfloor +1) r>n$, and a contradiction. 
\end{proof}

\begin{prop}\label{pbound}
If in a $2r$-regular graph $G$ of order $n$ there exist $r$ completely independent spanning trees $T_1, \ldots, T_r$, then for every integer $i$, $1\le i\le r$, $$\left\lceil\frac{n-2}{r}\right\rceil \le |\inn(T_i)|\le n-\left\lceil\frac{n-2}{r}\right\rceil (r-1).$$
\label{dv2}
\end{prop}
\begin{proof}
In a spanning tree $T$ of a graph of order $n$ we recall that $\sum\limits_{v\in V(T)} d_T (v)=2n-2$. By Proposition \ref{pr+1}, we have $\sum\limits_{v\in V(T)}d_T(v)\le |\inn(T)|r+n$ and we obtain $\lceil\frac{n-2}{r}\rceil \le |\inn(T)|$. By Theorem \ref{allf}, $\sum\limits_{i=1}^{r}|\inn(T_i)|\le n$. For a fixed integer $i$, using the previous inequality, we obtain $|\inn(T_i)|\le n-\lceil\frac{n-2}{r}\rceil (r-1)$.
\end{proof}

\begin{de}
Let $G$ be a $2r$-regular graph of order $n$ for which there exist $r$ completely independent spanning trees $T_1, \ldots, T_r$. 
A {\em lost edge} is an edge of $G$ that is in none of the spanning trees $T_1, \ldots, T_r$. 
We let $E^l$ be the set of lost edges, i.e. $E^{l}=E(G)-\bigcup\limits_{1\le i\le r}E(T_i)$. 
Let also $E^{l}_{T_i}=\{uv\in E(G)| u, v \in \inn(T_i),\ uv \notin E(T_i)\}$, for $i\in \{1,\ldots ,r\}$, i.e. $E^{l}_{T_i}$ is the subset of edges of $E^{l}$ that have their two extremities in $\inn(T_i)$. 
\end{de}

\begin{prop}
If in a $2r$-regular graph $G$ of order $n$ there exist $r$ completely independent spanning trees $T_1, \ldots, T_r$, then $|E^{l}|=r$.
\end{prop}
\begin{proof}
We have $\sum\limits_{i=1}^{r} |E(T_i)|+|E^{l}|=E(G)=rn$ and $\sum\limits_{i=1}^{r} |E(T_i)|=r(n-1)$. Hence, $|E^{l}|=r.$
\end{proof}

Since each edge of $E^l_{T_i}$ is also in $E^l$ and each edge of $E^l$ is in at most one set $E^l_{T_i}$ for some integer $i$, we have the following observation.

\begin{obs}
In a $2r$-regular graph $G$ of order $n$ for which there exist $r$ completely independent spanning trees $T_1, \ldots, T_r$, we have $\sum\limits_{1\le i\le r}|E^{l}_{T_i}|\le |E^{l}|= r$.
\label{enadj}
\end{obs}

\begin{de}
The {\em potential extra degree} of a spanning tree $T$ in a $2r$-regular graph $G$ of order $n$ is $\text{ped}(T)=|\inn(T)| r-n+2$.
\end{de}
With Proposition~\ref{pbound}, we have the following easy observation:
\begin{obs}
Let $G$ be a graph, for which there exist $r$ completely independent spanning trees $T_1, \ldots, T_r$. Then, for every $i$, $0\le i \le r$, $\text{ped}(T_i)\ge 0$.
\end{obs}
Note also that, by definition, the number of inner vertices of $T_i$ of degree at most $r$ is bounded by $\text{ped}(T_i)$.
\begin{prop}
If in a $2r$-regular graph $G$ of order $n$ there exist $r$ completely independent spanning trees, then there exists a spanning tree $T$ among them such that $\text{ped}(T)\le 2$ and $E^{l}_{T}\le 1$, with strict inequalities if $r$ does not divide $n$.
\end{prop}
\begin{proof}
By Proposition \ref{tinf}, there exists a tree $T$ among them such that $|\inn(T)|\le \lfloor n/r \rfloor$.
Hence, $\text{ped}(T)\le \lfloor n/r \rfloor r-n+2\le 2$, with strict inequality if $r$ does not divide $n$.
For every edge $uv$ in $E_{T}^{l}$, both $u$ and $v$ are adjacent to one inner vertex of every spanning tree other than $T$. Hence, both $u$ and $v$ have degree at most $r$ in $T$ and thus $\text{ped}(T)\ge2|E_{T}^{l}|$.
\end{proof}

Note that the inequality $\text{ped}(T)\ge2|E_{T}^{l}|$ can be strict.

\begin{cor}
Suppose that $G$ is a $2r$-regular graph of order $n$ for which there exist $r$ completely independent spanning trees $T_1, \ldots, T_r$, for $r\ge 3$ and $n\equiv 0 \pmod r$. 
Then, for every integer $i$, $1\le i \le r$, $|\inn(T_i)|=n/r$ and $\text{ped}(T_i)=2$.
\end{cor}
\begin{obs}\label{degretree}
For a $2r$-regular graph $G$ of order $n$ for which there exist $r$ completely independent spanning trees $T_1, \ldots, T_r$, for every tree $T_i$, $1\le i\le r$, and every edge $e$ in $E^{l}_{T_i}$, the extremities of $e$ have degree at most $r$ in $T_i$.
\end{obs}


\section{Cartesian product of a complete graph and a cycle}
\label{sec:cartprod}
Let $m\ge 3$ and $n\ge 2$ be integers.
In this section, the considered graphs are $K_{m}\square P_n$, and $K_{m}\square C_n$ $n\ge 3$.

Let $V(K_{m}\square P_n)=V(K_{m}\square C_n)=\{u_i^j, \ 0\le i\le m-1, 0\le j\le n-1\}$ and $E(K_{m}\square P_n)=\{u_i^ju_k^j,\ 0\le i, k\le m-1, i\ne k, 0\le j\le n-1\}\cup\{u_i^ju_{i}^{j+1},\ 0\le i\le m-1, 0\le j\le n-2\})$. $E(K_{m}\square C_n)=E(K_{m}\square P_n)\cup\{u_i^0u_i^{n-1},\ 0\le i\le m-1\}$.

For $j\in\{0,\ldots,n-1\}$, the subgraphs $K^j$ induced by $\{u_i^j,\ 0\le i\le m-1$ are thus complete graphs on $n$ vertices that we call {\em $K$-copies}. In order to study the distribution of inner vertices of the spanning trees among the $K$-copies, we let $V_j(T)=\inn(T)\cap V(K^j)$ and $n_j(T)=|V_j(T)|$ for any spanning tree $T$ of $K_{m}\square C_n$.

In the remaining, the subscript of $u_i^j$ is considered modulo $m$ and its superscript and the subscripts of $V_j(T)$ and $n_j(T)$ are considered modulo $n$.

%
%

\begin{prop}
Let $n$ and $r$ be integers, $n\ge 2$, $r\ge 2$. There exist $r$ completely independent spanning trees in $K_{2r}\square P_n$.
\end{prop}
\begin{proof}
We construct $r$ completely independent spanning trees $T_1$, $\ldots$, $T_r$ as follows:
$E(T_i)=\{u^{j}_{i-1} u^{j+1}_{i-1}, u^{j}_{r+i-1} u^{j+1}_{r+i-1}| j \in\{ 0,\ldots, n-2\}\} \cup \{u^{0}_{i-1} u^{0}_{r+i-1} \} \cup$\newline
$\{u^{j}_{i-1} u^{j}_{i+k},u^{j}_{r+i-1} u^{j}_{r+i+k} |k\in\{0,\ldots, r-2\} \ ,j\in\{0,\ldots,n-1\} \}$.
\end{proof}
\begin{cor}
\label{existr-1}
Let $n$ and $r$ be integers, $n\ge 3$, $r\ge 2$. There exist $r$ completely independent spanning trees in $K_{2r}\square C_n$.
\end{cor}
In the three next propositions, we will prove that there do not exist $r$ completely independent spanning trees in $K_{2r-1}\square C_n$, for some integers $r$ and $n$.
Let $p=|V(K_{2r-1}\square C_n)|=n (2r-1)$ and assume that there exist $r$ completely independent spanning trees $T_1,\ldots, T_r$ in $K_{2r-1}\square C_n$.
Let $T$ be the spanning tree among them which minimizes $|\inn(T)|$, i.e. $\text{ped}(T)$. By Proposition \ref{tinf}, $T$ is such that $|\inn(T)|\le \lfloor p/r \rfloor=2n-\lceil n/r \rceil$, $\text{ped}(T)\le 2nr-\lceil n/r\rceil r-p+2\le n-\lceil n/r\rceil r+2 \le 2$ and $|E^{l}_{T}|\le 1$.
In order to establish this property we will consider all possible distributions of inner vertices of $T$ among the different $K$-copies and prove that for each of them we have a contradiction. 

The properties given in the following lemma will be useful.
\begin{lemma} \label{le1}
Let $a_i(T)$ be the number of $K$-copies which contains exactly $i$ inner vertices of $T$.
The distribution of inner vertices among the different $K$-copies is such that:
\begin{enumerate}
 \item[i)] if $n_j(T)\ge k$, for some integer $j$, then $|E^{l}_{T}|\ge \frac{1}{2} (k-1)(k-2)$;
 \item[ii)] $n_j(T)< 4$, for every integer $j$;
 \item[iii)] $a_3(T)\le 1$;
 \item[iv)] if $a_3(T)=1$, then $n \equiv 0\pmod{r}$ and $n\ge r$;
 \item[v)] if $a_0(T)=0$, then $a_3(T)\le a_1(T)-\lceil n/r \rceil$; in particular $a_1(T)>a_3(T)$ and $a_1(T)\ge\lceil n/r \rceil$.
\end{enumerate}

\end{lemma}
\begin{proof}

i) : A complete graph of order $k$ contains $\frac{1}{2} k(k-1)$ edges and only $k-1$ edges are in $E(T)$. Thus we have $|E^{l}_{T}|\ge \frac{1}{2} (k-1)(k-2)$.\newline
ii) and iii) : If $n_j(T)\ge 4$ for some $j$ or $a_3(T)> 1$, then by i), we have $|E^{l}_{T}|\ge 2$. Hence, a contradiction.\newline
iv) : As $\text{ped}(T)\le n-\lceil n/r \rceil r+2$, we have $|E^{l}_{T}|<1$ in the case $n \not\equiv 0\pmod{r}$. As $n>0$, we have $n\ge r$.\newline
v) : By ii), we have $|\inn(T)|=a_1(T)+2a_2(T)+3a_3(T)$ and $a_2=n-a_1(T)-a_3(T)$. Hence $|\inn(T)|=a_1(T)+2(n-a_1(T)-a_3(T))+3a_3(T)\le 2n-\lceil n/r \rceil$ by the choice of $T$. Thus, $a_3(T)\le a_1(T)-\lceil n/r \rceil$ and consequently $a_1(T)>a_3$ and $a_1(T)\ge \lceil n/r \rceil$.
\end{proof}

We recall the following observation used in \cite{FE2011}.
\begin{obs}[\cite{FE2011}]\label{FE20}
If in a graph $G$ there exist $r$ completely independent spanning trees $T_1,\ldots, T_r$, then for every integer $i$, $1\le i \le r$, every vertex is adjacent to an inner vertex of $T_i$.
\end{obs}

\begin{prop}
Let $n,r$ be integers, with $n\ge 3$ and $r\ge 6$. There do not exist $r$ completely independent spanning trees in $K_{2r-1}\square C_n$.
\label{krgrand}
\end{prop}
\begin{proof} The proof is by contradiction, using Properties i)-v) of Lemma~\ref{le1}.
Suppose that there exist $r$ completely independent spanning trees in $K_{2r-1}\square C_n$ and let $T$ be the tree from Proposition \ref{tinf}.
If a $K$-copy $K^i$, $1\le i\le n$, contains no inner vertex, then, by Observation \ref{FE20}, $n_{i-1}(T)+n_{i+1}(T)\ge 2r-1 \ge 11$. Consequently, we have $n_{i-1}(T)\ge 6$ or $n_{i+1}(T)\ge 6$, contradicting Property ii). Hence $a_0(T)=0$.

By Property v), $a_1(T)\ge \lceil n/r \rceil\ge 1$. Hence there exists an integer $i$, $0\le i\le n-1$, such that $n_i=1$.
Let $u$ be the (unique) vertex of $V_i(T)$. The vertex $u$ has degree at most $r+1$ in $T$ and is adjacent in $T$ to a vertex of $V_{i-1}(T)\cup V_{i+1}(T)$. Then, $u$ is adjacent in $T$ to at most $r$ vertices of $V(K^{i})$. Thus, at least $r-2\ge 4$ vertices are not adjacent in $T$ to $u$. Hence, these $r-2$ vertices are adjacent in $T$ to vertices of $V_{i-1}(T)\cup V_{i+1}(T)$ and consequently $n_{i-1}(T)+n_{i+1}(T)\ge 5$. Therefore, we have $n_{i-1}(T)\ge 3$ or $n_{i+1}(T)\ge 3$. 

Assume, without loss of generality, that $n_{i+1}(T)\ge 3$. By Property ii), $n_{i+1}(T)= 3$ and by Property iii), $a_3(T)=1$, i.e., $n_j(T)<3$ for any $j\ne i$.
But, by Property iv), $n\ge r$ and by Property v), $a_1\ge 2$. Let $j$ be such that $n_j(T)=1$, with $j\ne i$. Using a similar argument than above, we obtain that $n_{j-1}(T)\ge 3$ or $n_{j+1}(T)\ge 3$. But, as $a_3(T)=1$, the only possibility is to have $j=i+2$, i.e. both $K$-copies with one internal vertices are adjacent to the same $K$-copy with three internal vertices. 

Let $v$ be the (unique) vertex of $V_{j} (T)$.
One vertex among $u$ and $v$ is adjacent in $T$ to two inner vertices (if not $T$ would be not connected). 
Suppose, without loss of generality, that $u$ is adjacent in $T$ to two inner vertices. Then $u$ is adjacent in $T$ to at most $r-1$ vertices in $V(K^{i})$. Thus, at least $r-1\ge 5$ vertices are not adjacent in $T$ to $u$. Therefore, at least $5$ vertices are adjacent in $T$ to vertices of $V_{i-1}(T)\cup V_{i+1}(T)$ and consequently $n_{i-1}(T)+n_{i+1}(T)\ge 7$.
Hence, we have $n_{i-1}(T)\ge 4$ or $n_{i+1}(T)\ge 4$, contradicting Property ii).
\end{proof}

\begin{figure}[t]
\begin{center}
\begin{tikzpicture}[scale=0.7]
\node at (0,0) [circle,draw=black!50,fill=black!50,scale=0.5]{};
\node at (1,0) [circle,draw=black!50,fill=black!50,scale=0.5]{};
\node at (0,0.8) [circle,draw=black!50,fill=black!50,scale=0.5]{};
\node at (1,0.8) [circle,draw=black!50,fill=black!50,scale=0.5]{};
\node at (0,1.6) [circle,draw=black!50,fill=black!50,scale=0.5]{};
\node at (1,1.6) [circle,draw=black!50,fill=black!50,scale=0.5]{};
\node at (0.5,2.4) [circle,draw=black!50,fill=black!50,scale=0.5]{};
\node at (0+3,0) [circle,draw=black!50,fill=black!50,scale=0.5]{};
\node at (1+3,0) [circle,draw=black!50,fill=black!50,scale=0.5]{};
\node at (0+3,0.8) [circle,draw=black!50,fill=black!50,scale=0.5]{};;
\node at (1+3,0.8) [circle,draw=black!50,fill=black!50,scale=0.5]{};
\node at (0+3,1.6) [circle,draw=black!50,fill=black!50,scale=0.5]{};
\node at (1+3,1.6) [circle,draw=black!50,fill=black!50,scale=0.5]{};
\node at (0.5+3,2.4) [circle,draw=black!50,fill=black!50,scale=0.5]{};
\node at (0+6,0) [circle,draw=black!50,fill=black!50,scale=0.5]{};
\node at (1+6,0) [circle,draw=black!50,fill=black!50,scale=0.5]{};
\node at (0+6,0.8) [circle,draw=black!50,fill=black!50,scale=0.5]{};
\node at (1+6,0.8) [circle,draw=black!50,fill=black!50,scale=0.5]{};
\node at (0+6,1.6) [circle,draw=black!50,fill=black!50,scale=0.5]{};
\node at (1+6,1.6) [circle,draw=black!50,fill=black!50,scale=0.5]{};
\node at (0.5+6,2.4) [circle,draw=black!50,fill=black!50,scale=0.5]{};
\node at (0+9,0) [circle,draw=black!50,fill=black!50,scale=0.5]{};
\node at (1+9,0) [circle,draw=black!50,fill=black!50,scale=0.5]{};
\node at (0+9,0.8) [circle,draw=black!50,fill=black!50,scale=0.5]{};
\node at (1+9,0.8) [circle,draw=black!50,fill=black!50,scale=0.5]{};
\node at (0+9,1.6) [circle,draw=black!50,fill=black!50,scale=0.5]{};
\node at (1+9,1.6) [circle,draw=black!50,fill=black!50,scale=0.5]{};
\node at (0.5+9,2.4) [circle,draw=black!50,fill=black!50,scale=0.5]{};
\node at (0+12,0) [circle,draw=black!50,fill=black!50,scale=0.5]{};
\node at (1+12,0) [circle,draw=black!50,fill=black!50,scale=0.5]{};
\node at (0+12,0.8) [circle,draw=black!50,fill=black!50,scale=0.5]{};
\node at (1+12,0.8) [circle,draw=black!50,fill=black!50,scale=0.5]{};
\node at (0+12,1.6) [circle,draw=black!50,fill=black!50,scale=0.5]{};
\node at (1+12,1.6) [circle,draw=black!50,fill=black!50,scale=0.5]{};
\node at (0.5+12,2.4) [circle,draw=black!50,fill=black!50,scale=0.5]{};
\node at (0.5,-0.7){$K^{i-1}$}; 
\node at (0.5+3,-0.7){$K^{i}$}; 
\node at (0.5+6,-0.7){$K^{i+1}$}; 
\node at (0.5+9,-0.7){$K^{j}$}; 
\node at (0.5+12,-0.7){$K^{j+1}$}; 
\node at (0.5+6,2.4) [regular polygon, regular polygon sides=4,draw=black,fill=blue,scale=0.5]{};
\node at (0+6,1.6) [regular polygon, regular polygon sides=4,draw=black,fill=blue,scale=0.5]{};
\node at (1+6,1.6) [regular polygon, regular polygon sides=4,draw=black,fill=blue,scale=0.5]{};
\node at (0+3,1.6) [regular polygon, regular polygon sides=4,draw=black,fill=blue,scale=0.5]{};
\node at (1+9,1.6) [regular polygon, regular polygon sides=4,draw=black,fill=blue,scale=0.5]{};
\node at (0,1.6) [regular polygon, regular polygon sides=4,draw=black,fill=blue,scale=0.5]{};
\node at (1+12,1.6) [regular polygon, regular polygon sides=4,draw=black,fill=blue,scale=0.5]{};
\node at (0,0.8) [regular polygon, regular polygon sides=4,draw=black,fill=blue,scale=0.5]{};
\node at (1+12,0.8) [regular polygon, regular polygon sides=4,draw=black,fill=blue,scale=0.5]{};
\draw[ultra thick,color=blue] (0.5+6,2.4) -- (0+6,1.6);
\draw[ultra thick,color=blue] (0.5+6,2.4) -- (1+6,1.6);
\draw[dashed, ultra thick,color=blue] (0+6,1.6) -- (1+6,1.6);
\draw[ultra thick,color=blue] (0,1.6) -- (0,0.8);
\draw[ultra thick,color=blue] (1+12,1.6) -- (1+12,0.8);
\draw[ultra thick,color=blue] (0,1.6) .. controls (1.5,2) .. (0+3,1.6);
\draw[ultra thick,color=blue] (0+3,1.6) .. controls (1.5+3,2) .. (0+6,1.6);
\draw[ultra thick,color=blue] (1+6,1.6) .. controls (2.5+6,2) .. (1+9,1.6);
\draw[ultra thick,color=blue] (1+9,1.6) .. controls (2.5+9,2) .. (1+12,1.6);

\draw[ultra thick,color=blue] (0,0.8) .. controls (1.5,1.2) .. (0+3,0.8);
\draw[ultra thick,color=blue] (1+9,0.8) .. controls (2.5+9,1.2) .. (1+12,0.8);
\draw[ultra thick,color=blue] (0+3,1.6) -- (1+3,0.8);
\draw[ultra thick,color=blue] (0+3,1.6) -- (1+3,0);
\draw[ultra thick,color=blue] (1+9,1.6) -- (0+9,0.8);
\draw[ultra thick,color=blue] (1+9,1.6) -- (0+9,0);
\draw[ultra thick,color=blue] (0+3,1.6)  .. controls (-0.2+3,0.8) ..  (0+3,0);
\draw[ultra thick,color=blue] (1+9,1.6)  .. controls (1.2+9,0.8) ..  (1+9,0);

\end{tikzpicture}
\end{center}
\caption{A configuration of inner vertices in the proof of Proposition \ref{propk45}. Boxes are inner vertices and the dashed edge represents a lost edge.}
\label{illpreuve}
\end{figure}
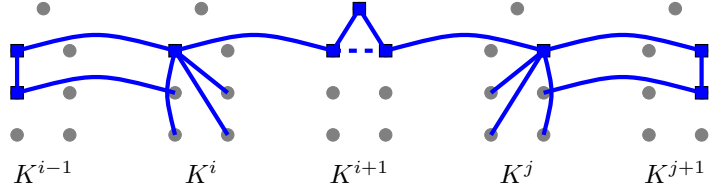
\begin{prop}
\label{propk45}
Let $n,r$ be integers, with $4\le r \le 5$ and $n\ge r+1$. There do not exist $r$ completely independent spanning trees in $K_{2r-1}\square C_n$.
\end{prop}
\begin{proof} The proof is by contradiction, using Properties i)-v) of Lemma~\ref{le1}.
Suppose that there exist $r$ completely independent spanning trees in $K_{2r-1}\square C_n$ and let $T$ be the tree from Proposition \ref{tinf}.
If a $K$-copy $K^i$, $0\le i\le n-1$, contains no inner vertex, then $n_{i-1}(T)+n_{i+1}(T)\ge 7$. Consequently, we have $n_{i-1}(T)\ge 4$ or $n_{i+1}(T)\ge 4$, contradicting Property ii). Hence $a_0(T)=0$.
By Property v), $a_1(T)\ge\lceil n/r \rceil\ge 2$. Thus, there exist two integers $i$ and $j$, $0\le i \le j \le n-1$, such that $n_i(T)= n_j(T) = 1$, with $u\in V_i(T)$ and $v\in V_j(T)$.

First, suppose that $i=j-1$. Each of $u$ and $v$ has degree at most $r+1$ in $T$ and $u$ ($v$, respectively) is adjacent in $T$ to a vertex of $V_{i-1}(T)\cup V_{i+1}(T)$ (of $V_{j-1}(T)\cup V_{j+1}(T)$, respectively).

If $u$ and $v$ are adjacent in $T$, then one vertex among $u$ and $v$ is adjacent in $T$ to a vertex of $V_{i-1}(T)\cup V_{j+1}(T)$ (if not $T$ would be not connected). Suppose, without loss of generality, that $u$ is adjacent to two inner vertices. Then, at least $r-1\ge 3$ vertices of $V(K^i)$ are not adjacent in $T$ to $u$.
Consequently, $n_{i-1}(T)\ge 4$ and we have a contradiction with Property ii).

Else if $u$ and $v$ are not adjacent in $T$, then both $u$ and $v$ are adjacent in $T$ to vertices of $V_{i-1}(T)\cup V_{j+1}(T)$ (if not, $T$ would be not connected). The vertices $u$ and $v$ are each adjacent in $T$ to at most $r$ vertices in $V(K^i)\cup V(K^j)$. 
Hence, there remain at least $4r-2-2r-2=2r-4\ge 4$ vertices in $V(K^i)\cup V(K^j)$ that must be adjacent in $T$ to vertices of $V_{i-1}(T)\cup V_{j+1}(T)$ other than the neighbors of $u$ and of $v$.
Consequently $n_{i-1}(T)+n_{j+1}(T)\ge 6$. Hence, we have $n_{i-1}(T)\ge 3$ and $n_{j+1}(T)\ge 3$, contradicting Property iii) or $n_{i-1}(T)\ge 4$ or $n_{j+1}(T)\ge 4$, contradicting Property ii).

Second, if $|i- j|>1$, then one vertex among $u$ and $v$ is adjacent in $T$ to two inner vertices (if not $T$ would be not connected). Suppose, without loss of generality, that $u$ is adjacent to two inner vertices. At least $r-1$ vertices of $V(K^i)$ are not adjacent in $T$ to $u$. Hence, if $r=5$, we have $n_{i-1}(T)\ge 3$ and $n_{i+1}(T)\ge 3$, contradicting Property iii) or $n_{i-1}(T)\ge 4$ or $n_{i+1}(T)\ge 4$, contradicting Property ii). Consequently, we suppose that $r=4$. Then, at least $r-1\ge 3$ vertices of $V(K^i)$ are not adjacent in $T$ to $u$. Therefore, we have $n_{i-1}(T)\ge 3$ or $n_{i+1}(T)\ge 3$.

Assume, without loss of generality, that $n_{i+1}(T)\ge 3$. By Property ii), $n_{i+1}(T)= 3$ and by Property iii), $a_3(T)=1$, i.e., $n_j(T)<3$ for any $j\ne i$. 
But, as $n> r$ and by Property v), $a_1\ge 3$. Let $i'$ be such that $n_{i'}(T)=1$, with $i' \ne i$ and $i' \ne i$. If $|i'-i|= 1$ or $|i'-j|= 1$, we have a contradiction, using the first point. Two vertices among $u$, $v$ and $u'$ should be adjacent to two inner vertices. Suppose it is the vertices $u$ and $v$. Using a similar argument than above, we obtain that $n_{j-1}(T)\ge 3$ or $n_{j+1}(T)\ge 3$. But, as $a_3(T)=1$, the only possibility is to have $j=i+2$, i.e. both $K$-copies with one internal vertices are adjacent to the same $K$-copy with three internal vertices.

In this case, as $r=4$, then four vertices are not inner vertices in $V(K^{i+1})$, at least three vertices of $V(K^i)$ are not adjacent in $T$ to $u$ and at least three vertices of $V(K^j)$ are not adjacent in $T$ to $v$.
Moreover, we have $n_{i-1}(T)\le 2$ and $n_{j+1}(T)\le 2$.
Figure \ref{illpreuve} illustrates this configuration.
Thus, four vertices of $V(K^{i+1})$ are adjacent in $T$ to vertices of $V_{i+1}(T)$ and four vertices of $V(K^i)\cup V(K^j)$ are adjacent in $T$ to vertices of $V_{i+1}(T)$.
However, by Observation \ref{degretree}, the vertices of $V_{i+1}(T)$ can be adjacent to at most seven leaves in $T$.
Hence, we have a contradiction.

\end{proof}

\begin{prop}
\label{propk93}
There do not exist five completely independent spanning trees in $K_{9}\square C_3$.
\end{prop}
\begin{proof}
Suppose that there exist five completely independent spanning trees in $K_{9}\square C_3$ and let $T$ be the tree from Proposition \ref{tinf}.
We recall that $|V(K_{9}\square C_3)|=27$ and $|\inn(T)|\le 6-\lceil 3/4 \rceil= 5$.
If a $K$-copy $K^i$, $0\le i\le n-1$, contains no inner vertex, then $n_{i-1}(T)\ge 5$ or $n_{i+1}(T)\ge 5$. Thus, we have a contradiction with Property ii). By property iv), as $n\not \equiv 0 \pmod{r}$, we have $a_3 (T)=0$. Thus, the only possible distribution of inner vertices of $T$ is $a_1(T)=1$ and $a_2(T)=2$.
Without loss of generality, suppose that $n_0(T)=1$, $n_1(T)=2$ and $n_2(T)=2$, with $u\in V_1(T)$.

Let the position of a vertex $u_i^j$ be $i$.
As $T$ should be connected, two pairs of inner vertices in different $K$-copies should be adjacent in $T$ among these five inner vertices.
Thus, these five vertices have only three different positions.
The vertex $u$ has degree at most $6$ in $T$.
Hence, there are $r-2\ge 3$ vertices of $V(K^1)$ not adjacent in $T$ to $u$.
As the inner vertices have only two positions different from the position of $u$, it is impossible that every vertex is adjacent in $T$ to an inner vertex of $T$.
\end{proof}

We now show positive results for the remaining values of $r$ and $n$. Some of the spanning trees were found using a computer to solve an ILP formulation of the problem.

\begin{prop}
Let $n\ge 3$ be an integer such that $n\equiv0\pmod{3}$. There exist three completely independent spanning trees in $K_{5}\square C_n$.
\label{propk53}
\end{prop}
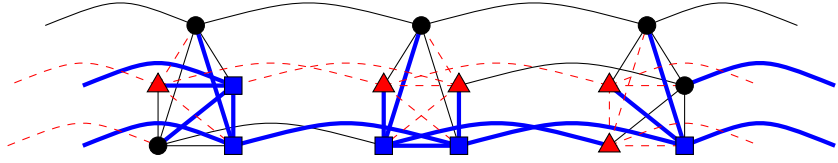
\begin{figure}[t]
\begin{center}
\begin{tikzpicture}
\draw (0.5,1.6) -- (0,0);
\draw (6.5,1.6) -- (7,0.8);
\draw (0,0) -- (1,0);
\draw (0.5,1.6) -- (1,0.8);
\draw (0,0) -- (0,0.8);
\draw (3.5,1.6) -- (3,0.8);
\draw (3.5,1.6) -- (4,0);
\draw (6.5,1.6) -- (6,0.8);
\draw (7,0.8) -- (6,0);
\draw (7,0.8) -- (7,0);
\draw (0.5,1.6) .. controls (2,2) .. (3.5,1.6);
\draw (0,0) .. controls (1.5,0.4) .. (3,0);
\draw (3.5,1.6) .. controls (5,2) .. (6.5,1.6);
\draw (6.5,1.6) .. controls (7.5,2) .. (8.5,1.6);
\draw (4,0.8) .. controls (5.5,1.2) .. (7,0.8);
\draw (0.5,1.6) .. controls (-0.5,2) .. (-1.5,1.6);
\draw[ultra thick,color=blue] (1,0) -- (1,0.8);
\draw[ultra thick,color=blue] (4,0) -- (3,0);
\draw[ultra thick,color=blue] (3,0) -- (3.5,1.6);
\draw[ultra thick,color=blue] (4,0) -- (4,0.8);
\draw[ultra thick,color=blue] (3,0) -- (3,0.8);
\draw[ultra thick,color=blue] (0,0) -- (1,0.8);
\draw[ultra thick,color=blue] (0.5,1.6) -- (1,0);
\draw[ultra thick,color=blue] (0,0.8) -- (1,0.8);
\draw[ultra thick,color=blue] (7,0) -- (6,0.8);
\draw[ultra thick,color=blue] (7,0) -- (6.5,1.6);
\draw[ultra thick,color=blue] (1,0) .. controls (0,0.4) .. (-1,0);
\draw[ultra thick,color=blue] (1,0) .. controls (2.5,0.4) .. (4,0);
\draw[ultra thick,color=blue] (4,0) .. controls (5.5,0.4) .. (7,0);
\draw[ultra thick,color=blue] (3,0) .. controls (4.5,0.4) .. (6,0);
\draw[ultra thick,color=blue] (7,0) .. controls (8,0.4) .. (9,0);
\draw[ultra thick,color=blue] (1,0.8) .. controls (0,1.2) .. (-1,0.8);
\draw[ultra thick,color=blue] (7,0.8) .. controls (8,1.2) .. (9,0.8);
\draw[dashed, draw=red] (0,0.8) -- (0.5,1.6);
\draw[dashed, draw=red] (0,0.8) -- (1,0);
\draw[dashed, draw=red] (3,0.8) -- (4,0.8);
\draw[dashed, draw=red] (4,0.8) -- (3,0);
\draw[dashed, draw=red] (3,0.8) -- (4,0);
\draw[dashed, draw=red] (3.5,1.6) -- (4,0.8);
\draw[dashed, draw=red] (6,0.8) -- (6,0);
\draw[dashed, draw=red] (6,0.8) -- (7,0.8);
\draw[dashed, draw=red] (6,0) -- (6.5,1.6);
\draw[dashed, draw=red] (6,0) -- (7,0);
\draw[dashed, draw=red] (0,0.8) .. controls (-1,1.2) .. (-2,0.8);
\draw[dashed, draw=red] (0,0.8) .. controls (1.5,1.2) .. (3,0.8);
\draw[dashed, draw=red] (1,0.8) .. controls (2.5,1.2) .. (4,0.8);
\draw[dashed, draw=red] (3,0.8) .. controls (4.5,1.2) .. (6,0.8);
\draw[dashed, draw=red] (6,0.8) .. controls (7,1.2) .. (8,0.8);
\draw[dashed, draw=red] (6,0) .. controls (7,0.4) .. (8,0);
\draw[dashed, draw=red] (-2,0) .. controls (-1,0.4) .. (0,0);
\node at (1,0) [regular polygon, regular polygon sides=4,draw=black,fill=blue,scale=0.7]{};
\node at (1,0.8) [regular polygon, regular polygon sides=4,draw=black,fill=blue,scale=0.7]{};
\node at (0,0) [circle,draw=black,fill=black,scale=0.7]{};
\node at (0,0.8) [regular polygon, regular polygon sides=3,draw=black,fill=red,scale=0.5]{};
\node at (0.5,1.6) [circle,draw=black,fill=black,scale=0.7]{};
\node at (0+3,0) [regular polygon, regular polygon sides=4,draw=black,fill=blue,scale=0.7]{};
\node at (1+3,0) [regular polygon, regular polygon sides=4,draw=black,fill=blue,scale=0.7]{};
\node at (0+3,0.8) [regular polygon, regular polygon sides=3,draw=black,fill=red,scale=0.5]{};
\node at (1+3,0.8) [regular polygon, regular polygon sides=3,draw=black,fill=red,scale=0.5]{};
\node at (0.5+3,1.6) [circle,draw=black,fill=black,scale=0.7]{};
\node at (0+6,0) [regular polygon, regular polygon sides=3,draw=black,fill=red,scale=0.5]{};
\node at (1+6,0) [regular polygon, regular polygon sides=4,draw=black,fill=blue,scale=0.7]{};
\node at (0+6,0.8) [regular polygon, regular polygon sides=3,draw=black,fill=red,scale=0.5]{};
\node at (1+6,0.8) [circle,draw=black,fill=black,scale=0.7]{};
\node at (0.5+6,1.6) [circle,draw=black,fill=black,scale=0.7]{};
\end{tikzpicture}
\end{center}
\caption{A pattern to have three completely independent spanning trees in $K_{5}\square C_n$, for $n\equiv 0\pmod 3$.}
\label{K53}
\end{figure}
\begin{proof}
We construct three completely independent spanning trees $T_1$, $T_2$ and $T_3$ using repeatedly the pattern illustrated in Figure~\ref{K53} on each three consecutive $K$-copies:\newline
$E(T_1)=\{u^{3j}_{0} u^{1+3j}_{0},u^{1+3j}_{0} u^{2+3j}_{0}, u^{2+3j}_{0} u^{3+3j}_{0},
u^{3j}_{0} u^{3j}_{2}, u^{3j}_{0}  u^{3j}_{3},$\newline
$u^{3j}_{3}  u^{3j}_{1}, u^{3j}_{3} u^{3j}_{4},
u^{3j}_{3} u^{1+3j}_{3}, u^{1+3j}_{0} u^{1+3j}_{1}, u^{1+3j}_{0} u^{1+3j}_{4},  u^{1+3j}_{2} u^{2+3j}_{2},$\newline
$u^{2+3j}_{0} u^{2+3j}_{2}, u^{2+3j}_{0} u^{2+3j}_{1}, u^{2+3j}_{2} u^{2+3j}_{3}, u^{2+3j}_{2} u^{2+3j}_{4}|j\in\{0,\ldots,n/3-1\}\}-\{u^0_0,u^1_0\}$; \newline
$E(T_2)=\{u^{3j}_{1} u^{1+3j}_{1},u^{1+3j}_{1} u^{2+3j}_{1}, u^{2+3j}_{1} u^{3+3j}_{1},
u^{3j}_{1} u^{3j}_{0}, u^{3j}_{1}  u^{3j}_{4},$\newline
$ u^{3j}_{2} u^{1+3j}_{2}, u^{1+3j}_{1} u^{1+3j}_{2},
 u^{1+3j}_{1} u^{1+3j}_{4}, u^{1+3j}_{2} u^{1+3j}_{0},  u^{1+3j}_{2} u^{1+3j}_{3},u^{2+3j}_{1} u^{2+3j}_{3},$\newline
$ u^{2+3j}_{1} u^{2+3j}_{2}, u^{2+3j}_{3} u^{2+3j}_{0}, u^{2+3j}_{3} u^{2+3j}_{4}, u^{2+3j}_{3} u^{3+3j}_{3}|j\in\{0,\ldots,n/3-1\}\}-\{u^0_1,u^1_1\}$;\newline
$E(T_3)=\{u^{3j}_{4} u^{1+3j}_{4},u^{1+3j}_{4} u^{2+3j}_{4}, u^{2+3j}_{4} u^{3+3j}_{4},
u^{3j}_{2} u^{3j}_{4}, u^{3j}_{2}  u^{3j}_{1},$\newline
$u^{3j}_{2}  u^{3j}_{3}, u^{3j}_{4} u^{3j}_{0},
u^{1+3j}_{3} u^{1+3j}_{4}, u^{1+3j}_{3} u^{1+3j}_{0}, u^{1+3j}_{3} u^{1+3j}_{1},  u^{1+3j}_{4} u^{1+3j}_{2},$\newline
$u^{1+3j}_{3} u^{2+3j}_{3}, u^{2+3j}_{4} u^{2+3j}_{0}, u^{2+3j}_{4} u^{2+3j}_{1}, u^{2+3j}_{2} u^{3+3j}_{2}|j\in\{0,\ldots,n/3-1\}\}-\{u^0_4,u^1_4\}$.
\end{proof}
\begin{prop}
Let $n\ge 3$ be an integer. There exist three completely independent spanning trees in $K_{5}\square C_n$.
\label{propk5n}
\end{prop}
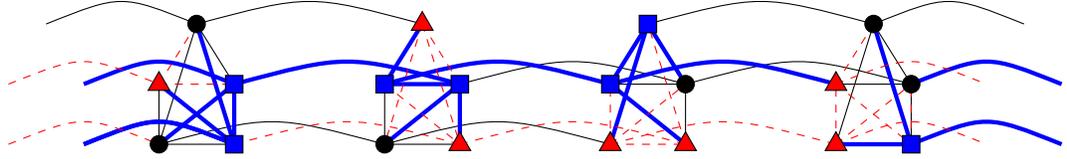
\begin{figure}[t]
\begin{center}
\begin{tikzpicture}
\draw (0.5,1.6) -- (0,0);
\draw (0.5,1.6) -- (1,0.8);
\draw (0,0) -- (1,0);
\draw (0,0) -- (0,0.8);
\draw (3,0) -- (3,0.8);
\draw (3,0) -- (4,0);
\draw (7,0.8) -- (6,0.8);
\draw (7,0) -- (7,0.8);
\draw (9.5,1.6) -- (10,0.8);
\draw (9.5,1.6) -- (9,0);
\draw (9,0.8) -- (10,0.8);
\draw (10,0) -- (10,0.8);
\draw (0.5,1.6) .. controls (2,2) .. (3.5,1.6);
\draw (0,0) .. controls (1.5,0.4) .. (3,0);
\draw (3,0) .. controls (4.5,0.4) .. (6,0);
\draw (4,0.8) .. controls (5.5,1.2) .. (7,0.8);
\draw (7,0.8) .. controls (8.5,1.2) .. (10,0.8);
\draw (6.5,1.6) .. controls (8,2) .. (9.5,1.6);
\draw (-1.5,1.6) .. controls (-0.5,2) .. (0.5,1.6);
\draw (9.5,1.6) .. controls (10.5,2) .. (11.5,1.6);
\draw[ultra thick,color=blue] (1,0) -- (1,0.8);
\draw[ultra thick,color=blue] (1,0.8) -- (0,0);
\draw[ultra thick,color=blue] (1,0) -- (0.5,1.6);
\draw[ultra thick,color=blue] (1,0) -- (0,0.8);
\draw[ultra thick,color=blue] (3,0.8) -- (4,0.8);
\draw[ultra thick,color=blue] (4,0) -- (4,0.8);
\draw[ultra thick,color=blue] (3,0.8) -- (3.5,1.6);
\draw[ultra thick,color=blue] (3,0) -- (4,0.8);
\draw[ultra thick,color=blue] (7,0) -- (6,0.8);
\draw[ultra thick,color=blue] (6.5,1.6) -- (7,0.8);
\draw[ultra thick,color=blue] (6,0) -- (6.5,1.6);
\draw[ultra thick,color=blue] (6.5,1.6) -- (6,0.8);
\draw[ultra thick,color=blue] (9.5,1.6) -- (10,0);
\draw[ultra thick,color=blue] (9,0) -- (10,0);
\draw[ultra thick,color=blue] (1,0) .. controls (0,0.4) .. (-1,0);
\draw[ultra thick,color=blue] (1,0.8) .. controls (2.5,1.2) .. (4,0.8);
\draw[ultra thick,color=blue] (3,0.8) .. controls (4.5,1.2) .. (6,0.8);
\draw[ultra thick,color=blue] (6,0.8) .. controls (7.5,1.2) .. (9,0.8);
\draw[ultra thick,color=blue] (1,0.8) .. controls (0,1.2) .. (-1,0.8);
\draw[ultra thick,color=blue] (1,0) .. controls (0,0.4) .. (-1,0);
\draw[ultra thick,color=blue] (10,0.8) .. controls (11,1.2) .. (12,0.8);
\draw[ultra thick,color=blue] (10,0) .. controls (11,0.4) .. (12,0);
\draw[dashed, draw=red] (0,0.8) -- (0.5,1.6);
\draw[dashed, draw=red] (0,0.8) -- (1,0.8);
\draw[dashed, draw=red] (3,0.8) -- (4,0);
\draw[dashed, draw=red] (3.5,1.6) -- (4,0);
\draw[dashed, draw=red] (3.5,1.6) -- (4,0.8);
\draw[dashed, draw=red] (3.5,1.6) -- (3,0);
\draw[dashed, draw=red] (7,0) -- (6.5,1.6);
\draw[dashed, draw=red] (6,0) -- (7,0);
\draw[dashed, draw=red] (6,0) -- (7,0.8);
\draw[dashed, draw=red] (6,0) -- (6,0.8);
\draw[dashed, draw=red] (9,0) -- (9,0.8);
\draw[dashed, draw=red] (9,0.8) -- (10,0);
\draw[dashed, draw=red] (9,0.8) -- (9.5,1.6);
\draw[dashed, draw=red] (9,0) -- (10,0.8);
\draw[dashed, draw=red] (10,0) -- (10,0.8);
\draw[dashed, draw=red] (1,0) .. controls (2.5,0.4) .. (4,0);
\draw[dashed, draw=red] (4,0) .. controls (5.5,0.4) .. (7,0);
\draw[dashed, draw=red] (6,0) .. controls (7.5,0.4) .. (9,0);
\draw[dashed, draw=red] (-2,0) .. controls (-1,0.4) .. (0,0);
\draw[dashed, draw=red] (-2,0.8) .. controls (-1,1.2) .. (0,0.8);
\draw[dashed, draw=red] (9,0) .. controls (10,0.4) .. (11,0);
\draw[dashed, draw=red] (9,0.8) .. controls (10,1.2) .. (11,0.8);
\node at (1,0) [regular polygon, regular polygon sides=4,draw=black,fill=blue,scale=0.7]{};
\node at (1,0.8) [regular polygon, regular polygon sides=4,draw=black,fill=blue,scale=0.7]{};
\node at (0,0) [circle,draw=black,fill=black,scale=0.7]{};
\node at (0,0.8) [regular polygon, regular polygon sides=3,draw=black,fill=red,scale=0.5]{};
\node at (0.5,1.6) [circle,draw=black,fill=black,scale=0.7]{};
\node at (0+3,0) [circle,draw=black,fill=black,scale=0.7]{};
\node at (1+3,0) [regular polygon, regular polygon sides=3,draw=black,fill=red,scale=0.5]{};
\node at (0+3,0.8) [regular polygon, regular polygon sides=4,draw=black,fill=blue,scale=0.7]{};
\node at (1+3,0.8) [regular polygon, regular polygon sides=4,draw=black,fill=blue,scale=0.7]{};
\node at (0.5+3,1.6) [regular polygon, regular polygon sides=3,draw=black,fill=red,scale=0.5]{};
\node at (0+6,0) [regular polygon, regular polygon sides=3,draw=black,fill=red,scale=0.5]{};
\node at (1+6,0) [regular polygon, regular polygon sides=3,draw=black,fill=red,scale=0.5]{};
\node at (0+6,0.8) [regular polygon, regular polygon sides=4,draw=black,fill=blue,scale=0.7]{};
\node at (1+6,0.8) [circle,draw=black,fill=black,scale=0.7]{};
\node at (0.5+6,1.6) [regular polygon, regular polygon sides=4,draw=black,fill=blue,scale=0.7]{};
\node at (0+9,0) [regular polygon, regular polygon sides=3,draw=black,fill=red,scale=0.5]{};
\node at (1+9,0) [regular polygon, regular polygon sides=4,draw=black,fill=blue,scale=0.7]{};
\node at (0+9,0.8) [regular polygon, regular polygon sides=3,draw=black,fill=red,scale=0.5]{};
\node at (1+9,0.8) [circle,draw=black,fill=black,scale=0.7]{};
\node at (0.5+9,1.6) [circle,draw=black,fill=black,scale=0.7]{};
\end{tikzpicture}
\end{center}
\caption{The three completely independent spanning trees in $K_{5}\square C_n$, for $K^0\cup K^1\cup K^2\cup K^3$ and $n\equiv 1\pmod 3$.}
\label{K54}
\end{figure}

\begin{figure}[t]
\begin{center}
\begin{tikzpicture}
\draw (0.5,1.6) -- (0,0);
\draw (0.5,1.6) -- (1,0.8);
\draw (0,0) -- (1,0);
\draw (0,0) -- (0,0.8);
\draw (2,0) -- (2.5,1.6);
\draw (3,0) -- (2.5,1.6);
\draw (3,0.8) -- (2.5,1.6);
\draw (2,0) -- (2,0.8);
\draw (4,0.8) -- (5,0.8);
\draw (4.5,1.6) -- (5,0.8);
\draw (7,0) -- (7,0.8);
\draw (7,0) -- (6,0.8);
\draw (7,0) -- (6,0);
\draw (7,0.8) -- (6.5,1.6);
\draw (9,0.8) -- (9,0);
\draw (9,0.8) -- (8.5,1.6);
\draw (8,0.8) -- (8.5,1.6);
\draw (8,0) -- (9,0.8);
\draw (0.5,1.6) .. controls (1.5,2) .. (2.5,1.6);
\draw (2,0) .. controls (3,0.4) .. (4,0);
\draw (5,0.8) .. controls (6,1.2) .. (7,0.8);
\draw (5,0) .. controls (6,0.4) .. (7,0);
\draw (7,0.8) .. controls (8,1.2) .. (9,0.8);
\draw (-1.5,1.6) .. controls (-0.5,2) .. (0.5,1.6);
\draw (8.5,1.6) .. controls (9.5,2) .. (10.5,1.6);
\draw[ultra thick,color=blue] (1,0) -- (1,0.8);
\draw[ultra thick,color=blue] (1,0.8) -- (0,0);
\draw[ultra thick,color=blue] (1,0) -- (0.5,1.6);
\draw[ultra thick,color=blue] (1,0) -- (0,0.8);
\draw[ultra thick,color=blue] (2,0.8) -- (3,0.8);
\draw[ultra thick,color=blue] (2,0.8) -- (2.5,1.6);
\draw[ultra thick,color=blue] (2,0.8) -- (3,0);
\draw[ultra thick,color=blue] (3,0.8) -- (2,0);
\draw[ultra thick,color=blue] (4,0) -- (4,0.8);
\draw[ultra thick,color=blue] (4,0.8) -- (4.5,1.6);
\draw[ultra thick,color=blue] (4,0) -- (5,0);
\draw[ultra thick,color=blue] (4,0) -- (5,0.8);
\draw[ultra thick,color=blue] (6,0.8) -- (7,0.8);
\draw[ultra thick,color=blue] (6,0.8) -- (6.5,1.6);
\draw[ultra thick,color=blue] (9,0) -- (8.5,1.6);
\draw[ultra thick,color=blue] (8,0) -- (9,0);
\draw[ultra thick,color=blue] (1,0) .. controls (0,0.4) .. (-1,0);
\draw[ultra thick,color=blue] (1,0.8) .. controls (2,1.2) .. (3,0.8);
\draw[ultra thick,color=blue] (2,0.8) .. controls (3,1.2) .. (4,0.8);
\draw[ultra thick,color=blue] (4,0.8) .. controls (5,1.2) .. (6,0.8);
\draw[ultra thick,color=blue] (4,0) .. controls (5,0.4) .. (6,0);
\draw[ultra thick,color=blue] (6,0.8) .. controls (7,1.2) .. (8,0.8);
\draw[ultra thick,color=blue] (7,0) .. controls (8,0.4) .. (9,0);
\draw[ultra thick,color=blue] (1,0.8) .. controls (0,1.2) .. (-1,0.8);
\draw[ultra thick,color=blue] (9,0.8) .. controls (10,1.2) .. (11,0.8);
\draw[ultra thick,color=blue] (9,0) .. controls (10,0.4) .. (11,0);
\draw[dashed, draw=red] (0,0.8) -- (0.5,1.6);
\draw[dashed, draw=red] (0,0.8) -- (1,0.8);
\draw[dashed, draw=red] (3,0) -- (3,0.8);
\draw[dashed, draw=red] (2,0) -- (3,0);
\draw[dashed, draw=red] (5,0) -- (4.5,1.6);
\draw[dashed, draw=red] (5,0) -- (5,0.8);
\draw[dashed, draw=red] (5,0) -- (4,0.8);
\draw[dashed, draw=red] (4,0) -- (4.5,1.6);
\draw[dashed, draw=red] (6,0) -- (6.5,1.6);
\draw[dashed, draw=red] (6,0) -- (6,0.8);
\draw[dashed, draw=red] (6,0) -- (7,0.8);
\draw[dashed, draw=red] (7,0) -- (6.5,1.6);
\draw[dashed, draw=red] (8,0) -- (8,0.8);
\draw[dashed, draw=red] (9,0) -- (8,0.8);
\draw[dashed, draw=red] (9,0.8) -- (8,0.8);
\draw[dashed, draw=red] (8,0) -- (8.5,1.6);
\draw[dashed, draw=red] (0,0.8) .. controls (1,1.2) .. (2,0.8);
\draw[dashed, draw=red] (1,0) .. controls (2,0.4) .. (3,0);
\draw[dashed, draw=red] (3,0) .. controls (4,0.4) .. (5,0);
\draw[dashed, draw=red] (2.5,1.6) .. controls (3.5,2) .. (4.5,1.6);
\draw[dashed, draw=red] (4.5,1.6) .. controls (5.5,2) .. (6.5,1.6);
\draw[dashed, draw=red] (6,0) .. controls (7,0.4) .. (8,0);
\draw[dashed, draw=red] (-2,0) .. controls (-1,0.4) .. (0,0);
\draw[dashed, draw=red] (-2,0.8) .. controls (-1,1.2) .. (0,0.8);
\draw[dashed, draw=red] (8,0) .. controls (9,0.4) .. (10,0);
\draw[dashed, draw=red] (8,0.8) .. controls (9,1.2) .. (10,0.8);
\node at (1,0) [regular polygon, regular polygon sides=4,draw=black,fill=blue,scale=0.7]{};
\node at (1,0.8) [regular polygon, regular polygon sides=4,draw=black,fill=blue,scale=0.7]{};
\node at (0,0) [circle,draw=black,fill=black,scale=0.7]{};
\node at (0,0.8) [regular polygon, regular polygon sides=3,draw=black,fill=red,scale=0.5]{};
\node at (0.5,1.6) [circle,draw=black,fill=black,scale=0.7]{};
\node at (0+2,0) [circle,draw=black,fill=black,scale=0.7]{};
\node at (1+2,0) [regular polygon, regular polygon sides=3,draw=black,fill=red,scale=0.5]{};
\node at (0+2,0.8) [regular polygon, regular polygon sides=4,draw=black,fill=blue,scale=0.7]{};
\node at (1+2,0.8) [regular polygon, regular polygon sides=4,draw=black,fill=blue,scale=0.7]{};
\node at (0.5+2,1.6) [circle,draw=black,fill=black,scale=0.7]{};
\node at (0+4,0) [regular polygon, regular polygon sides=4,draw=black,fill=blue,scale=0.7]{};
\node at (1+4,0) [regular polygon, regular polygon sides=3,draw=black,fill=red,scale=0.5]{};
\node at (0+4,0.8) [regular polygon, regular polygon sides=4,draw=black,fill=blue,scale=0.7]{};
\node at (1+4,0.8) [circle,draw=black,fill=black,scale=0.7]{};
\node at (0.5+4,1.6) [regular polygon, regular polygon sides=3,draw=black,fill=red,scale=0.5]{};
\node at (0+6,0) [regular polygon, regular polygon sides=3,draw=black,fill=red,scale=0.5]{};
\node at (1+6,0) [circle,draw=black,fill=black,scale=0.7]{};
\node at (0+6,0.8) [regular polygon, regular polygon sides=4,draw=black,fill=blue,scale=0.7]{};
\node at (1+6,0.8) [circle,draw=black,fill=black,scale=0.7]{};
\node at (0.5+6,1.6) [regular polygon, regular polygon sides=3,draw=black,fill=red,scale=0.5]{};
\node at (0+8,0) [regular polygon, regular polygon sides=3,draw=black,fill=red,scale=0.5]{};
\node at (1+8,0) [regular polygon, regular polygon sides=4,draw=black,fill=blue,scale=0.7]{};
\node at (0+8,0.8) [regular polygon, regular polygon sides=3,draw=black,fill=red,scale=0.5]{};
\node at (1+8,0.8) [circle,draw=black,fill=black,scale=0.7]{};
\node at (0.5+8,1.6) [circle,draw=black,fill=black,scale=0.7]{};
\end{tikzpicture}
\end{center}
\caption{The three completely independent spanning trees in $K_{5}\square C_n$, for $K^0\cup K^1\cup K^2\cup K^3\cup K^4$ and for $n\equiv 2\pmod 3$.}
\label{K55}
\end{figure}
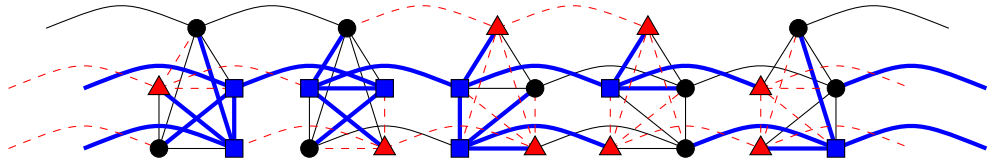
\begin{proof}
By Proposition \ref{propk53}, there exist three completely independent spanning trees in $K_{5}\square C_n$, for $n\equiv0\pmod{3}$.
For $n\equiv1\pmod{3}$, we use the pattern from Proposition \ref{propk53} for $K^{4}\cup \ldots \cup K^{n-1}$, completed by the pieces of three completely independent spanning trees of $K^0\cup K^1\cup K^2\cup K^3$ depicted in Figure~\ref{K54} and whose edge sets are given in Appendix \ref{dek54}.
For $n\equiv2\pmod{3}$, we use the pattern from Proposition \ref{propk53} for $K^{5}\cup \ldots \cup K^{n-1}$, completed by the pieces of three completely independent spanning trees of $K^0\cup K^1\cup K^2\cup K^3\cup K^4$ depicted in Figure~\ref{K55} and whose edge sets are given in Appendix \ref{dek55}.
Note that Figures~\ref{K54} and~\ref{K55} depicte also three completely independent spanning trees in $K_{5}\square C_4$ and $K_{5}\square C_5$.
\end{proof}
\begin{prop}
\label{propk73}
There exist four completely independent spanning trees in $K_{7}\square C_3$.
\end{prop}
\begin{figure}[t]
\begin{center}
\begin{tikzpicture}
\draw (0.5,2.4) -- (0,1.6);
\draw (0.5,2.4) -- (0,0.8);
\draw (0.5,2.4) -- (0,0);
\draw (0.5,2.4) -- (1,0);
\draw (3.5,2.4) -- (4,1.6);
\draw (3.5,2.4) -- (4,0.8);
\draw (3.5,2.4) -- (3,0);
\draw (4,1.6) -- (3,1.6);
\draw (4,1.6) -- (3,0.8);
\draw (4,1.6) .. controls (4.2,0.8) .. (4,0);
\draw (7,1.6) -- (6,0);
\draw (7,1.6) .. controls (7.2,0.8) .. (7,0);
\draw (7,1.6) -- (7,0.8);
\draw (7,0.8) -- (6,0.8);
\draw (7,0.8) -- (6,1.6);
\draw (6.5,2.4) -- (7,0.8);
\draw (0.5,2.4) .. controls (2,2.8) .. (3.5,2.4);
\draw (4,1.6) .. controls (5.5,2) .. (7,1.6);
\draw (7,0.8) .. controls (9,1.6) and (-1,1.6) .. (1,0.8);
\draw (7,1.6) .. controls (9,2.4) and (-1,2.4) .. (1,1.6);
\draw[dashed, draw=red] (0,1.6) -- (1,1.6);
\draw[dashed, draw=red] (0,1.6) -- (0,0.8);
\draw[dashed, draw=red] (1,1.6) -- (1,0.8);
\draw[dashed, draw=red] (0,1.6)  .. controls (-0.2,0.8) ..  (0,0);
\draw[dashed, draw=red] (1,1.6)  .. controls (1.2,0.8) ..  (1,0);
\draw[dashed, draw=red] (0.5,2.4) -- (1,1.6);
\draw[dashed, draw=red] (4,0) -- (3,0);
\draw[dashed, draw=red] (4,0) -- (3,0.8);
\draw[dashed, draw=red] (4,0) -- (4,0.8);
\draw[dashed, draw=red] (4,0) -- (3.5,2.4);
\draw[dashed, draw=red] (7,0) -- (7,0.8);
\draw[dashed, draw=red] (7,0) -- (6,0);
\draw[dashed, draw=red] (7,0) -- (6,1.6);
\draw[dashed, draw=red] (6,0.8) -- (6,1.6);
\draw[dashed, draw=red] (7,1.6) -- (6,1.6);
\draw[dashed, draw=red] (6.5,2.4) -- (6,1.6);
\draw[dashed, draw=red] (0,1.6) .. controls (1.5,2) .. (3,1.6);
\draw[dashed, draw=red] (1,1.6) .. controls (2.5,2) .. (4,1.6);
\draw[dashed, draw=red] (4,0) .. controls (5.5,0.4) .. (7,0);
\draw[dashed, draw=red] (6,1.6) .. controls (8,2.4) and (-2,2.4) .. (0,1.6);
\draw[ultra thick,color=blue] (0,0.8) -- (1,0.8);
\draw[ultra thick,color=blue] (0,0.8) -- (0,0);
\draw[ultra thick,color=blue] (1,0.8) -- (1,0);
\draw[ultra thick,color=blue] (1,0.8) -- (0,1.6);
\draw[ultra thick,color=blue] (0,0.8) -- (1,1.6);
\draw[ultra thick,color=blue] (1,0.8) -- (0.5,2.4);
\draw[ultra thick,color=blue] (4,0.8) -- (3,0);
\draw[ultra thick,color=blue] (4,0.8) -- (4,1.6);
\draw[ultra thick,color=blue] (4,0.8) -- (3,1.6);
\draw[ultra thick,color=blue] (3,1.6) -- (3.5,2.4);
\draw[ultra thick,color=blue] (3,1.6) -- (3,0.8);
\draw[ultra thick,color=blue] (3,1.6) -- (4,0);
\draw[ultra thick,color=blue] (6,0.8) -- (6,0);
\draw[ultra thick,color=blue] (6,0.8) -- (7,0);
\draw[ultra thick,color=blue] (6,0.8) -- (7,1.6);
\draw[ultra thick,color=blue] (6,0.8) -- (6.5,2.4);
\draw[ultra thick,color=blue] (1,0.8) .. controls (2.5,1.2) .. (4,0.8);
\draw[ultra thick,color=blue] (4,0.8) .. controls (5.5,1.2) .. (7,0.8);
\draw[ultra thick,color=blue] (3,1.6) .. controls (4.5,2) .. (6,1.6);
\draw[ultra thick,color=blue] (6,0.8) .. controls (8,1.6) and (-2,1.6) .. (0,0.8);
\draw[dotted, thick,color=green] (0,0) -- (1,0);
\draw[dotted, thick,color=green] (0,0) -- (1,0.8);
\draw[dotted, thick,color=green] (1,0) -- (0,0.8);
\draw[dotted, thick,color=green] (0,0) -- (1,1.6);
\draw[dotted, thick,color=green] (1,0) -- (0,1.6);
\draw[dotted, thick,color=green] (3,0) -- (3,0.8);
\draw[dotted, thick,color=green] (3,0) -- (4,1.6);
\draw[dotted, thick,color=green] (3,0) .. controls (2.8,0.8) .. (3,1.6);
\draw[dotted, thick,color=green] (3,0.8) -- (4,0.8);
\draw[dotted, thick,color=green] (3,0.8) -- (3.5,2.4);
\draw[dotted, thick,color=green] (6,0) -- (6.5,2.4);
\draw[dotted, thick,color=green] (6,0) -- (7,0.8);
\draw[dotted, thick,color=green] (6,0) .. controls (5.8,0.8) .. (6,1.6);
\draw[dotted, thick,color=green] (6.5,2.4) -- (7,1.6);
\draw[dotted, thick,color=green] (6.5,2.4) -- (7,0);
\draw[dotted, thick,color=green] (0,0) .. controls (1.5,0.4) .. (3,0);
\draw[dotted, thick,color=green] (1,0) .. controls (2.5,0.4) .. (4,0);
\draw[dotted, thick,color=green] (3,0.8) .. controls (4.5,1.2) .. (6,0.8);
\draw[dotted, thick,color=green] (3,0) .. controls (4.5,0.4) .. (6,0);
\draw[dotted, thick,color=green] (6.5,2.4) .. controls (8.5,3.2) and (-1.5,3.2) .. (0.5,2.4);
\node at (0,0) [regular polygon, regular polygon sides=5,draw=black,fill=green,scale=0.7]{};
\node at (1,0) [regular polygon, regular polygon sides=5,draw=black,fill=green,scale=0.7]{};
\node at (0,0.8) [regular polygon, regular polygon sides=4,draw=black,fill=blue,scale=0.7]{};
\node at (1,0.8) [regular polygon, regular polygon sides=4,draw=black,fill=blue,scale=0.7]{};
\node at (0,1.6) [regular polygon, regular polygon sides=3,draw=black,fill=red,scale=0.5]{};
\node at (1,1.6) [regular polygon, regular polygon sides=3,draw=black,fill=red,scale=0.5]{};
\node at (0.5,2.4) [circle,draw=black,fill=black,scale=0.7]{};
\node at (0+3,0) [regular polygon, regular polygon sides=5,draw=black,fill=green,scale=0.7]{};
\node at (1+3,0) [regular polygon, regular polygon sides=3,draw=black,fill=red,scale=0.5]{};
\node at (0+3,0.8) [regular polygon, regular polygon sides=5,draw=black,fill=green,scale=0.7]{};
\node at (1+3,0.8) [regular polygon, regular polygon sides=4,draw=black,fill=blue,scale=0.7]{};
\node at (0+3,1.6) [regular polygon, regular polygon sides=4,draw=black,fill=blue,scale=0.7]{};
\node at (1+3,1.6) [circle,draw=black,fill=black,scale=0.7]{};
\node at (0.5+3,2.4) [circle,draw=black,fill=black,scale=0.7]{};
\node at (0+6,0) [regular polygon, regular polygon sides=5,draw=black,fill=green,scale=0.7]{};
\node at (1+6,0) [regular polygon, regular polygon sides=3,draw=black,fill=red,scale=0.5]{};
\node at (0+6,0.8) [regular polygon, regular polygon sides=4,draw=black,fill=blue,scale=0.7]{};
\node at (1+6,0.8) [circle,draw=black,fill=black,scale=0.7]{};
\node at (0+6,1.6) [regular polygon, regular polygon sides=3,draw=black,fill=red,scale=0.5]{};
\node at (1+6,1.6) [circle,draw=black,fill=black,scale=0.7]{};
\node at (0.5+6,2.4) [regular polygon, regular polygon sides=5,draw=black,fill=green,scale=0.7]{};
\end{tikzpicture}
\end{center}
\caption{Four completely independent spanning trees in $K_{7}\square C_3$.}
\label{K73}
\end{figure}
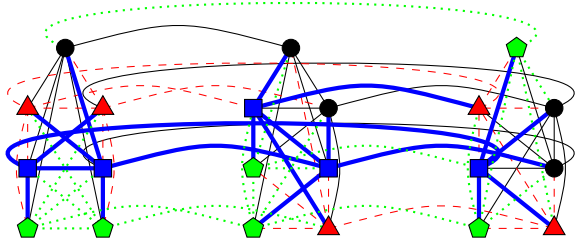
\begin{proof}
The four completely independent spanning trees in $K_{7}\square C_3$ are depicted in Figure~\ref{K73} and their edge sets are given in Appendix \ref{dek73}.
\end{proof}
\begin{prop}
\label{propk74}
There exist four completely independent spanning trees in $K_{7}\square C_4$.
\end{prop}
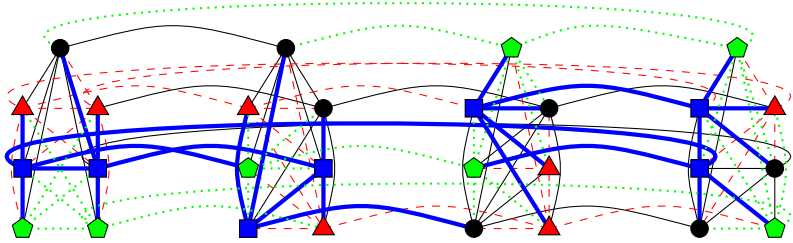
\begin{figure}[t]
\begin{center}
\begin{tikzpicture}
\draw (0.5,2.4) -- (0,1.6);
\draw (0.5,2.4) -- (0,0.8);
\draw (0.5,2.4) -- (0,0);
\draw (0.5,2.4) -- (1,0);
\draw (3.5,2.4) -- (4,1.6);
\draw (3.5,2.4) -- (3,1.6);
\draw (3.5,2.4) -- (3,0.8);
\draw (3.5,2.4) -- (4,0.8);
\draw (4,1.6) -- (3,0);
\draw (4,1.6) .. controls (4.2,0.8) .. (4,0);
\draw (7,1.6) -- (6,0);
\draw (7,1.6) -- (6,0.8);
\draw (7,1.6) .. controls (7.2,0.8) .. (7,0);
\draw (6,0) -- (7,0.8);
\draw (6,1.6) .. controls (5.8,0.8) .. (6,0);
\draw (6,0) -- (6.5,2.4);
\draw (9,0) -- (10,0.8);
\draw (9,0) .. controls (8.8,0.8) .. (9,1.6);
\draw (9,0) -- (9.5,2.4);
\draw (10,0.8) -- (9,0.8);
\draw (10,0.8) -- (10,0);
\draw (0.5,2.4) .. controls (2,2.8) .. (3.5,2.4);
\draw (1,1.6) .. controls (2.5,2) .. (4,1.6);
\draw (4,1.6) .. controls (5.5,2) .. (7,1.6);
\draw (7,1.6) .. controls (8.5,2) .. (10,1.6);
\draw (6,0) .. controls (7.5,0.4) .. (9,0);
\draw (10,0.8) .. controls (12,1.6) and (-1,1.6) .. (1,0.8);
\draw[dashed, draw=red] (0,1.6) -- (1,1.6);
\draw[dashed, draw=red] (0,1.6) .. controls (-0.2,0.8) .. (0,0);
\draw[dashed, draw=red] (1,1.6) .. controls (1.2,0.8) .. (1,0);
\draw[dashed, draw=red] (0,1.6) -- (1,0.8);
\draw[dashed, draw=red] (1,1.6) -- (0,0.8);
\draw[dashed, draw=red] (1,1.6) -- (0.5,2.4);
\draw[dashed, draw=red] (3,1.6) -- (4,0);
\draw[dashed, draw=red] (3,1.6) -- (4,0.8);
\draw[dashed, draw=red] (3,1.6) -- (4,1.6);
\draw[dashed, draw=red] (4,0) -- (3.5,2.4);
\draw[dashed, draw=red] (4,0) -- (3,0.8);
\draw[dashed, draw=red] (4,0) -- (3,0);
\draw[dashed, draw=red] (7,0) -- (7,0.8);
\draw[dashed, draw=red] (7,0) -- (6.5,2.4);
\draw[dashed, draw=red] (7,0) -- (6,0);
\draw[dashed, draw=red] (7,0.8) -- (7,1.6);
\draw[dashed, draw=red] (7,0.8) -- (6,0.8);
\draw[dashed, draw=red] (10,1.6) -- (9,0.8);
\draw[dashed, draw=red] (10,1.6) -- (9,0);
\draw[dashed, draw=red] (10,1.6) -- (10,0.8);
\draw[dashed, draw=red] (10,1.6) -- (9.5,2.4);
\draw[dashed, draw=red] (0,1.6) .. controls (1.5,2) .. (3,1.6);
\draw[dashed, draw=red] (4,0) .. controls (5.5,0.4) .. (7,0);
\draw[dashed, draw=red] (3,1.6) .. controls (4.5,2) .. (6,1.6);
\draw[dashed, draw=red] (7,0) .. controls (8.5,0.4) .. (10,0);
\draw[dashed, draw=red] (9,1.6) .. controls (11,2.4) and (-2,2.4) .. (0,1.6);
\draw[dashed, draw=red] (10,1.6) .. controls (12,2.4) and (-1,2.4) .. (1,1.6);
\draw[ultra thick,color=blue] (0,0.8) -- (1,0.8);
\draw[ultra thick,color=blue] (0,0.8) -- (0,0);
\draw[ultra thick,color=blue] (0,0.8) -- (0,1.6);
\draw[ultra thick,color=blue] (1,0.8) -- (1,0);
\draw[ultra thick,color=blue] (1,0.8) -- (1,1.6);
\draw[ultra thick,color=blue] (1,0.8) -- (0.5,2.4);
\draw[ultra thick,color=blue] (4,0.8) -- (3,0);
\draw[ultra thick,color=blue] (4,0.8) -- (4,0);
\draw[ultra thick,color=blue] (4,0.8) -- (4,1.6);
\draw[ultra thick,color=blue] (3,0) .. controls (2.8,0.8) .. (3,1.6);
\draw[ultra thick,color=blue] (3,0) -- (3.5,2.4);
\draw[ultra thick,color=blue] (6,1.6) -- (7,1.6);
\draw[ultra thick,color=blue] (6,1.6) -- (6.5,2.4);
\draw[ultra thick,color=blue] (6,1.6) -- (7,0.8);
\draw[ultra thick,color=blue] (6,1.6) -- (7,0);
\draw[ultra thick,color=blue] (9,0.8) -- (9,1.6);
\draw[ultra thick,color=blue] (9,0.8) -- (9,0);
\draw[ultra thick,color=blue] (9,0.8) -- (10,0);
\draw[ultra thick,color=blue] (9,1.6) -- (10,0.8);
\draw[ultra thick,color=blue] (9,1.6) -- (10,1.6);
\draw[ultra thick,color=blue] (9,1.6) -- (9.5,2.4);
\draw[ultra thick,color=blue] (0,0.8) .. controls (1.5,1.2) .. (3,0.8);
\draw[ultra thick,color=blue] (1,0.8) .. controls (2.5,1.2) .. (4,0.8);
\draw[ultra thick,color=blue] (3,0) .. controls (4.5,0.4) .. (6,0);
\draw[ultra thick,color=blue] (6,1.6) .. controls (7.5,2) .. (9,1.6);
\draw[ultra thick,color=blue] (6,0.8) .. controls (7.5,1.2) .. (9,0.8);
\draw[ultra thick,color=blue] (9,0.8) .. controls (11,1.6) and (-2,1.6) .. (0,0.8);
\draw[dotted, thick,color=green] (0,0) --(1,0);
\draw[dotted, thick,color=green] (0,0) -- (1,0.8);
\draw[dotted, thick,color=green] (0,0) -- (1,1.6);
\draw[dotted, thick,color=green] (1,0) -- (0,0.8);
\draw[dotted, thick,color=green] (1,0) -- (0,1.6);
\draw[dotted, thick,color=green] (3,0.8) -- (3,0);
\draw[dotted, thick,color=green] (3,0.8) -- (4,0.8);
\draw[dotted, thick,color=green] (3,0.8) -- (4,1.6);
\draw[dotted, thick,color=green] (3,0.8) -- (3,1.6);
\draw[dotted, thick,color=green] (6,0.8) -- (6.5,2.4);
\draw[dotted, thick,color=green] (6,0.8) -- (6,0);
\draw[dotted, thick,color=green] (6,0.8) -- (6,1.6);
\draw[dotted, thick,color=green] (6,0.8) -- (7,0);
\draw[dotted, thick,color=green] (6.5,2.4) -- (7,1.6);
\draw[dotted, thick,color=green] (6.5,2.4) -- (7,0.8);
\draw[dotted, thick,color=green] (10,0) -- (9.5,2.4);
\draw[dotted, thick,color=green] (10,0) -- (9,0);
\draw[dotted, thick,color=green] (10,0) -- (9,1.6);
\draw[dotted, thick,color=green] (10,0) .. controls (10.2,0.8) .. (10,1.6);
\draw[dotted, thick,color=green] (9.5,2.4) -- (9,0.8);
\draw[dotted, thick,color=green] (9.5,2.4) -- (10,0.8);
\draw[dotted, thick,color=green] (1,0) .. controls (2.5,0.4) .. (4,0);
\draw[dotted, thick,color=green] (3,0.8) .. controls (4.5,1.2) .. (6,0.8);
\draw[dotted, thick,color=green] (3.5,2.4) .. controls (5,2.8) .. (6.5,2.4);
\draw[dotted, thick,color=green] (9.5,2.4) .. controls (7.5,2.8) .. (6.5,2.4);
\draw[dotted, thick,color=green] (10,0) .. controls (12,0.8) and (-1,0.8) .. (1,0);
\draw[dotted, thick,color=green] (9.5,2.4) .. controls (11.5,3.2) and (-1.5,3.2) .. (0.5,2.4);
\node at (0,0) [regular polygon, regular polygon sides=5,draw=black,fill=green,scale=0.7]{};
\node at (1,0) [regular polygon, regular polygon sides=5,draw=black,fill=green,scale=0.7]{};
\node at (0,0.8) [regular polygon, regular polygon sides=4,draw=black,fill=blue,scale=0.7]{};
\node at (1,0.8) [regular polygon, regular polygon sides=4,draw=black,fill=blue,scale=0.7]{};
\node at (0,1.6) [regular polygon, regular polygon sides=3,draw=black,fill=red,scale=0.5]{};
\node at (1,1.6) [regular polygon, regular polygon sides=3,draw=black,fill=red,scale=0.5]{};
\node at (0.5,2.4) [circle,draw=black,fill=black,scale=0.7]{};
\node at (0+3,0) [regular polygon, regular polygon sides=4,draw=black,fill=blue,scale=0.7]{};
\node at (1+3,0) [regular polygon, regular polygon sides=3,draw=black,fill=red,scale=0.5]{};
\node at (0+3,0.8) [regular polygon, regular polygon sides=5,draw=black,fill=green,scale=0.7]{};
\node at (1+3,0.8) [regular polygon, regular polygon sides=4,draw=black,fill=blue,scale=0.7]{};
\node at (0+3,1.6) [regular polygon, regular polygon sides=3,draw=black,fill=red,scale=0.5]{};
\node at (1+3,1.6) [circle,draw=black,fill=black,scale=0.7]{};
\node at (0.5+3,2.4) [circle,draw=black,fill=black,scale=0.7]{};
\node at (0+6,0) [circle,draw=black,fill=black,scale=0.7]{};
\node at (1+6,0) [regular polygon, regular polygon sides=3,draw=black,fill=red,scale=0.5]{};
\node at (0+6,0.8) [regular polygon, regular polygon sides=5,draw=black,fill=green,scale=0.7]{};
\node at (1+6,0.8) [regular polygon, regular polygon sides=3,draw=black,fill=red,scale=0.5]{};
\node at (0+6,1.6) [regular polygon, regular polygon sides=4,draw=black,fill=blue,scale=0.7]{};
\node at (1+6,1.6) [circle,draw=black,fill=black,scale=0.7]{};
\node at (0.5+6,2.4) [regular polygon, regular polygon sides=5,draw=black,fill=green,scale=0.7]{};
\node at (0+9,0) [circle,draw=black,fill=black,scale=0.7]{};
\node at (1+9,0) [regular polygon, regular polygon sides=5,draw=black,fill=green,scale=0.7]{};
\node at (0+9,0.8) [regular polygon, regular polygon sides=4,draw=black,fill=blue,scale=0.7]{};
\node at (1+9,0.8) [circle,draw=black,fill=black,scale=0.7]{};
\node at (0+9,1.6) [regular polygon, regular polygon sides=4,draw=black,fill=blue,scale=0.7]{};
\node at (1+9,1.6) [regular polygon, regular polygon sides=3,draw=black,fill=red,scale=0.5]{};
\node at (0.5+9,2.4) [regular polygon, regular polygon sides=5,draw=black,fill=green,scale=0.7]{};
\end{tikzpicture}
\end{center}
\caption{Four completely independent spanning trees in $K_{7}\square C_4$.}
\label{K74}
\end{figure}
\begin{proof}
The four completely independent spanning trees in $K_{7}\square C_4$ are depicted in Figure~\ref{K74} and their edge sets are given in Appendix \ref{dek74}.
\end{proof}
\begin{prop}
\label{propk94}
There exist five completely independent spanning trees in $K_{9}\square C_4$.
\end{prop}
\begin{sidewaysfigure}
\begin{center}
\begin{tikzpicture}[scale=1.7]
\draw (0.5,3.2) -- (1,2.4);
\draw (0.5,3.2) -- (1,1.6);
\draw (0.5,3.2) -- (1,0);
\draw (0.5,3.2) -- (0,0.8);
\draw (0,0.8) -- (1,0.8);
\draw (0,0.8) -- (0,0);
\draw (0,0.8) -- (0,1.6);
\draw (0,0.8) .. controls (-0.2,1.6) .. (0,2.4);
\draw (3.5,3.2) -- (4,2.4);
\draw (3.5,3.2) -- (4,1.6);
\draw (3.5,3.2) -- (4,0.8);
\draw (3.5,3.2) -- (3,0);
\draw (4,1.6) -- (3,0.8);
\draw (4,1.6) -- (3,1.6);
\draw (4,1.6) -- (3,2.4);
\draw (4,1.6) .. controls (4.2,0.8) .. (4,0);
\draw (7,0) -- (6,0);
\draw (7,0) -- (6,0.8);
\draw (7,0) .. controls (7.2,0.8) .. (7,1.6);
\draw (7,0) .. controls (7.4,1.2) .. (7,2.4);
\draw (7,1.6) -- (6,2.4);
\draw (7,1.6) -- (6.5,3.2);
\draw (6,1.6) -- (6.5,3.2);
\draw (7,0.8) -- (6.5,3.2);
\draw (10,0) -- (9,0);
\draw (10,0) -- (9,1.6);
\draw (10,0) -- (9,2.4);
\draw (10,0) -- (10,0.8);
\draw (10,0) .. controls (10.4,1.2) .. (10,2.4);
\draw (0.5,3.2) .. controls (2,3.6) .. (3.5,3.2);
\draw (3.5,3.2) .. controls (5,3.6) .. (6.5,3.2);
\draw (7,0) .. controls (8.5,0.4) .. (10,0);
\draw (7,1.6) .. controls (8.5,2) .. (10,1.6);
\draw (9.5,3.2) .. controls (11.5,4) and (-1.5,4) .. (0.5,3.2);
\draw (9,0.8) .. controls (11,1.6) and (-2,1.6) .. (0,0.8);
\draw[dashed, draw=red] (0,1.6) -- (1,1.6);
\draw[dashed, draw=red] (0,1.6) .. controls (-0.2,0.8) .. (0,0);
\draw[dashed, draw=red] (1,1.6) -- (1,0.8);
\draw[dashed, draw=red] (1,1.6) -- (0,0.8);
\draw[dashed, draw=red] (1,1.6) -- (0,2.4);
\draw[dashed, draw=red] (1,1.6) -- (1,2.4);
\draw[dashed, draw=red] (0,1.6) -- (1,0);
\draw[dashed, draw=red] (0,1.6) -- (0.5,3.2);
\draw[dashed, draw=red] (3,2.4) -- (4,2.4);
\draw[dashed, draw=red] (3,2.4) -- (4,0.8);
\draw[dashed, draw=red] (3,2.4) -- (4,0);
\draw[dashed, draw=red] (3,2.4) -- (3.5,3.2);
\draw[dashed, draw=red] (3,2.4) .. controls (2.6,1.2).. (3,0);
\draw[dashed, draw=red] (6,2.4) -- (7,2.4);
\draw[dashed, draw=red] (6,2.4) -- (6.5,3.2);
\draw[dashed, draw=red] (6,2.4) -- (7,0);
\draw[dashed, draw=red] (6,2.4) .. controls (5.6,1.2) .. (6,0);
\draw[dashed, draw=red] (6,2.4) .. controls (5.8,1.6) .. (6,0.8);
\draw[dashed, draw=red] (6,0.8) -- (7,1.6);
\draw[dashed, draw=red] (6,0.8) -- (7,0.8);
\draw[dashed, draw=red] (9,0.8) -- (9,1.6);
\draw[dashed, draw=red] (9,0.8) -- (10,0);
\draw[dashed, draw=red] (9,0.8) -- (10,1.6);
\draw[dashed, draw=red] (9,0.8) -- (9.5,3.2);
\draw[dashed, draw=red] (9,1.6) -- (10,0.8);
\draw[dashed, draw=red] (9,1.6) -- (10,2.4);
\draw[dashed, draw=red] (9,1.6) .. controls (8.8,0.8) .. (9,0);
\draw[dashed, draw=red] (9,2.4) .. controls (8.8,1.6) .. (9,0.8);
\draw[dashed, draw=red] (0,1.6) .. controls (1.5,2) .. (3,1.6);
\draw[dashed, draw=red] (1,1.6) .. controls (2.5,2) .. (4,1.6);
\draw[dashed, draw=red] (3,0.8) .. controls (4.5,1.2) .. (6,0.8);
\draw[dashed, draw=red] (3,2.4) .. controls (4.5,2.8) .. (6,2.4);
\draw[dashed, draw=red] (6,0.8) .. controls (7.5,1.2) .. (9,0.8);
\draw[dashed, draw=red] (6,1.6) .. controls (7.5,2) .. (9,1.6);
\draw[dashed, draw=red] (9,1.6) .. controls (11,2.4) and (-2,2.4) .. (0,1.6);
\draw[ultra thick,color=blue] (0,0) -- (1,0);
\draw[ultra thick,color=blue] (0,0) -- (1,0.8);
\draw[ultra thick,color=blue] (0,0) -- (1,1.6);
\draw[ultra thick,color=blue] (0,0) -- (1,2.4);
\draw[ultra thick,color=blue] (0,0) -- (0.5,3.2);
\draw[ultra thick,color=blue] (1,2.4) -- (0,2.4);
\draw[ultra thick,color=blue] (1,2.4) -- (0,1.6);
\draw[ultra thick,color=blue] (1,2.4) -- (0,0.8);
\draw[ultra thick,color=blue] (3,1.6) -- (4,0);
\draw[ultra thick,color=blue] (3,1.6) -- (4,2.4);
\draw[ultra thick,color=blue] (3,0) .. controls (2.8,0.8) .. (3,1.6);
\draw[ultra thick,color=blue] (3,1.6) -- (3.5,3.2);
\draw[ultra thick,color=blue] (3,1.6) -- (3,2.4);
\draw[ultra thick,color=blue] (4,2.4) .. controls (4.2,1.6) .. (4,0.8);
\draw[ultra thick,color=blue] (4,2.4) -- (4,1.6);
\draw[ultra thick,color=blue] (4,2.4) -- (3,1.6);
\draw[ultra thick,color=blue] (3,0.8) -- (4,2.4);
\draw[ultra thick,color=blue] (6,1.6) -- (7,1.6);
\draw[ultra thick,color=blue] (6,1.6) -- (6,2.4);
\draw[ultra thick,color=blue] (6,1.6) -- (6,0.8);
\draw[ultra thick,color=blue] (6,1.6) -- (7,0.8);
\draw[ultra thick,color=blue] (6,1.6) -- (7,0);
\draw[ultra thick,color=blue] (9,0) -- (10,0.8);
\draw[ultra thick,color=blue] (9,0) -- (9,0.8);
\draw[ultra thick,color=blue] (9,0) -- (9.5,3.2);
\draw[ultra thick,color=blue] (9,0) .. controls (8.6,1.2) .. (9,2.4);
\draw[ultra thick,color=blue] (9.5,3.2) -- (10,2.4);
\draw[ultra thick,color=blue] (9.5,3.2) -- (10,1.6);
\draw[ultra thick,color=blue] (9.5,3.2) -- (9,1.6);
\draw[ultra thick,color=blue] (9.5,3.2) -- (10,0);
\draw[ultra thick,color=blue] (1,2.4) .. controls (2.5,2.8) .. (4,2.4);
\draw[ultra thick,color=blue] (4,2.4) .. controls (5.5,2.8) .. (7,2.4);
\draw[ultra thick,color=blue] (3,1.6) .. controls (4.5,2) .. (6,1.6);
\draw[ultra thick,color=blue] (6,0) .. controls (7.5,0.4) .. (9,0);
\draw[ultra thick,color=blue] (6.5,3.2) .. controls (8,3.6) .. (9.5,3.2);
\draw[ultra thick,color=blue] (9,0) .. controls (11,0.8) and (-2,0.8) .. (0,0);
\draw[dotted, thick,color=green] (1,0.8) --(1,0);
\draw[dotted, thick,color=green] (1,0.8) -- (0,1.6);
\draw[dotted, thick,color=green] (1,0.8) -- (0,2.4);
\draw[dotted, thick,color=green] (1,0.8) .. controls (1.2,1.6) .. (1,2.4);
\draw[dotted, thick,color=green] (1,0.8) -- (0.5,3.2);
\draw[dotted, thick,color=green] (3,0.8) -- (3,0);
\draw[dotted, thick,color=green] (3,0.8) -- (4,0);
\draw[dotted, thick,color=green] (3,0.8) -- (3,1.6);
\draw[dotted, thick,color=green] (3,0.8) -- (3.5,3.2);
\draw[dotted, thick,color=green] (3,0.8) .. controls (2.8,1.6) .. (3,2.4);
\draw[dotted, thick,color=green] (3,0) -- (4,0.8);
\draw[dotted, thick,color=green] (3,0) -- (4,1.6);
\draw[dotted, thick,color=green] (3,0) -- (4,2.4);
\draw[dotted, thick,color=green] (6,0) -- (6.5,3.2);
\draw[dotted, thick,color=green] (6,0) -- (6,0.8);
\draw[dotted, thick,color=green] (6,0) -- (7,0.8);
\draw[dotted, thick,color=green] (6,0) .. controls (5.8,0.8) .. (6,1.6);
\draw[dotted, thick,color=green] (7,0.8) .. controls (7.2,1.6) .. (7,2.4);
\draw[dotted, thick,color=green] (7,0) -- (7,0.8);
\draw[dotted, thick,color=green] (7,0.8) -- (7,1.6);
\draw[dotted, thick,color=green] (7,0.8) -- (6,2.4);
\draw[dotted, thick,color=green] (10,0.8) -- (9.5,3.2);
\draw[dotted, thick,color=green] (10,0.8) -- (9,2.4);
\draw[dotted, thick,color=green] (10,0.8) -- (10,1.6);
\draw[dotted, thick,color=green] (10,0.8) -- (9,0.8);
\draw[dotted, thick,color=green] (10,0) .. controls (10.2,0.8) .. (10,1.6);
\draw[dotted, thick,color=green] (10,1.6) -- (9,1.6);
\draw[dotted, thick,color=green] (10,1.6) -- (10,2.4);
\draw[dotted, thick,color=green] (10,1.6) -- (9,0);
\draw[dotted, thick,color=green] (0,0) .. controls (1.5,0.4) .. (3,0);
\draw[dotted, thick,color=green] (0,0.8) .. controls (1.5,1.2) .. (3,0.8);
\draw[dotted, thick,color=green] (3,0) .. controls (4.5,0.4) .. (6,0);
\draw[dotted, thick,color=green] (7,0.8) .. controls (8.5,1.2) .. (10,0.8);
\draw[dotted, thick,color=green] (10,0.8) .. controls (12,1.6) and (-1,1.6) .. (1,0.8);
\draw[dotted, thick,color=green] (10,1.6) .. controls (12,2.4) and (-1,2.4) .. (1,1.6);
\draw[densely dashed, ultra thick,color=orange] (0,1.6) -- (0,2.4);
\draw[densely dashed, ultra thick,color=orange] (0,2.4) -- (0.5,3.2);
\draw[densely dashed, ultra thick,color=orange] (0,2.4) -- (1,0);
\draw[densely dashed, ultra thick,color=orange] (0,2.4) .. controls (-0.4,1.2) .. (0,0);
\draw[densely dashed, ultra thick,color=orange] (1,2.4) .. controls (1.4,1.2) .. (1,0);
\draw[densely dashed, ultra thick,color=orange] (1,1.6) .. controls (1.2,0.8) .. (1,0);
\draw[densely dashed, ultra thick,color=orange] (1,0) -- (0,0.8);
\draw[densely dashed, ultra thick,color=orange] (4,0) -- (4,0.8);
\draw[densely dashed, ultra thick,color=orange] (4,0) -- (3,0);
\draw[densely dashed, ultra thick,color=orange] (4,0) -- (3.5,3.2);
\draw[densely dashed, ultra thick,color=orange] (4,0) .. controls (4.2,1.2) .. (4,2.4);
\draw[densely dashed, ultra thick,color=orange] (4,0.8) -- (3,0.8);
\draw[densely dashed, ultra thick,color=orange] (4,0.8) -- (3,1.6);
\draw[densely dashed, ultra thick,color=orange] (4,0.8) -- (4,1.6);
\draw[densely dashed, ultra thick,color=orange] (7,2.4) -- (6.5,3.2);
\draw[densely dashed, ultra thick,color=orange] (7,2.4) -- (6,0);
\draw[densely dashed, ultra thick,color=orange] (7,2.4) -- (6,0.8);
\draw[densely dashed, ultra thick,color=orange] (7,2.4) -- (6,1.6);
\draw[densely dashed, ultra thick,color=orange] (7,2.4) -- (7,1.6);
\draw[densely dashed, ultra thick,color=orange] (9,2.4) -- (9,1.6);
\draw[densely dashed, ultra thick,color=orange] (9,2.4) -- (9.5,3.2);
\draw[densely dashed, ultra thick,color=orange] (9,2.4) -- (10,2.4);
\draw[densely dashed, ultra thick,color=orange] (9,2.4) -- (10,1.6);
\draw[densely dashed, ultra thick,color=orange] (10,2.4) -- (9,0.8);
\draw[densely dashed, ultra thick,color=orange] (10,2.4) -- (9,0);
\draw[densely dashed, ultra thick,color=orange] (10,2.4) .. controls (10.2,1.6) .. (10,0.8);
\draw[densely dashed, ultra thick,color=orange] (1,0.8) .. controls (2.5,1.2) .. (4,0.8);
\draw[densely dashed, ultra thick,color=orange] (1,0) .. controls (2.5,0.4) .. (4,0);
\draw[densely dashed, ultra thick,color=orange] (0,2.4) .. controls (1.5,2.8) .. (3,2.4);
\draw[densely dashed, ultra thick,color=orange] (4,0.8) .. controls (5.5,1.2) .. (7,0.8);
\draw[densely dashed, ultra thick,color=orange] (4,0) .. controls (5.5,0.4) .. (7,0);
\draw[densely dashed, ultra thick,color=orange] (6,2.4) .. controls (7.5,2.8) .. (9,2.4);
\draw[densely dashed, ultra thick,color=orange] (7,2.4) .. controls (8.5,2.8) .. (10,2.4);
\draw[densely dashed, ultra thick,color=orange] (10,0) .. controls (12,0.8) and (-1,0.8) .. (1,0);
\draw[densely dashed, ultra thick,color=orange] (9,2.4) .. controls (11,3.2) and (-2,3.2) .. (0,2.4);
\node at (1,0) [diamond,draw=black,fill=orange,scale=0.7]{};
\node at (0,2.4) [diamond,draw=black,fill=orange,scale=0.7]{};
\node at (1,0.8) [regular polygon, regular polygon sides=5,draw=black,fill=green,scale=0.7]{};
\node at (1,2.4) [regular polygon, regular polygon sides=4,draw=black,fill=blue,scale=0.7]{};
\node at (0,0) [regular polygon, regular polygon sides=4,draw=black,fill=blue,scale=0.7]{};
\node at (0,1.6) [regular polygon, regular polygon sides=3,draw=black,fill=red,scale=0.5]{};
\node at (1,1.6) [regular polygon, regular polygon sides=3,draw=black,fill=red,scale=0.5]{};
\node at (0,0.8) [circle,draw=black,fill=black,scale=0.7]{};
\node at (0.5,3.2) [circle,draw=black,fill=black,scale=0.7]{};
\node at (1+3,0.8) [diamond,draw=black,fill=orange,scale=0.7]{};
\node at (1+3,0) [diamond,draw=black,fill=orange,scale=0.7]{};
\node at (0+3,0) [regular polygon, regular polygon sides=5,draw=black,fill=green,scale=0.7]{};
\node at (0+3,0.8) [regular polygon, regular polygon sides=5,draw=black,fill=green,scale=0.7]{};
\node at (0+3,1.6) [regular polygon, regular polygon sides=4,draw=black,fill=blue,scale=0.7]{};
\node at (1+3,2.4) [regular polygon, regular polygon sides=4,draw=black,fill=blue,scale=0.7]{};
\node at (0+3,2.4) [regular polygon, regular polygon sides=3,draw=black,fill=red,scale=0.5]{};
\node at (1+3,1.6) [circle,draw=black,fill=black,scale=0.7]{};
\node at (0.5+3,3.2) [circle,draw=black,fill=black,scale=0.7]{};
\node at (1+6,2.4) [diamond,draw=black,fill=orange,scale=0.7]{};
\node at (1+6,0.8) [regular polygon, regular polygon sides=5,draw=black,fill=green,scale=0.7]{};
\node at (0+6,1.6) [regular polygon, regular polygon sides=4,draw=black,fill=blue,scale=0.7]{};
\node at (0+6,0) [regular polygon, regular polygon sides=5,draw=black,fill=green,scale=0.7]{};
\node at (0+6,0.8) [regular polygon, regular polygon sides=3,draw=black,fill=red,scale=0.5]{};
\node at (0+6,2.4) [regular polygon, regular polygon sides=3,draw=black,fill=red,scale=0.5]{};
\node at (0.5+6,3.2) [circle,draw=black,fill=black,scale=0.7]{};
\node at (1+6,1.6) [circle,draw=black,fill=black,scale=0.7]{};
\node at (1+6,0) [circle,draw=black,fill=black,scale=0.7]{};
\node at (0+9,0) [circle,draw=black,fill=black,scale=0.7]{};
\node at (1+9,2.4) [diamond,draw=black,fill=orange,scale=0.7]{};
\node at (0+9,2.4) [diamond,draw=black,fill=orange,scale=0.7]{};
\node at (0.5+9,3.2) [regular polygon, regular polygon sides=4,draw=black,fill=blue,scale=0.7]{};
\node at (0+9,0) [regular polygon, regular polygon sides=4,draw=black,fill=blue,scale=0.7]{};
\node at (1+9,0.8) [regular polygon, regular polygon sides=5,draw=black,fill=green,scale=0.7]{};
\node at (1+9,1.6) [regular polygon, regular polygon sides=5,draw=black,fill=green,scale=0.7]{};
\node at (0+9,1.6) [regular polygon, regular polygon sides=3,draw=black,fill=red,scale=0.5]{};
\node at (0+9,0.8) [regular polygon, regular polygon sides=3,draw=black,fill=red,scale=0.5]{};
\node at (1+9,0) [circle,draw=black,fill=black,scale=0.7]{};
\end{tikzpicture}
\end{center}
\caption{Five completely independent spanning trees in $K_{9}\square C_4$.}
\label{K94}
\end{sidewaysfigure}
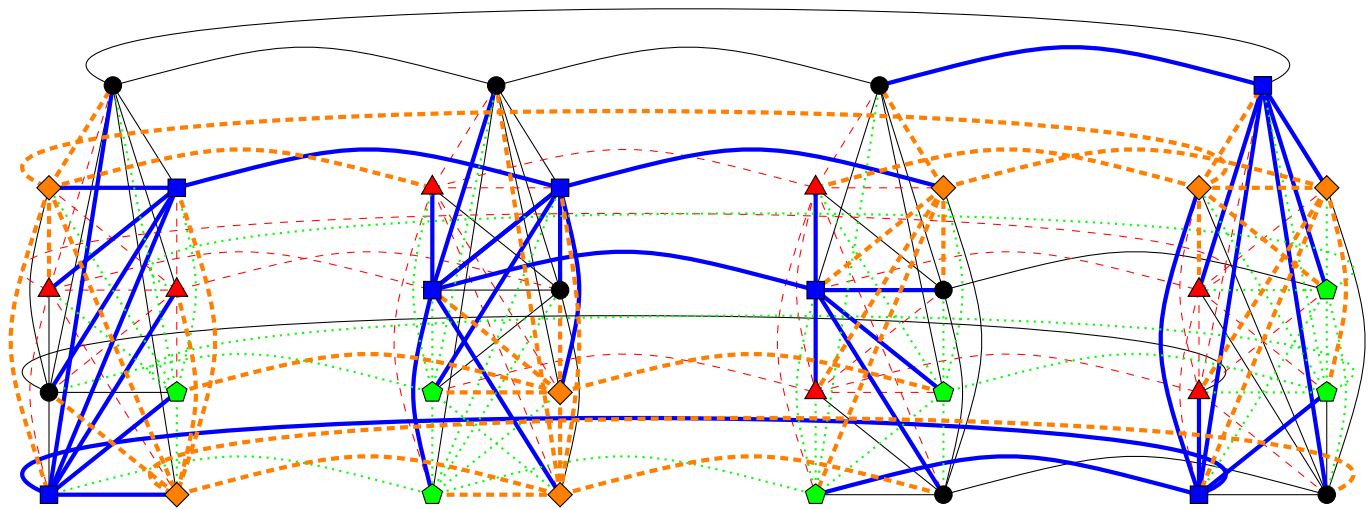
\begin{proof}
The five completely independent spanning trees in $K_{9}\square C_4$ are depicted in Figure~\ref{K94} and their edge sets are given in Appendix \ref{dek94}.
\end{proof}
\begin{prop}
\label{propk95}
There exist five completely independent spanning trees in $K_{9}\square C_5$.
\end{prop}
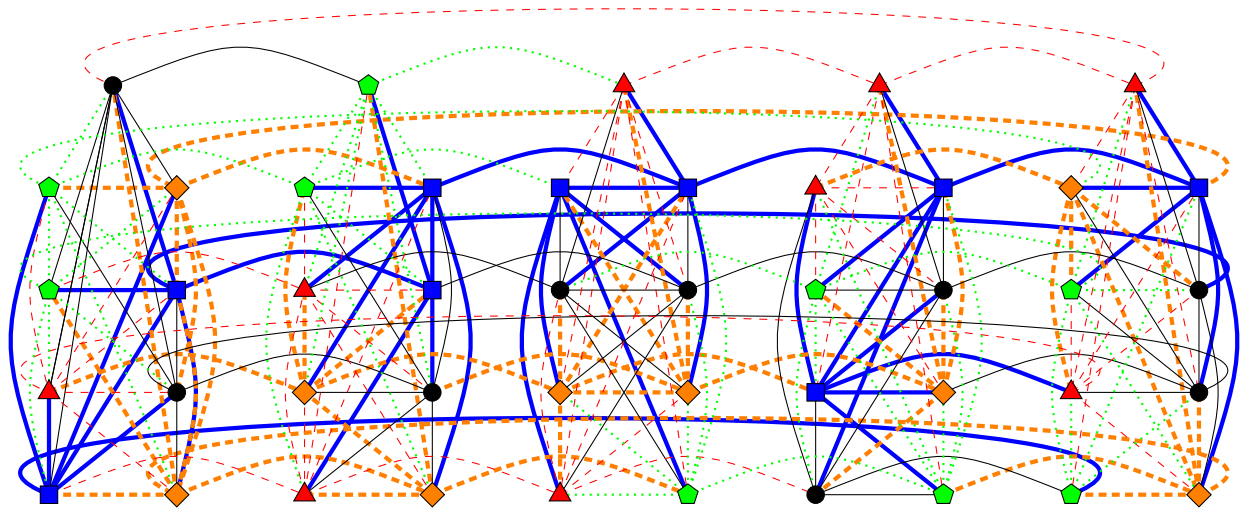
\begin{sidewaysfigure}
\begin{center}
\begin{tikzpicture}[scale=1.7]
\draw[ultra thick,color=blue] (0,0) -- (0,0.8);
\draw[ultra thick,color=blue] (0,0) -- (1,0.8);
\draw[ultra thick,color=blue] (0,0) -- (1,1.6);
\draw[ultra thick,color=blue] (0,0) -- (1,2.4);
\draw[ultra thick,color=blue] (0,0) .. controls (-0.4,1.2) .. (0,2.4);
\draw[ultra thick,color=blue] (1,1.6) -- (0.5,3.2);
\draw[ultra thick,color=blue] (1,1.6) -- (0,1.6);
\draw[ultra thick,color=blue] (1,1.6)  .. controls (1.2,0.8) .. (1,0);
\draw[ultra thick,color=blue] (3,2.4) -- (2,2.4);
\draw[ultra thick,color=blue] (3,2.4) -- (2,1.6);
\draw[ultra thick,color=blue] (3,2.4) -- (2,0.8);
\draw[ultra thick,color=blue] (3,2.4) -- (3,1.6);
\draw[ultra thick,color=blue] (3,0) .. controls (3.4,1.2) .. (3,2.4);
\draw[ultra thick,color=blue] (3,1.6) -- (3,0.8);
\draw[ultra thick,color=blue] (3,1.6) -- (2,0);
\draw[ultra thick,color=blue] (3,1.6) -- (2.5,3.2);
\draw[ultra thick,color=blue] (5,2.4) -- (4.5,3.2);
\draw[ultra thick,color=blue] (5,2.4) -- (4,2.4);
\draw[ultra thick,color=blue] (5,2.4) -- (4,1.6);
\draw[ultra thick,color=blue] (4,2.4) -- (5,1.6);
\draw[ultra thick,color=blue] (4,2.4) -- (5,0);
\draw[ultra thick,color=blue] (4,0.8) .. controls (3.8,1.6) .. (4,2.4);
\draw[ultra thick,color=blue] (4,0) .. controls (3.6,1.2) .. (4,2.4);
\draw[ultra thick,color=blue] (5,0.8) .. controls (5.2,1.6) .. (5,2.4);
\draw[ultra thick,color=blue] (6,0.8) -- (7,1.6);
\draw[ultra thick,color=blue] (6,0.8) -- (7,2.4);
\draw[ultra thick,color=blue] (6,0.8) -- (7,0.8);
\draw[ultra thick,color=blue] (6,0.8) -- (7,0);
\draw[ultra thick,color=blue] (6,0.8) .. controls (5.8,1.6) .. (6,2.4);
\draw[ultra thick,color=blue] (7,2.4) -- (6.5,3.2);
\draw[ultra thick,color=blue] (7,2.4) -- (6,1.6);
\draw[ultra thick,color=blue] (7,2.4) -- (6,0);
\draw[ultra thick,color=blue] (9,0) .. controls (9.4,1.2) .. (9,2.4);
\draw[ultra thick,color=blue] (9,2.4) -- (8,2.4);
\draw[ultra thick,color=blue] (9,2.4) -- (8,1.6);
\draw[ultra thick,color=blue] (9,2.4) .. controls (9.2,1.6) .. (9,0.8);
\draw[ultra thick,color=blue] (9,2.4) -- (8.5,3.2);
\draw[ultra thick,color=blue] (1,1.6) .. controls (2,2) .. (3,1.6);
\draw[ultra thick,color=blue] (3,2.4) .. controls (4,2.8) .. (5,2.4);
\draw[ultra thick,color=blue] (5,2.4) .. controls (6,2.8) .. (7,2.4);
\draw[ultra thick,color=blue] (7,2.4) .. controls (8,2.8) .. (9,2.4);
\draw[ultra thick,color=blue] (6,0.8) .. controls (7,1.2) .. (8,0.8);
\draw[ultra thick,color=blue] (8,0) .. controls (10,0.8) and (-2,0.8) .. (0,0);
\draw[ultra thick,color=blue] (9,1.6) .. controls (11,2.4) and (-1,2.4) .. (1,1.6);
\draw[densely dashed, ultra thick,color=orange] (1,0) -- (0,0);
\draw[densely dashed, ultra thick,color=orange] (1,0) -- (0,1.6);
\draw[densely dashed, ultra thick,color=orange] (1,0) -- (0.5,3.2);
\draw[densely dashed, ultra thick,color=orange] (1,2.4) -- (0,2.4);
\draw[densely dashed, ultra thick,color=orange] (1,2.4) -- (1,1.6);
\draw[densely dashed, ultra thick,color=orange] (1,2.4) .. controls (1.4,1.2) .. (1,0);
\draw[densely dashed, ultra thick,color=orange] (1,2.4) .. controls (1.2,1.6) .. (1,0.8);
\draw[densely dashed, ultra thick,color=orange] (1,1.6) .. controls (1.2,0.8) .. (1,0);
\draw[densely dashed, ultra thick,color=orange] (3,0) -- (2,0);
\draw[densely dashed, ultra thick,color=orange] (3,0) -- (2,0.8);
\draw[densely dashed, ultra thick,color=orange] (3,0) -- (2.5,3.2);
\draw[densely dashed, ultra thick,color=orange] (2,0.8) .. controls (1.8,1.6) .. (2,2.4);
\draw[densely dashed, ultra thick,color=orange] (2,0.8) -- (2,1.6);
\draw[densely dashed, ultra thick,color=orange] (2,0.8) -- (3,1.6);
\draw[densely dashed, ultra thick,color=orange] (4,0.8) -- (4,1.6);
\draw[densely dashed, ultra thick,color=orange] (4,0.8) -- (5,0.8);
\draw[densely dashed, ultra thick,color=orange] (4,0.8) -- (4,0);
\draw[densely dashed, ultra thick,color=orange] (4,0.8) -- (5,2.4);
\draw[densely dashed, ultra thick,color=orange] (5,0.8) -- (4.5,3.2);
\draw[densely dashed, ultra thick,color=orange] (5,0.8) -- (4,2.4);
\draw[densely dashed, ultra thick,color=orange] (5,0.8) -- (5,1.6);
\draw[densely dashed, ultra thick,color=orange] (7,0.8) -- (6,0);
\draw[densely dashed, ultra thick,color=orange] (7,0.8) -- (6,1.6);
\draw[densely dashed, ultra thick,color=orange] (7,0.8) -- (7,1.6);
\draw[densely dashed, ultra thick,color=orange] (7,0.8) -- (6.5,3.2);
\draw[densely dashed, ultra thick,color=orange] (7,0.8) .. controls (7.2,1.6) .. (7,2.4);
\draw[densely dashed, ultra thick,color=orange] (8,2.4) -- (9,0);
\draw[densely dashed, ultra thick,color=orange] (8,2.4) -- (9,1.6);
\draw[densely dashed, ultra thick,color=orange] (8,2.4) -- (8,1.6);
\draw[densely dashed, ultra thick,color=orange] (9,0) -- (8,0);
\draw[densely dashed, ultra thick,color=orange] (9,0) -- (9,0.8);
\draw[densely dashed, ultra thick,color=orange] (9,0) -- (8.5,3.2);
\draw[densely dashed, ultra thick,color=orange] (8,2.4) .. controls (7.8,1.6) .. (8,0.8);
\draw[densely dashed, ultra thick,color=orange] (1,2.4) .. controls (2,2.8) .. (3,2.4);
\draw[densely dashed, ultra thick,color=orange] (0,0.8) .. controls (1,1.2) .. (2,0.8);
\draw[densely dashed, ultra thick,color=orange] (1,0) .. controls (2,0.4) .. (3,0);
\draw[densely dashed, ultra thick,color=orange] (3,0) .. controls (4,0.4) .. (5,0);
\draw[densely dashed, ultra thick,color=orange] (3,0.8) .. controls (4,1.2) .. (5,0.8);
\draw[densely dashed, ultra thick,color=orange] (2,0.8) .. controls (3,1.2) .. (4,0.8);
\draw[densely dashed, ultra thick,color=orange] (4,0.8) .. controls (5,1.2) .. (6,0.8);
\draw[densely dashed, ultra thick,color=orange] (5,0.8) .. controls (6,1.2) .. (7,0.8);
\draw[densely dashed, ultra thick,color=orange] (7,0) .. controls (8,0.4) .. (9,0);
\draw[densely dashed, ultra thick,color=orange] (6,2.4) .. controls (7,2.8) .. (8,2.4);
\draw[densely dashed, ultra thick,color=orange] (9,0) .. controls (11,0.8) and (-1,0.8) .. (1,0);
\draw[densely dashed, ultra thick,color=orange] (9,2.4) .. controls (11,3.2) and (-1,3.2) .. (1,2.4);
\draw (0.5,3.2) -- (1,2.4);
\draw (0.5,3.2) -- (0,1.6);
\draw (0.5,3.2) -- (0,0);
\draw (0.5,3.2) -- (0,0.8);
\draw (0.5,3.2) -- (1,0.8);
\draw (1,0.8) -- (1,1.6);
\draw (1,0.8) -- (1,0);
\draw (1,0.8) -- (0,2.4);
\draw (0,0.8) -- (0.5,3.2);
\draw (3,0.8) -- (2,2.4);
\draw (3,0.8) .. controls (3.2,1.6) .. (3,2.4);
\draw (3,0.8) -- (2,0.8);
\draw (3,0.8) -- (2,0);
\draw (3,0.8) -- (3,0);
\draw (4,1.6) -- (5,1.6);
\draw (4,1.6) -- (4,2.4);
\draw (4,1.6) -- (4.5,3.2);
\draw (4,1.6) -- (5,0.8);
\draw (4,1.6) -- (5,0);
\draw (5,1.6) -- (5,2.4);
\draw (5,1.6) -- (4,0.8);
\draw (5,1.6) -- (4,0);
\draw (6,0) -- (6,0.8);
\draw (6,0) -- (7,0);
\draw (6,2.4) .. controls (5.6,1.2) .. (6,0);
\draw (6,0) -- (7,1.6);
\draw (7,1.6) -- (7,2.4);
\draw (7,1.6) -- (6,1.6);
\draw (7,1.6) -- (6.5,3.2);
\draw (9,0.8) -- (9,1.6);
\draw (9,0.8) -- (8,0.8);
\draw (9,0.8) -- (8,1.6);
\draw (9,0.8) -- (8,2.4);
\draw (9,1.6) -- (9,2.4);
\draw (9,1.6) -- (8.5,3.2);
\draw (9,1.6) .. controls (9.2,0.8) .. (9,0);
\draw (0.5,3.2) .. controls (1.5,3.6) .. (2.5,3.2);
\draw (1,0.8) .. controls (2,1.2) .. (3,0.8);
\draw (2,1.6) .. controls (3,2) .. (4,1.6);
\draw (3,1.6) .. controls (4,2) .. (5,1.6);
\draw (5,1.6) .. controls (6,2) .. (7,1.6);
\draw (7,1.6) .. controls (8,2) .. (9,1.6);
\draw (6,0) .. controls (7,0.4) .. (8,0);
\draw (7,0.8) .. controls (8,1.2) .. (9,0.8);
\draw (9,0.8) .. controls (11,1.6) and (-1,1.6) .. (1,0.8);
\draw[dashed, draw=red] (0,0.8) -- (1,0);
\draw[dashed, draw=red] (0,2.4) .. controls (-0.2,1.6) .. (0,0.8);
\draw[dashed, draw=red] (0,0.8) -- (1,0.8);
\draw[dashed, draw=red] (0,0.8) -- (1,1.6);
\draw[dashed, draw=red] (0,0.8) -- (1,2.4);
\draw[dashed, draw=red] (2,0) -- (2,0.8);
\draw[dashed, draw=red] (2,0) -- (3,2.4);
\draw[dashed, draw=red] (2,0) -- (2.5,3.2);
\draw[dashed, draw=red] (2,0) ..controls (1.8,0.8) .. (2,1.6);
\draw[dashed, draw=red] (2,1.6) -- (3,1.6);
\draw[dashed, draw=red] (2,1.6) -- (3,0.8);
\draw[dashed, draw=red] (2,1.6) -- (3,0);
\draw[dashed, draw=red] (2,1.6) -- (2,2.4);
\draw[dashed, draw=red] (4,0) .. controls (3.8,0.8).. (4,1.6);
\draw[dashed, draw=red] (4.5,3.2) -- (4,0);
\draw[dashed, draw=red] (4.5,3.2) -- (5,0);
\draw[dashed, draw=red] (4.5,3.2) -- (4,2.4);
\draw[dashed, draw=red] (4.5,3.2) -- (4,0.8);
\draw[dashed, draw=red] (4.5,3.2) -- (5,1.6);
\draw[dashed, draw=red] (4,0) -- (5,2.4);
\draw[dashed, draw=red] (4,0) -- (5,0.8);
\draw[dashed, draw=red] (6.5,3.2) -- (6,2.4);
\draw[dashed, draw=red] (6.5,3.2) -- (6,0.8);
\draw[dashed, draw=red] (6.5,3.2) -- (7,0);
\draw[dashed, draw=red] (6,2.4) -- (7,0.8);
\draw[dashed, draw=red] (6,2.4) -- (7,1.6);
\draw[dashed, draw=red] (6,2.4) -- (6,1.6);
\draw[dashed, draw=red] (6,2.4) -- (7,2.4);
\draw[dashed, draw=red] (8.5,3.2) -- (8,2.4);
\draw[dashed, draw=red] (8.5,3.2) -- (8,0.8);
\draw[dashed, draw=red] (8.5,3.2) -- (9,0.8);
\draw[dashed, draw=red] (8.5,3.2) -- (8,0);
\draw[dashed, draw=red] (8,0.8) -- (9,0);
\draw[dashed, draw=red] (8,0.8) -- (9,1.6);
\draw[dashed, draw=red] (8,0.8) -- (9,2.4);
\draw[dashed, draw=red] (8,0.8) -- (8,1.6);
\draw[dashed, draw=red] (0,1.6) .. controls (1,2) .. (2,1.6);
\draw[dashed, draw=red] (0,0) .. controls (1,0.4) .. (2,0);
\draw[dashed, draw=red] (2,0) .. controls (3,0.4) .. (4,0);
\draw[dashed, draw=red] (4,0) .. controls (5,0.4) .. (6,0);
\draw[dashed, draw=red] (4.5,3.2) .. controls (5.5,3.6) .. (6.5,3.2);
\draw[dashed, draw=red] (6.5,3.2) .. controls (7.5,3.6) .. (8.5,3.2);
\draw[dashed, draw=red] (8,0.8) .. controls (10,1.6) and (-2,1.6) .. (0,0.8);
\draw[dashed, draw=red] (8.5,3.2) .. controls (10.5,4) and (-1.5,4) .. (0.5,3.2);
\draw[dotted, thick,color=green] (0,1.6) -- (0,2.4);
\draw[dotted, thick,color=green] (0,1.6) -- (0,0.8);
\draw[dotted, thick,color=green] (0,1.6) -- (1,2.4);
\draw[dotted, thick,color=green] (0,1.6) -- (1,0.8);
\draw[dotted, thick,color=green] (0,1.6) .. controls (-0.2,0.8) .. (0,0);
\draw[dotted, thick,color=green] (0,2.4) -- (0.5,3.2);
\draw[dotted, thick,color=green] (0,2.4) -- (1,0);
\draw[dotted, thick,color=green] (0,2.4) -- (1,1.6);
\draw[dotted, thick,color=green] (2,2.4) -- (2.5,3.2);
\draw[dotted, thick,color=green] (2,2.4) -- (3,1.6);
\draw[dotted, thick,color=green] (2,2.4) -- (3,0);
\draw[dotted, thick,color=green] (2,2.4) .. controls (1.6,1.2) .. (2,0);
\draw[dotted, thick,color=green] (2.5,3.2) -- (3,2.4);
\draw[dotted, thick,color=green] (2.5,3.2) -- (2,1.6);
\draw[dotted, thick,color=green] (2.5,3.2) -- (2,0.8);
\draw[dotted, thick,color=green] (2.5,3.2) -- (3,0.8);
\draw[dotted, thick,color=green] (5,0) -- (4,0);
\draw[dotted, thick,color=green] (5,0) -- (4,0.8);
\draw[dotted, thick,color=green] (5,0) .. controls (5.2,0.8) .. (5,1.6);
\draw[dotted, thick,color=green] (5,0) -- (5,0.8);
\draw[dotted, thick,color=green] (5,0) .. controls (5.4,1.2) .. (5,2.4);
\draw[dotted, thick,color=green] (7,0) -- (7,0.8);
\draw[dotted, thick,color=green] (7,0) -- (6,1.6);
\draw[dotted, thick,color=green] (7,0) -- (6,2.4);
\draw[dotted, thick,color=green] (6,1.6) -- (6.5,3.2);
\draw[dotted, thick,color=green] (6,1.6) -- (6,0.8);
\draw[dotted, thick,color=green] (7,0) .. controls (7.4,1.2) .. (7,2.4);
\draw[dotted, thick,color=green] (7,0) .. controls (7.2,0.8) .. (7,1.6);
\draw[dotted, thick,color=green] (6,0) .. controls (5.8,0.8) .. (6,1.6);
\draw[dotted, thick,color=green] (8,0) -- (9,2.4);
\draw[dotted, thick,color=green] (8,0) -- (9,0.8);
\draw[dotted, thick,color=green] (8,0) -- (8,0.8);
\draw[dotted, thick,color=green] (8,0) .. controls (7.8,0.8) .. (8,1.6);
\draw[dotted, thick,color=green] (8,1.6) -- (9,1.6);
\draw[dotted, thick,color=green] (8,1.6) -- (9,0);
\draw[dotted, thick,color=green] (8,1.6) -- (8.5,3.2);
\draw[dotted, thick,color=green] (0,2.4) .. controls (1,2.8) .. (2,2.4);
\draw[dotted, thick,color=green] (2,2.4) .. controls (3,2.8) .. (4,2.4);
\draw[dotted, thick,color=green] (2.5,3.2) .. controls (3.5,3.6) .. (4.5,3.2);
\draw[dotted, thick,color=green] (4,1.6) .. controls (5,2) .. (6,1.6);
\draw[dotted, thick,color=green] (5,0) .. controls (6,0.4) .. (7,0);
\draw[dotted, thick,color=green] (6,1.6) .. controls (7,2) .. (8,1.6);
\draw[dotted, thick,color=green] (8,2.4) .. controls (10,3.2) and (-2,3.2) .. (0,2.4);
\draw[dotted, thick,color=green] (8,1.6) .. controls (10,2.4) and (-2,2.4) .. (0,1.6);
\node at (1,0) [diamond,draw=black,fill=orange,scale=0.7]{};
\node at (1,2.4) [diamond,draw=black,fill=orange,scale=0.7]{};
\node at (0,1.6) [regular polygon, regular polygon sides=5,draw=black,fill=green,scale=0.7]{};
\node at (0,2.4) [regular polygon, regular polygon sides=5,draw=black,fill=green,scale=0.7]{};
\node at (1,1.6) [regular polygon, regular polygon sides=4,draw=black,fill=blue,scale=0.7]{};
\node at (0,0) [regular polygon, regular polygon sides=4,draw=black,fill=blue,scale=0.7]{};
\node at (0,0.8) [regular polygon, regular polygon sides=3,draw=black,fill=red,scale=0.5]{};
\node at (1,0.8) [circle,draw=black,fill=black,scale=0.7]{};
\node at (0.5,3.2) [circle,draw=black,fill=black,scale=0.7]{};
\node at (0+2,0.8) [diamond,draw=black,fill=orange,scale=0.7]{};
\node at (1+2,0) [diamond,draw=black,fill=orange,scale=0.7]{};
\node at (0+2,2.4) [regular polygon, regular polygon sides=5,draw=black,fill=green,scale=0.7]{};
\node at (0.5+2,3.2) [regular polygon, regular polygon sides=5,draw=black,fill=green,scale=0.7]{};
\node at (1+2,1.6) [regular polygon, regular polygon sides=4,draw=black,fill=blue,scale=0.7]{};
\node at (1+2,2.4) [regular polygon, regular polygon sides=4,draw=black,fill=blue,scale=0.7]{};
\node at (0+2,0) [regular polygon, regular polygon sides=3,draw=black,fill=red,scale=0.5]{};
\node at (0+2,1.6) [regular polygon, regular polygon sides=3,draw=black,fill=red,scale=0.5]{};
\node at (1+2,0.8) [circle,draw=black,fill=black,scale=0.7]{};
\node at (0+4,0.8) [diamond,draw=black,fill=orange,scale=0.7]{};
\node at (1+4,0.8) [diamond,draw=black,fill=orange,scale=0.7]{};
\node at (1+4,0) [regular polygon, regular polygon sides=5,draw=black,fill=green,scale=0.7]{};
\node at (0+4,2.4) [regular polygon, regular polygon sides=4,draw=black,fill=blue,scale=0.7]{};
\node at (1+4,2.4) [regular polygon, regular polygon sides=4,draw=black,fill=blue,scale=0.7]{};
\node at (0+4,0) [regular polygon, regular polygon sides=3,draw=black,fill=red,scale=0.5]{};
\node at (0.5+4,3.2) [regular polygon, regular polygon sides=3,draw=black,fill=red,scale=0.5]{};
\node at (0+4,1.6) [circle,draw=black,fill=black,scale=0.7]{};
\node at (1+4,1.6) [circle,draw=black,fill=black,scale=0.7]{};
\node at (1+6,0.8) [diamond,draw=black,fill=orange,scale=0.7]{};
\node at (0+6,0.8) [regular polygon, regular polygon sides=4,draw=black,fill=blue,scale=0.7]{};
\node at (1+6,2.4) [regular polygon, regular polygon sides=4,draw=black,fill=blue,scale=0.7]{};
\node at (1+6,0) [regular polygon, regular polygon sides=5,draw=black,fill=green,scale=0.7]{};
\node at (0+6,1.6) [regular polygon, regular polygon sides=5,draw=black,fill=green,scale=0.7]{};
\node at (0+6,2.4) [regular polygon, regular polygon sides=3,draw=black,fill=red,scale=0.5]{};
\node at (0.5+6,3.2) [regular polygon, regular polygon sides=3,draw=black,fill=red,scale=0.5]{};
\node at (0+6,0) [circle,draw=black,fill=black,scale=0.7]{};
\node at (1+6,1.6) [circle,draw=black,fill=black,scale=0.7]{};
\node at (1+8,0) [diamond,draw=black,fill=orange,scale=0.7]{};
\node at (0+8,2.4) [diamond,draw=black,fill=orange,scale=0.7]{};
\node at (1+8,2.4) [regular polygon, regular polygon sides=4,draw=black,fill=blue,scale=0.7]{};
\node at (0+8,0) [regular polygon, regular polygon sides=5,draw=black,fill=green,scale=0.7]{};
\node at (0+8,1.6) [regular polygon, regular polygon sides=5,draw=black,fill=green,scale=0.7]{};
\node at (0.5+8,3.2) [regular polygon, regular polygon sides=3,draw=black,fill=red,scale=0.5]{};
\node at (0+8,0.8) [regular polygon, regular polygon sides=3,draw=black,fill=red,scale=0.5]{};
\node at (1+8,0.8) [circle,draw=black,fill=black,scale=0.7]{};
\node at (1+8,1.6) [circle,draw=black,fill=black,scale=0.7]{};
\end{tikzpicture}
\end{center}
\caption{Five completely independent spanning trees in $K_{9}\square C_5$.}
\label{K95}
\end{sidewaysfigure}
\begin{proof}
The five completely independent spanning trees in $K_{9}\square C_5$ are depicted in Figure~\ref{K95} and their edge sets are given in Appendix \ref{dek95}.
\end{proof}

We end this section with a theorem summarizing the results for $K_{m}\square C_n$.
Given a graph $G$, let $\text{mcist}(G)$ be the maximum integer $k$ such that there exist $k$ completely independent spanning trees in $G$.
\begin{theo}
Let $m\ge3$ and $n\ge3$ be integers. We have:\newline
$\text{mcist}(K_{m}\square C_n)=\left\{
    \begin{array}{ll}
        \lceil m/2\rceil,  & \mbox{if $(m=3,5\lor (m=7\land n=3,4) \lor (m=9\land n=4,5))$;} \\
        \lfloor m/2\rfloor,  & \mbox{otherwise.}
    \end{array}
\right.$
\end{theo}
\begin{proof}
For every even $m$, by Corollary \ref{existr-1}, there exist $m/2$ completely independent spanning trees.
Suppose $m$ is odd. 
For $m=3$, Hasunuma and Morisaka \cite{HA2012} has proven that in any Cartesian product of 2-connected graphs, there are two completely independent spanning trees.
By Propositions \ref{propk94}, \ref{propk95}, \ref{propk73}, \ref{propk74} and \ref{propk5n}, we obtain that there exist
$\lceil m/2\rceil$ completely independent spanning trees for $m=5$ or ($m=7\land n=3,4$) or ($m=9\land n=4,5$).

In the other cases, by Propositions \ref{krgrand}, \ref{propk45}, \ref{propk93}, there do not exist $\lceil m/2\rceil$ completely independent spanning trees in these graphs.
By Corollary \ref{existr-1}, there exist $\lfloor m/2\rfloor$ completely independent spanning trees in $K_{m-1}\square C_n$.
From these $\lfloor m/2\rfloor$ completely independent spanning trees in $K_{m-1}\square C_n$, we can construct $\lfloor m/2\rfloor$ completely independent spanning trees in $K_{m}\square C_n$.
The graph $K_{m}\square C_n$ contains $n$ vertices $u_0,\ldots, u_{n-1}$ not in $K_{m-1}\square C_n$, with $u_j\in V(K^{j})$ for $j=0,\ldots,n-1$.
For each $1\le i \le \lfloor m/2\rfloor$, it suffices to add an edge between $u_j$, $1\le j\le n$, and a vertex of $V_j(T_i)$ to obtain $\lfloor m/2\rfloor$ completely independent spanning trees in $K_{m}\square C_n$.
\end{proof}

\section{3-dimensional toroidal grids}
\label{sec:3Dgrid}
Hasunuma and Morisaka~\cite{HA2012} have shown that there are two completely independent spanning trees in any 2-dimensional toroidal grid and left as an open problem the question of whether there are $n$ completely independent spanning trees in any $n$-dimensional toroidal grid, for $n\ge 3$. In this section we give a partial answer for $n=3$ by finding three completely independent spanning trees in some 3-dimensional toroidal grids.

Let $n_1$, $n_2$ and $n_3$ be positive integers, $3\le n_1 \le n_2\le n_3$. The 3-dimensional toroidal grid $TM(n_1,n_2,n_3)$ is the Cartesian product of three cycles: $C_{n_1}\square C_{n_2}\square C_{n_3}$.
We let $V(TM(n_1,n_2,n_3))=\{(i,j,k)| 0\le i<n_1, 0\le j<n_2, 0\le k <n_3\}$ and $E(TM(n_1,n_2,n_3))=\{(i,j,k)\ (i',j',k')|i\equiv i'\pm1\pmod{n_1}, j=j',k=k'\lor i=i', j\equiv j'\pm1\pmod{n_2}, k=k'\lor i=i', j=j',k\equiv k'\pm1\pmod{n_3}\}$. 
In the remainder of the section, the integers $i,j$ and $k$ in a vertex $(i,j,k)$ are considered modulo $n_1$, $n_2$ and $n_3$, respectively.

By a {\em level} of $TM(3,3,q)$ we mean a subgraph of it induced by the vertices with the same third coordinate.

\begin{prop}
Let $p$, $p'$ and $q$ be positive integers such that $gcd(p,p',q)=1$. There exist three completely independent spanning trees in $TM(3p,3p',3q)$.
\label{tor333}
\end{prop}
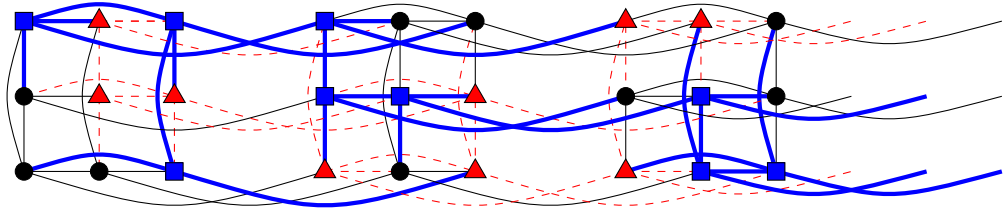
\begin{figure}[t]
\begin{center}
\begin{tikzpicture}
\draw (0,0) -- (1,0);
\draw (0,0) -- (0,1);
\draw (0,0) .. controls (-0.3,1) .. (0,2);
\draw (1,0) -- (2,0);
\draw (1,0) .. controls (0.7,1) .. (1,2);
\draw (0,1) -- (1,1);
\draw[style=dashed,color=red] (1,1) -- (2,1);
\draw[style=dashed,color=red] (1,1) -- (1,2);
\draw[style=dashed,color=red] (1,1) -- (1,0);
\draw[style=dashed,color=red] (2,1) -- (2,0);
\draw[style=dashed,color=red] (1,2) -- (2,2);
\draw[style=dashed,color=red] (1,2) -- (2,2);
\draw[style=dashed,color=red] (0,1) .. controls (1,1.3) .. (2,1);
\draw[ultra thick,color=blue] (2,2) .. controls (1.7,1) .. (2,0);
\draw[ultra thick,color=blue] (2,2) .. controls (1,2.3) .. (0,2);
\draw[ultra thick,color=blue] (0,0) .. controls (1,0.3) .. (2,0);
\draw[ultra thick,color=blue] (0,2) -- (0,1);
\draw[ultra thick,color=blue] (2,2) -- (2,1);
\draw[ultra thick,color=blue] (0,2) -- (1,2);
\draw (5,2) -- (6,2);
\draw (5,0) .. controls (4.7,1).. (5,2);
\draw (5,0) -- (6,0);
\draw (5,2) -- (5,1);
\draw (6,2) -- (6,1);
\draw (4,2) .. controls (5,2.3) .. (6,2);
\draw[style=dashed,color=red] (6,0) -- (6,1);
\draw[style=dashed,color=red] (4,0) .. controls (5,0.3) .. (6,0);
\draw[style=dashed,color=red] (4,0) -- (5,0);
\draw[style=dashed,color=red] (4,1) .. controls (5,1.3) .. (6,1);
\draw[style=dashed,color=red] (4,0) .. controls (3.7,1).. (4,2);
\draw[style=dashed,color=red] (6,0) .. controls (5.7,1).. (6,2);;
\draw[ultra thick,color=blue] (4,1) -- (4,2);
\draw[ultra thick,color=blue] (4,1) -- (5,1);
\draw[ultra thick,color=blue] (4,0) -- (4,1);
\draw[ultra thick,color=blue] (5,0) -- (5,1);
\draw[ultra thick,color=blue] (5,1) -- (6,1);
\draw[ultra thick,color=blue] (4,2) -- (5,2);
\draw (10,1) -- (10,2);
\draw (10,1) .. controls (9,1.3).. (8,1);
\draw (8,1) -- (8,0);
\draw (10,1) -- (10,0);
\draw (8,1) -- (9,1);
\draw (10,2) .. controls (9,2.3).. (8,2);
\draw[style=dashed,color=red] (8,2) -- (9,2);
\draw[style=dashed,color=red] (8,1) -- (8,2);
\draw[style=dashed,color=red] (9,1) -- (9,2);
\draw[style=dashed,color=red] (8,0) -- (9,0);
\draw[style=dashed,color=red] (9,2) -- (10,2);
\draw[style=dashed,color=red] (8,0) .. controls (7.7,1) .. (8,2);
\draw[ultra thick,color=blue] (9,0) -- (10,0);
\draw[ultra thick,color=blue] (9,0) -- (9,1);
\draw[ultra thick,color=blue] (9,1) -- (10,1);
\draw[ultra thick,color=blue] (9,0) .. controls (8.7,1) .. (9,2);
\draw[ultra thick,color=blue] (10,0) .. controls (9.7,1) .. (10,2);
\draw[ultra thick,color=blue] (10,0) .. controls (9,0.3).. (8,0);
\draw (0,0) .. controls (2,-0.6) .. (4,0);
\draw (1,0) .. controls (3,-0.6) .. (5,0);
\draw (0,1) .. controls (2,0.4) .. (4,1);
\draw[style=dashed,color=red] (1,1) .. controls (3,0.4) .. (5,1);
\draw[style=dashed,color=red] (2,1) .. controls (4,0.4) .. (6,1);
\draw[style=dashed,color=red] (1,2) .. controls (3,1.4) .. (5,2);
\draw[ultra thick,color=blue] (0,2) .. controls (2,1.4) .. (4,2);
\draw[ultra thick,color=blue] (2,2) .. controls (4,1.4) .. (6,2);
\draw[ultra thick,color=blue] (2,0) .. controls (4,-0.6) .. (6,0);
\draw (5,0) .. controls (7,-0.6) .. (9,0);
\draw (5,2) .. controls (7,1.4) .. (9,2);
\draw (6,2) .. controls (8,1.4) .. (10,2);
\draw[style=dashed,color=red] (4,0) .. controls (6,-0.6) .. (8,0);
\draw[style=dashed,color=red] (6,0) .. controls (8,-0.6) .. (10,0);
\draw[style=dashed,color=red] (6,1) .. controls (8,0.4) .. (10,1);
\draw[ultra thick,color=blue] (4,1) .. controls (6,0.4) .. (8,1);
\draw[ultra thick,color=blue] (5,1) .. controls (7,0.4) .. (9,1);
\draw[ultra thick,color=blue] (4,2) .. controls (6,1.4) .. (8,2);
\draw (8,1) .. controls (9.5,0.6) .. (11,1);
\draw (10,1) .. controls (11.5,0.6) .. (13,1);
\draw (10,2) .. controls (11.5,1.6) .. (13,2);
\draw[style=dashed,color=red] (8,0) .. controls (9.5,-0.4) .. (11,0);
\draw[style=dashed,color=red] (8,2) .. controls (9.5,1.6) .. (11,2);
\draw[style=dashed,color=red] (9,2) .. controls (10.5,1.6) .. (12,2);
\draw[ultra thick,color=blue] (9,0) .. controls (10.5,-0.4) .. (12,0);
\draw[ultra thick,color=blue] (9,1) .. controls (10.5,0.6) .. (12,1);
\draw[ultra thick,color=blue] (10,0) .. controls (11.5,-0.4) .. (13,0);

\node at (0,2) [regular polygon, regular polygon sides=4,draw=black,fill=blue,scale=0.7]{};
\node at (0,1) [circle,draw=black,fill=black,scale=0.7]{};
\node at (0,0) [circle,draw=black,fill=black,scale=0.7]{};
\node at (1,2) [regular polygon, regular polygon sides=3,draw=black,fill=red,scale=0.5]{};
\node at (1,1) [regular polygon, regular polygon sides=3,draw=black,fill=red,scale=0.5]{};
\node at (1,0) [circle,draw=black,fill=black,scale=0.7]{};
\node at (2,2) [regular polygon, regular polygon sides=4,draw=black,fill=blue,scale=0.7]{};
\node at (2,1) [regular polygon, regular polygon sides=3,draw=black,fill=red,scale=0.5]{};
\node at (2,0) [regular polygon, regular polygon sides=4,draw=black,fill=blue,scale=0.7]{};
\node at (4,2) [regular polygon, regular polygon sides=4,draw=black,fill=blue,scale=0.7]{};
\node at (5,2) [circle,draw=black,fill=black,scale=0.7]{};
\node at (5,0) [circle,draw=black,fill=black,scale=0.7]{};
\node at (4,0) [regular polygon, regular polygon sides=3,draw=black,fill=red,scale=0.5]{};
\node at (6,0) [regular polygon, regular polygon sides=3,draw=black,fill=red,scale=0.5]{};
\node at (6,2) [circle,draw=black,fill=black,scale=0.7]{};
\node at (4,1) [regular polygon, regular polygon sides=4,draw=black,fill=blue,scale=0.7]{};
\node at (6,1) [regular polygon, regular polygon sides=3,draw=black,fill=red,scale=0.5]{};
\node at (5,1) [regular polygon, regular polygon sides=4,draw=black,fill=blue,scale=0.7]{};
\node at (9,0) [regular polygon, regular polygon sides=4,draw=black,fill=blue,scale=0.7]{};
\node at (8,1) [circle,draw=black,fill=black,scale=0.7]{};
\node at (10,1) [circle,draw=black,fill=black,scale=0.7]{};
\node at (8,0) [regular polygon, regular polygon sides=3,draw=black,fill=red,scale=0.5]{};
\node at (8,2) [regular polygon, regular polygon sides=3,draw=black,fill=red,scale=0.5]{};
\node at (10,2) [circle,draw=black,fill=black,scale=0.7]{};
\node at (9,1) [regular polygon, regular polygon sides=4,draw=black,fill=blue,scale=0.7]{};
\node at (9,2) [regular polygon, regular polygon sides=3,draw=black,fill=red,scale=0.5]{};
\node at (10,0) [regular polygon, regular polygon sides=4,draw=black,fill=blue,scale=0.7]{};
\end{tikzpicture}
\end{center}
\caption{The pattern for the three completely independent spanning trees of $TM(3,3,3q)$, with $q\ge 2$.}
\label{tor31}
\end{figure}
\begin{proof}
We define three completely independent spanning trees $T_1$, $T_2$ and $T_3$ in $TM(3p,3p',3q)$ as follows:
for $j\in\{0,1, 2\}$,\newline
$E(T_{j-1})=\{(i+j,j-i,i) (1+i+j,-i+j,i),(i+j,j-i,i)(i+j,1-i+j,i), 
(i+j,j-i,i) (i+j,j-i,1+i),(i+j,j-i,i) (i+j,-1-i+j,i),$
$(1+i+j,j-i,i) (2+i+j,j-i,i),(1+i+j,j-i,i) (1+i+j,j-i,1+i),(1+i+j,j-i,i) (1+i,j-i-1,i),$
$(i+j,1-i+j,i) (1+i+j,1-i+j,i),(i+j,1-i+j,i) (i+j,1-i+j,1+i)|i\in\{0,\ldots,pp' q-1 \}-(j,j+1,0)(j,j+1,-1)$.\newline
We require $gcd(p,p',q)=1$, in order that $T_1$, $T_2$, $T_3$ contain every vertex of $TM(3p,3p',3q)$, i.e. every edge is different for each value of $i$, $0\le i \le pp' q-1$.
Figure \ref{tor31} describes the pattern on three levels for these three spanning trees for $p=1$ and $p'=1$.

\end{proof}

\begin{prop}
For any integer $q\ge 3$, there exists three completely independent spanning trees in $TM(3,3,q)$.
\end{prop}
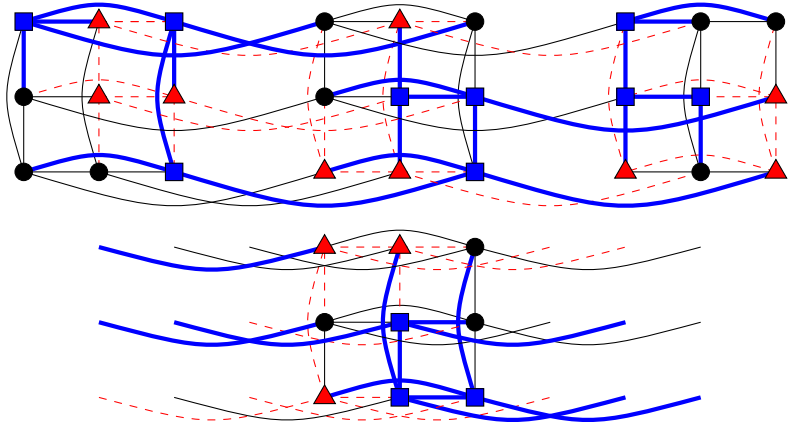
\begin{figure}[t]
\begin{center}
\begin{tikzpicture}
\draw (0,0) -- (1,0);
\draw (0,0) -- (0,1);
\draw (0,0) .. controls (-0.3,1) .. (0,2);
\draw (1,0) -- (2,0);
\draw (1,0) .. controls (0.7,1) .. (1,2);
\draw (0,1) -- (1,1);
\draw[style=dashed,color=red] (1,1) -- (2,1);
\draw[style=dashed,color=red] (1,1) -- (1,2);
\draw[style=dashed,color=red] (1,1) -- (1,0);
\draw[style=dashed,color=red] (2,1) -- (2,0);
\draw[style=dashed,color=red] (1,2) -- (2,2);
\draw[style=dashed,color=red] (1,2) -- (2,2);
\draw[style=dashed,color=red] (0,1) .. controls (1,1.3) .. (2,1);
\draw[ultra thick,color=blue] (2,2) .. controls (1.7,1) .. (2,0);
\draw[ultra thick,color=blue] (2,2) .. controls (1,2.3) .. (0,2);
\draw[ultra thick,color=blue] (0,0) .. controls (1,0.3) .. (2,0);
\draw[ultra thick,color=blue] (2,2) -- (2,1);
\draw[ultra thick,color=blue] (0,2) -- (0,1);
\draw[ultra thick,color=blue] (0,2) -- (1,2);
\draw (4,1) -- (4,2);
\draw (4,2) .. controls (5,2.3).. (6,2);
\draw (4,1) -- (5,1);
\draw (4,2) -- (5,2);
\draw (6,2) -- (6,1);
\draw (6,0) .. controls (5.7,1) .. (6,2);
\draw[style=dashed,color=red] (4,0) -- (5,0);
\draw[style=dashed,color=red] (5,0) .. controls (4.7,1).. (5,2);
\draw[style=dashed,color=red] (5,0) -- (6,0);
\draw[style=dashed,color=red] (5,2) -- (6,2);
\draw[style=dashed,color=red] (4,0) -- (4,1);
\draw[style=dashed,color=red] (4,0) .. controls (3.7,1).. (4,2);
\draw[ultra thick,color=blue] (6,0) -- (6,1);
\draw[ultra thick,color=blue] (5,1) -- (6,1);
\draw[ultra thick,color=blue] (5,0) -- (5,1);
\draw[ultra thick,color=blue] (5,1) -- (5,2);
\draw[ultra thick,color=blue] (4,0) .. controls (5,0.3).. (6,0);
\draw[ultra thick,color=blue] (4,1) .. controls (5,1.3).. (6,1);
\draw (9,2) -- (10,2);
\draw (9,0) .. controls (8.7,1).. (9,2);
\draw (9,0) -- (8,0);
\draw (9,0) -- (10,0);
\draw (9,1) -- (9,2);
\draw (10,1) -- (10,2);
\draw[style=dashed,color=red] (10,0) .. controls (9,0.3).. (8,0);
\draw[style=dashed,color=red] (10,0) -- (10,1);
\draw[style=dashed,color=red] (9,1) -- (10,1);
\draw[style=dashed,color=red] (10,1) .. controls (9,1.3).. (8,1);
\draw[style=dashed,color=red] (10,0) .. controls (9.7,1) .. (10,2);
\draw[style=dashed,color=red] (8,0) .. controls (7.7,1) .. (8,2);
\draw[ultra thick,color=blue] (8,1) -- (9,1);
\draw[ultra thick,color=blue] (8,1) -- (8,2);
\draw[ultra thick,color=blue] (8,0) -- (8,1);
\draw[ultra thick,color=blue] (9,0) -- (9,1);
\draw[ultra thick,color=blue] (8,2) -- (9,2);
\draw[ultra thick,color=blue] (10,2) .. controls (9,2.3).. (8,2);
\draw (0,0) .. controls (2,-0.6) .. (4,0);
\draw (1,0) .. controls (3,-0.6) .. (5,0);
\draw (0,1) .. controls (2,0.4) .. (4,1);
\draw[style=dashed,color=red] (1,1) .. controls (3,0.4) .. (5,1);
\draw[style=dashed,color=red] (2,1) .. controls (4,0.4) .. (6,1);
\draw[style=dashed,color=red] (1,2) .. controls (3,1.4) .. (5,2);
\draw[ultra thick,color=blue] (0,2) .. controls (2,1.4) .. (4,2);
\draw[ultra thick,color=blue] (2,2) .. controls (4,1.4) .. (6,2);
\draw[ultra thick,color=blue] (2,0) .. controls (4,-0.6) .. (6,0);
\draw (4,1) .. controls (6,0.4) .. (8,1);
\draw (4,2) .. controls (6,1.4) .. (8,2);
\draw[style=dashed,color=red] (5,0) .. controls (7,-0.6) .. (9,0);
\draw[style=dashed,color=red] (5,2) .. controls (7,1.4) .. (9,2);
\draw[ultra thick,color=blue] (6,0) .. controls (8,-0.6) .. (10,0);
\draw[ultra thick,color=blue] (6,1) .. controls (8,0.4) .. (10,1);

\node at (0,2) [regular polygon, regular polygon sides=4,draw=black,fill=blue,scale=0.7]{};
\node at (0,1) [circle,draw=black,fill=black,scale=0.7]{};
\node at (0,0) [circle,draw=black,fill=black,scale=0.7]{};
\node at (1,2) [regular polygon, regular polygon sides=3,draw=black,fill=red,scale=0.5]{};
\node at (1,1) [regular polygon, regular polygon sides=3,draw=black,fill=red,scale=0.5]{};
\node at (1,0) [circle,draw=black,fill=black,scale=0.7]{};
\node at (2,2) [regular polygon, regular polygon sides=4,draw=black,fill=blue,scale=0.7]{};
\node at (2,1) [regular polygon, regular polygon sides=3,draw=black,fill=red,scale=0.5]{};
\node at (2,0) [regular polygon, regular polygon sides=4,draw=black,fill=blue,scale=0.7]{};
\node at (6,1) [regular polygon, regular polygon sides=4,draw=black,fill=blue,scale=0.7]{};
\node at (4,2) [circle,draw=black,fill=black,scale=0.7]{};
\node at (4,1) [circle,draw=black,fill=black,scale=0.7]{};
\node at (4,0) [regular polygon, regular polygon sides=3,draw=black,fill=red,scale=0.5]{};
\node at (5,0) [regular polygon, regular polygon sides=3,draw=black,fill=red,scale=0.5]{};
\node at (6,2) [circle,draw=black,fill=black,scale=0.7]{};
\node at (6,0) [regular polygon, regular polygon sides=4,draw=black,fill=blue,scale=0.7]{};
\node at (5,2) [regular polygon, regular polygon sides=3,draw=black,fill=red,scale=0.5]{};
\node at (5,1) [regular polygon, regular polygon sides=4,draw=black,fill=blue,scale=0.7]{};
\node at (8,1) [regular polygon, regular polygon sides=4,draw=black,fill=blue,scale=0.7]{};
\node at (9,0) [circle,draw=black,fill=black,scale=0.7]{};
\node at (9,2) [circle,draw=black,fill=black,scale=0.7]{};
\node at (8,0) [regular polygon, regular polygon sides=3,draw=black,fill=red,scale=0.5]{};
\node at (10,0) [regular polygon, regular polygon sides=3,draw=black,fill=red,scale=0.5]{};
\node at (10,2) [circle,draw=black,fill=black,scale=0.7]{};
\node at (9,1) [regular polygon, regular polygon sides=4,draw=black,fill=blue,scale=0.7]{};
\node at (10,1) [regular polygon, regular polygon sides=3,draw=black,fill=red,scale=0.5]{};
\node at (8,2) [regular polygon, regular polygon sides=4,draw=black,fill=blue,scale=0.7]{};
\draw (6,-2) -- (6,-1);
\draw (6,-2) .. controls (5,-1.7).. (4,-2);
\draw (4,-2) -- (4,-3);
\draw (6,-2) -- (6,-3);
\draw (4,-2) -- (5,-2);
\draw (6,-1) .. controls (5,-0.7).. (4,-1);
\draw[style=dashed,color=red] (4,-1) -- (5,-1);
\draw[style=dashed,color=red] (4,-2) -- (4,-1);
\draw[style=dashed,color=red] (5,-2) -- (5,-1);
\draw[style=dashed,color=red] (4,-3) -- (5,-3);
\draw[style=dashed,color=red] (5,-1) -- (6,-1);
\draw[style=dashed,color=red] (4,-3) .. controls (3.7,-2) .. (4,-1);
\draw[ultra thick,color=blue] (5,-3) -- (6,-3);
\draw[ultra thick,color=blue] (5,-3) -- (5,-2);
\draw[ultra thick,color=blue] (5,-2) -- (6,-2);
\draw[ultra thick,color=blue] (5,-3) .. controls (4.7,-2) .. (5,-1);
\draw[ultra thick,color=blue] (6,-3) .. controls (5.7,-2) .. (6,-1);
\draw[ultra thick,color=blue] (6,-3) .. controls (5,-2.7).. (4,-3);
\draw (4,-2) .. controls (5.5,-2.4) .. (7,-2);
\draw (6,-2) .. controls (7.5,-2.4) .. (9,-2);
\draw (6,-1) .. controls (7.5,-1.4) .. (9,-1);
\draw[style=dashed,color=red] (4,-3) .. controls (5.5,-3.4) .. (7,-3);
\draw[style=dashed,color=red] (4,-1) .. controls (5.5,-1.4) .. (7,-1);
\draw[style=dashed,color=red] (5,-1) .. controls (6.5,-1.4) .. (8,-1);
\draw[ultra thick,color=blue] (5,-3) .. controls (6.5,-3.4) .. (8,-3);
\draw[ultra thick,color=blue] (5,-2) .. controls (6.5,-2.4) .. (8,-2);
\draw[ultra thick,color=blue] (6,-3) .. controls (7.5,-3.4) .. (9,-3);
\draw (2,-3) .. controls (3.5,-3.4) .. (5,-3);
\draw (2,-1) .. controls (3.5,-1.4) .. (5,-1);
\draw (3,-1) .. controls (4.5,-1.4) .. (6,-1);
\draw[style=dashed,color=red] (1,-3) .. controls (2.5,-3.4) .. (4,-3);
\draw[style=dashed,color=red] (3,-3) .. controls (4.5,-3.4) .. (6,-3);
\draw[style=dashed,color=red] (3,-2) .. controls (4.5,-2.4) .. (6,-2);
\draw[ultra thick,color=blue] (1,-2) .. controls (2.5,-2.4) .. (4,-2);
\draw[ultra thick,color=blue] (2,-2) .. controls (3.5,-2.4) .. (5,-2);
\draw[ultra thick,color=blue] (1,-1) .. controls (2.5,-1.4) .. (4,-1);
\node at (5,-3) [regular polygon, regular polygon sides=4,draw=black,fill=blue,scale=0.7]{};
\node at (4,-2) [circle,draw=black,fill=black,scale=0.7]{};
\node at (6,-2) [circle,draw=black,fill=black,scale=0.7]{};
\node at (4,-3) [regular polygon, regular polygon sides=3,draw=black,fill=red,scale=0.5]{};
\node at (4,-1) [regular polygon, regular polygon sides=3,draw=black,fill=red,scale=0.5]{};
\node at (6,-1) [circle,draw=black,fill=black,scale=0.7]{};
\node at (5,-2) [regular polygon, regular polygon sides=4,draw=black,fill=blue,scale=0.7]{};
\node at (5,-1) [regular polygon, regular polygon sides=3,draw=black,fill=red,scale=0.5]{};
\node at (6,-3) [regular polygon, regular polygon sides=4,draw=black,fill=blue,scale=0.7]{};
\end{tikzpicture}
\end{center}
\caption{The three completely independent spanning trees on the last four levels of $TM(3,3,q)$, for $q\equiv 1\pmod 3$ and $q>2$.}
\label{tor32}
\end{figure}
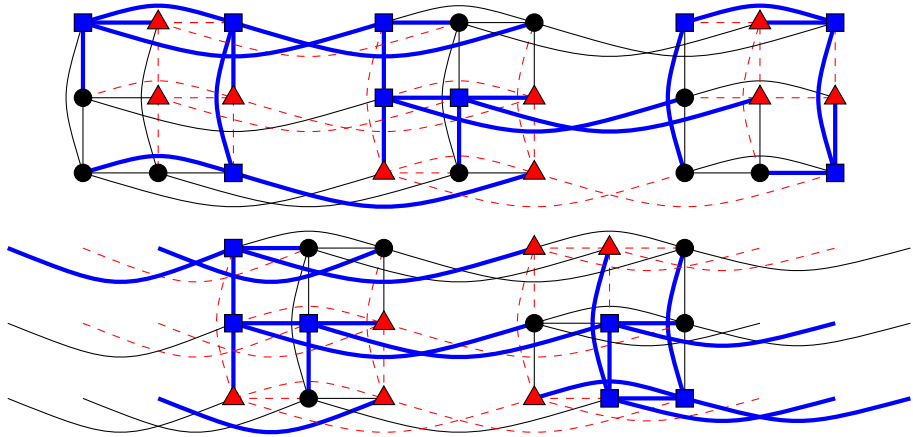
\begin{figure}[t]
\begin{center}
\begin{tikzpicture}
\draw (0,0) -- (1,0);
\draw (0,0) -- (0,1);
\draw (0,0) .. controls (-0.3,1) .. (0,2);
\draw (1,0) -- (2,0);
\draw (1,0) .. controls (0.7,1) .. (1,2);
\draw (0,1) -- (1,1);
\draw[style=dashed,color=red] (1,1) -- (2,1);
\draw[style=dashed,color=red] (1,1) -- (1,2);
\draw[style=dashed,color=red] (1,1) -- (1,0);
\draw[style=dashed,color=red] (2,1) -- (2,0);
\draw[style=dashed,color=red] (1,2) -- (2,2);
\draw[style=dashed,color=red] (1,2) -- (2,2);
\draw[style=dashed,color=red] (0,1) .. controls (1,1.3) .. (2,1);
\draw[ultra thick,color=blue] (2,2) .. controls (1.7,1) .. (2,0);
\draw[ultra thick,color=blue] (2,2) .. controls (1,2.3) .. (0,2);
\draw[ultra thick,color=blue] (0,0) .. controls (1,0.3) .. (2,0);
\draw[ultra thick,color=blue] (2,2) -- (2,1);
\draw[ultra thick,color=blue] (0,2) -- (0,1);
\draw[ultra thick,color=blue] (0,2) -- (1,2);
\draw (5,2) -- (6,2);
\draw (5,0) .. controls (4.7,1).. (5,2);
\draw (5,0) -- (6,0);
\draw (5,2) -- (5,1);
\draw (6,2) -- (6,1);
\draw (4,2) .. controls (5,2.3) .. (6,2);
\draw[style=dashed,color=red] (6,0) -- (6,1);
\draw[style=dashed,color=red] (4,0) .. controls (5,0.3) .. (6,0);
\draw[style=dashed,color=red] (4,0) -- (5,0);
\draw[style=dashed,color=red] (4,1) .. controls (5,1.3) .. (6,1);
\draw[style=dashed,color=red] (4,0) .. controls (3.7,1).. (4,2);
\draw[style=dashed,color=red] (6,0) .. controls (5.7,1).. (6,2);;
\draw[ultra thick,color=blue] (4,1) -- (4,2);
\draw[ultra thick,color=blue] (4,1) -- (5,1);
\draw[ultra thick,color=blue] (4,0) -- (4,1);
\draw[ultra thick,color=blue] (5,0) -- (5,1);
\draw[ultra thick,color=blue] (5,1) -- (6,1);
\draw[ultra thick,color=blue] (4,2) -- (5,2);
\draw (8,0) -- (9,0);
\draw (8,0) -- (8,1);
\draw (9,0) -- (9,1);
\draw (8,1) -- (8,2);
\draw (10,1) .. controls (9,1.3).. (8,1);
\draw (10,0) .. controls (9,0.3).. (8,0);
\draw[style=dashed,color=red] (9,1) -- (10,1);
\draw[style=dashed,color=red] (9,1) -- (9,2);
\draw[style=dashed,color=red] (8,1) -- (9,1);
\draw[style=dashed,color=red] (8,2) -- (9,2);
\draw[style=dashed,color=red] (10,1) -- (10,2);
\draw[style=dashed,color=red] (9,0) .. controls (8.7,1) .. (9,2);
\draw[ultra thick,color=blue] (9,0) -- (10,0);
\draw[ultra thick,color=blue] (9,2) -- (10,2);
\draw[ultra thick,color=blue] (10,0) -- (10,1);
\draw[ultra thick,color=blue] (8,0) .. controls (7.7,1) .. (8,2);
\draw[ultra thick,color=blue] (10,0) .. controls (9.7,1) .. (10,2);
\draw[ultra thick,color=blue] (10,2) .. controls (9,2.3).. (8,2);
\draw (0,0) .. controls (2,-0.6) .. (4,0);
\draw (1,0) .. controls (3,-0.6) .. (5,0);
\draw (0,1) .. controls (2,0.4) .. (4,1);
\draw[style=dashed,color=red] (1,1) .. controls (3,0.4) .. (5,1);
\draw[style=dashed,color=red] (2,1) .. controls (4,0.4) .. (6,1);
\draw[style=dashed,color=red] (1,2) .. controls (3,1.4) .. (5,2);
\draw[ultra thick,color=blue] (0,2) .. controls (2,1.4) .. (4,2);
\draw[ultra thick,color=blue] (2,2) .. controls (4,1.4) .. (6,2);
\draw[ultra thick,color=blue] (2,0) .. controls (4,-0.6) .. (6,0);
\draw (5,2) .. controls (7,1.4) .. (9,2);
\draw (6,2) .. controls (8,1.4) .. (10,2);
\draw[style=dashed,color=red] (4,0) .. controls (6,-0.6) .. (8,0);
\draw[style=dashed,color=red] (6,0) .. controls (8,-0.6) .. (10,0);
\draw[ultra thick,color=blue] (4,1) .. controls (6,0.4) .. (8,1);
\draw[ultra thick,color=blue] (5,1) .. controls (7,0.4) .. (9,1);
\node at (0,2) [regular polygon, regular polygon sides=4,draw=black,fill=blue,scale=0.7]{};
\node at (0,1) [circle,draw=black,fill=black,scale=0.7]{};
\node at (0,0) [circle,draw=black,fill=black,scale=0.7]{};
\node at (1,2) [regular polygon, regular polygon sides=3,draw=black,fill=red,scale=0.5]{};
\node at (1,1) [regular polygon, regular polygon sides=3,draw=black,fill=red,scale=0.5]{};
\node at (1,0) [circle,draw=black,fill=black,scale=0.7]{};
\node at (2,2) [regular polygon, regular polygon sides=4,draw=black,fill=blue,scale=0.7]{};
\node at (2,1) [regular polygon, regular polygon sides=3,draw=black,fill=red,scale=0.5]{};
\node at (2,0) [regular polygon, regular polygon sides=4,draw=black,fill=blue,scale=0.7]{};
\node at (4,1) [regular polygon, regular polygon sides=4,draw=black,fill=blue,scale=0.7]{};
\node at (5,2) [circle,draw=black,fill=black,scale=0.7]{};
\node at (5,0) [circle,draw=black,fill=black,scale=0.7]{};
\node at (4,0) [regular polygon, regular polygon sides=3,draw=black,fill=red,scale=0.5]{};
\node at (6,0) [regular polygon, regular polygon sides=3,draw=black,fill=red,scale=0.5]{};
\node at (6,2) [circle,draw=black,fill=black,scale=0.7]{};
\node at (4,2) [regular polygon, regular polygon sides=4,draw=black,fill=blue,scale=0.7]{};
\node at (6,1) [regular polygon, regular polygon sides=3,draw=black,fill=red,scale=0.5]{};
\node at (5,1) [regular polygon, regular polygon sides=4,draw=black,fill=blue,scale=0.7]{};
\node at (0+8,2) [regular polygon, regular polygon sides=4,draw=black,fill=blue,scale=0.7]{};
\node at (0+8,1) [circle,draw=black,fill=black,scale=0.7]{};
\node at (0+8,0) [circle,draw=black,fill=black,scale=0.7]{};
\node at (1+8,2) [regular polygon, regular polygon sides=3,draw=black,fill=red,scale=0.5]{};
\node at (1+8,1) [regular polygon, regular polygon sides=3,draw=black,fill=red,scale=0.5]{};
\node at (1+8,0) [circle,draw=black,fill=black,scale=0.7]{};
\node at (2+8,2) [regular polygon, regular polygon sides=4,draw=black,fill=blue,scale=0.7]{};
\node at (2+8,1) [regular polygon, regular polygon sides=3,draw=black,fill=red,scale=0.5]{};
\node at (2+8,0) [regular polygon, regular polygon sides=4,draw=black,fill=blue,scale=0.7]{};

\draw (5-2,2-3) -- (6-2,2-3);
\draw (5-2,0-3) .. controls (4.7-2,1-3).. (5-2,2-3);
\draw (5-2,0-3) -- (6-2,0-3);
\draw (5-2,2-3) -- (5-2,1-3);
\draw (6-2,2-3) -- (6-2,1-3);
\draw (4-2,2-3) .. controls (5-2,2.3-3) .. (6-2,2-3);
\draw[style=dashed,color=red] (6-2,0-3) -- (6-2,1-3);
\draw[style=dashed,color=red] (4-2,0-3) .. controls (5-2,0.3-3) .. (6-2,0-3);
\draw[style=dashed,color=red] (4-2,0-3) -- (5-2,0-3);
\draw[style=dashed,color=red] (4-2,1-3) .. controls (5-2,1.3-3) .. (6-2,1-3);
\draw[style=dashed,color=red] (4-2,0-3) .. controls (3.7-2,1-3).. (4-2,2-3);
\draw[style=dashed,color=red] (6-2,0-3) .. controls (5.7-2,1-3).. (6-2,2-3);
\draw[ultra thick,color=blue] (4-2,1-3) -- (4-2,2-3);
\draw[ultra thick,color=blue] (4-2,1-3) -- (5-2,1-3);
\draw[ultra thick,color=blue] (4-2,0-3) -- (4-2,1-3);
\draw[ultra thick,color=blue] (5-2,0-3) -- (5-2,1-3);
\draw[ultra thick,color=blue] (5-2,1-3) -- (6-2,1-3);
\draw[ultra thick,color=blue] (4-2,2-3) -- (5-2,2-3);
\draw (10-2,1-3) -- (10-2,2-3);
\draw (10-2,1-3) .. controls (9-2,1.3-3).. (8-2,1-3);
\draw (8-2,1-3) -- (8-2,0-3);
\draw (10-2,1-3) -- (10-2,0-3);
\draw (8-2,1-3) -- (9-2,1-3);
\draw (10-2,2-3) .. controls (9-2,2.3-3).. (8-2,2-3);
\draw[style=dashed,color=red] (8-2,2-3) -- (9-2,2-3);
\draw[style=dashed,color=red] (8-2,1-3) -- (8-2,2-3);
\draw[style=dashed,color=red] (9-2,1-3) -- (9-2,2-3);
\draw[style=dashed,color=red] (8-2,0-3) -- (9-2,0-3);
\draw[style=dashed,color=red] (9-2,2-3) -- (10-2,2-3);
\draw[style=dashed,color=red] (8-2,0-3) .. controls (7.7-2,1-3) .. (8-2,2-3);
\draw[ultra thick,color=blue] (9-2,0-3) -- (10-2,0-3);
\draw[ultra thick,color=blue] (9-2,0-3) -- (9-2,1-3);
\draw[ultra thick,color=blue] (9-2,1-3) -- (10-2,1-3);
\draw[ultra thick,color=blue] (9-2,0-3) .. controls (8.7-2,1-3) .. (9-2,2-3);
\draw[ultra thick,color=blue] (10-2,0-3) .. controls (9.7-2,1-3) .. (10-2,2-3);
\draw[ultra thick,color=blue] (10-2,0-3) .. controls (9-2,0.3-3).. (8-2,0-3);
\draw (0-1,0-3) .. controls (2-1.5,-0.6-3) .. (4-2,0-3);
\draw (1-1,0-3) .. controls (3-1.5,-0.6-3) .. (5-2,0-3);
\draw (0-1,1-3) .. controls (2-1.5,0.4-3) .. (4-2,1-3);
\draw[style=dashed,color=red] (1-1,1-3) .. controls (3-1.5,0.4-3) .. (5-2,1-3);
\draw[style=dashed,color=red] (2-1,1-3) .. controls (4-1.5,0.4-3) .. (6-2,1-3);
\draw[style=dashed,color=red] (1-1,2-3) .. controls (3-1.5,1.4-3) .. (5-2,2-3);
\draw[ultra thick,color=blue] (0-1,2-3) .. controls (2-1.5,1.4-3) .. (4-2,2-3);
\draw[ultra thick,color=blue] (2-1,2-3) .. controls (4-1.5,1.4-3) .. (6-2,2-3);
\draw[ultra thick,color=blue] (2-1,0-3) .. controls (4-1.5,-0.6-3) .. (6-2,0-3);
\draw (5-2,0-3) .. controls (7-2,-0.6-3) .. (9-2,0-3);
\draw (5-2,2-3) .. controls (7-2,1.4-3) .. (9-2,2-3);
\draw (6-2,2-3) .. controls (8-2,1.4-3) .. (10-2,2-3);
\draw[style=dashed,color=red] (4-2,0-3) .. controls (6-2,-0.6-3) .. (8-2,0-3);
\draw[style=dashed,color=red] (6-2,0-3) .. controls (8-2,-0.6-3) .. (10-2,0-3);
\draw[style=dashed,color=red] (6-2,1-3) .. controls (8-2,0.4-3) .. (10-2,1-3);
\draw[ultra thick,color=blue] (4-2,1-3) .. controls (6-2,0.4-3) .. (8-2,1-3);
\draw[ultra thick,color=blue] (5-2,1-3) .. controls (7-2,0.4-3) .. (9-2,1-3);
\draw[ultra thick,color=blue] (4-2,2-3) .. controls (6-2,1.4-3) .. (8-2,2-3);
\draw (8-2,1-3) .. controls (9.5-2,0.6-3) .. (11-2,1-3);
\draw (10-2,1-3) .. controls (11.5-2,0.6-3) .. (13-2,1-3);
\draw (10-2,2-3) .. controls (11.5-2,1.6-3) .. (13-2,2-3);
\draw[style=dashed,color=red] (8-2,0-3) .. controls (9.5-2,-0.4-3) .. (11-2,0-3);
\draw[style=dashed,color=red] (8-2,2-3) .. controls (9.5-2,1.6-3) .. (11-2,2-3);
\draw[style=dashed,color=red] (9-2,2-3) .. controls (10.5-2,1.6-3) .. (12-2,2-3);
\draw[ultra thick,color=blue] (9-2,0-3) .. controls (10.5-2,-0.4-3) .. (12-2,0-3);
\draw[ultra thick,color=blue] (9-2,1-3) .. controls (10.5-2,0.6-3) .. (12-2,1-3);
\draw[ultra thick,color=blue] (10-2,0-3) .. controls (11.5-2,-0.4-3) .. (13-2,0-3);
\node at (4-2,2-3) [regular polygon, regular polygon sides=4,draw=black,fill=blue,scale=0.7]{};
\node at (5-2,2-3) [circle,draw=black,fill=black,scale=0.7]{};
\node at (5-2,0-3) [circle,draw=black,fill=black,scale=0.7]{};
\node at (4-2,0-3) [regular polygon, regular polygon sides=3,draw=black,fill=red,scale=0.5]{};
\node at (6-2,0-3) [regular polygon, regular polygon sides=3,draw=black,fill=red,scale=0.5]{};
\node at (6-2,2-3) [circle,draw=black,fill=black,scale=0.7]{};
\node at (4-2,1-3) [regular polygon, regular polygon sides=4,draw=black,fill=blue,scale=0.7]{};
\node at (6-2,1-3) [regular polygon, regular polygon sides=3,draw=black,fill=red,scale=0.5]{};
\node at (5-2,1-3) [regular polygon, regular polygon sides=4,draw=black,fill=blue,scale=0.7]{};
\node at (9-2,0-3) [regular polygon, regular polygon sides=4,draw=black,fill=blue,scale=0.7]{};
\node at (8-2,1-3) [circle,draw=black,fill=black,scale=0.7]{};
\node at (10-2,1-3) [circle,draw=black,fill=black,scale=0.7]{};
\node at (8-2,0-3) [regular polygon, regular polygon sides=3,draw=black,fill=red,scale=0.5]{};
\node at (8-2,2-3) [regular polygon, regular polygon sides=3,draw=black,fill=red,scale=0.5]{};
\node at (10-2,2-3) [circle,draw=black,fill=black,scale=0.7]{};
\node at (9-2,1-3) [regular polygon, regular polygon sides=4,draw=black,fill=blue,scale=0.7]{};
\node at (9-2,2-3) [regular polygon, regular polygon sides=3,draw=black,fill=red,scale=0.5]{};
\node at (10-2,0-3) [regular polygon, regular polygon sides=4,draw=black,fill=blue,scale=0.7]{};
\end{tikzpicture}
\end{center}
\caption{The three completely independent spanning trees on the last five levels of $TM(3,3,q)$, for $q\equiv 2\pmod 3$ and $q>2$.}
\label{tor33}
\end{figure}
\begin{proof}
First, if $q\equiv 0\pmod{3}$, then Proposition \ref{tor333} allows us to conclude. 
For $q\equiv 1\pmod 3$ ($q\equiv 2\pmod 3$, respectively), we define three completely independent spanning trees by using the pattern of Proposition \ref{tor333} for every level except the last four (five, respectively) ones. 
If $q\equiv 1\pmod 3$, the trees are completed on the last four levels as depicted in Figure \ref{tor32} (the corresponding edge sets are given in Appendix \ref{ap:tor33}).
If $q\equiv 2\pmod 3$, the trees are completed on the last five levels as depicted in Figure \ref{tor33} (the corresponding edge sets are given in Appendix \ref{ap:tor34}).
\end{proof}
\section{Conclusion}
We conclude this paper by listing a few open problems:
\begin{enumerate}
 \item Determine conditions which ensure that there exist $r$ completely independent spanning trees in a graph.
 \item Does any $2r$-connected graph with sufficiently large girth admit $r$ completely independent spanning trees?
 \item Is it true that in every $4$-regular graph which is $4$-connected, there exist $2$ completely independent spanning trees?
 \item Does the 6-dimensional hypercube $Q_6=C_4\Box C_4 \Box C_4$ admit $3$ completely independent spanning trees?
\end{enumerate}

\newpage
\begin{appendices}
\section{Edge sets of the trees from Section \ref{sec:cartprod}}
\subsection{Three completely independent spanning trees in $K_{5}\square C_4$}\label{dek54}

$E(T_1)=\{u^{0}_{0} u^{0}_{3}, u^{0}_{0} u^{0}_{2}, u^{0}_{3} u^{0}_{1}, u^{0}_{3} u^{0}_{4}, u^{1}_{3} u^{1}_{1}, u^{1}_{3} u^{1}_{4}, u^{2}_{2} u^{2}_{1}, u^{2}_{2} u^{2}_{4}, u^{3}_{0} u^{3}_{2},  u^{3}_{0} u^{3}_{3}, u^{3}_{2} u^{3}_{1},  u^{3}_{2} u^{3}_{4},$ \newline
$ u^{0}_{0} u^{1}_{0}, u^{0}_{3} u^{1}_{3}, u^{1}_{2} u^{2}_{2}, u^{1}_{3} u^{2}_{3}, u^{2}_{0} u^{3}_{0}, u^{2}_{2} u^{3}_{2}, u^{3}_{0} u^{4}_{0}\};\newline$
$E(T_2)=\{u^{0}_{1} u^{0}_{0}, u^{0}_{1} u^{0}_{2}, u^{1}_{0} u^{1}_{4}, u^{1}_{0} u^{1}_{2}, u^{1}_{0} u^{1}_{3}, u^{1}_{4} u^{1}_{1}, u^{2}_{3} u^{2}_{4}, u^{2}_{3} u^{2}_{1}, u^{2}_{3} u^{2}_{2},  u^{2}_{4} u^{2}_{0}, u^{3}_{1} u^{3}_{3},  u^{3}_{1} u^{3}_{0},$ \newline
$ u^{3}_{1} u^{3}_{4}, u^{3}_{3} u^{3}_{2}, u^{0}_{4} u^{1}_{4}, u^{1}_{4} u^{2}_{4}, u^{2}_{3} u^{3}_{3}, u^{3}_{1} u^{4}_{1}, u^{3}_{3} u^{4}_{3}\};\newline$
$E(T_3)=\{u^{0}_{2} u^{0}_{4}, u^{0}_{2} u^{0}_{3}, u^{0}_{4} u^{0}_{0}, u^{0}_{4} u^{0}_{1}, u^{1}_{1} u^{1}_{2}, u^{1}_{1} u^{1}_{0}, u^{1}_{2} u^{1}_{3}, u^{1}_{2} u^{1}_{4}, u^{2}_{0} u^{2}_{1},  u^{2}_{0} u^{2}_{2}, u^{2}_{0} u^{2}_{3},  u^{2}_{1} u^{2}_{4},$ \newline
$ u^{3}_{4} u^{3}_{0}, u^{3}_{4} u^{3}_{3}, u^{0}_{2} u^{1}_{2}, u^{1}_{1} u^{2}_{1}, u^{2}_{1} u^{3}_{1}, u^{3}_{2} u^{4}_{2}, u^{3}_{4} u^{4}_{4}\}.$

\subsection{Three completely independent spanning trees in $K_{5}\square C_5$}\label{dek55}

$E(T_1)=\{u^{0}_{0} u^{0}_{3}, u^{0}_{0} u^{0}_{2}, u^{0}_{3} u^{0}_{1}, u^{0}_{3} u^{0}_{4}, u^{1}_{0} u^{1}_{3}, u^{1}_{0} u^{1}_{2}, u^{1}_{0} u^{1}_{4}, u^{1}_{3} u^{1}_{1}, u^{2}_{2} u^{2}_{0},  u^{2}_{2} u^{2}_{1}, u^{3}_{2} u^{3}_{4},  u^{3}_{2} u^{3}_{0},$ \newline
$u^{3}_{4} u^{3}_{1}, u^{3}_{4} u^{3}_{3}, u^{4}_{0} u^{4}_{2}, u^{4}_{0} u^{4}_{1}, u^{4}_{2} u^{4}_{3}, u^{4}_{2} u^{4}_{4}, u^{0}_{0} u^{1}_{0} , u^{1}_{3} u^{2}_{3}, u^{2}_{2} u^{3}_{2}, u^{2}_{4} u^{3}_{4}, u^{3}_{2} u^{4}_{2}, u^{4}_{0} u^{5}_{0} \}$; \newline
$E(T_2)=\{u^{0}_{1} u^{0}_{0}, u^{0}_{1} u^{0}_{2}, u^{1}_{4} u^{1}_{2}, u^{1}_{4} u^{1}_{3}, u^{2}_{0} u^{2}_{4}, u^{2}_{0} u^{2}_{3}, u^{2}_{4} u^{2}_{1}, u^{2}_{4} u^{2}_{2}, u^{3}_{0} u^{3}_{3},  u^{3}_{0} u^{3}_{4}, u^{3}_{3} u^{3}_{1},  u^{3}_{3} u^{3}_{2},$ \newline
$u^{4}_{1} u^{4}_{3}, u^{4}_{1} u^{4}_{2}, u^{4}_{1} u^{4}_{4}, u^{4}_{3} u^{4}_{0}, u^{0}_{1} u^{1}_{1}, u^{0}_{4} u^{1}_{4}, u^{1}_{0} u^{2}_{0} , u^{1}_{4} u^{2}_{4}, u^{2}_{0} u^{3}_{0}, u^{3}_{3} u^{4}_{3}, u^{4}_{1} u^{5}_{1}, u^{4}_{3} u^{5}_{3} \}$;
\newline
$E(T_3)=\{u^{0}_{2} u^{0}_{4}, u^{0}_{2} u^{0}_{3}, u^{0}_{4} u^{0}_{0}, u^{0}_{4} u^{0}_{1}, u^{1}_{1} u^{1}_{2}, u^{1}_{1} u^{1}_{0}, u^{1}_{1} u^{1}_{4}, u^{1}_{2} u^{1}_{3}, u^{2}_{1} u^{2}_{3},  u^{2}_{1} u^{2}_{0}, u^{2}_{3} u^{2}_{2},  u^{2}_{3} u^{2}_{4},$ \newline
$u^{3}_{1} u^{3}_{0}, u^{3}_{1} u^{3}_{2}, u^{4}_{4} u^{4}_{0}, u^{4}_{4} u^{4}_{3}, u^{0}_{2} u^{1}_{2}, u^{1}_{1} u^{2}_{1}, u^{2}_{1} u^{3}_{1} , u^{2}_{3} u^{3}_{3}, u^{3}_{1} u^{4}_{1}, u^{3}_{4} u^{4}_{4}, u^{4}_{2} u^{5}_{2}, u^{4}_{4} u^{5}_{4} \}$.

\subsection{Four completely independent spanning trees in $K_{7}\square C_3$}\label{dek73}
$E(T_1)=\{u^{0}_{0} u^{0}_{1},u^{0}_{0} u^{0}_{3}, u^{0}_{0} u^{0}_{5},
u^{0}_{0} u^{0}_{6}, u^{1}_{0}  u^{1}_{2},
u^{1}_{0}  u^{1}_{4}, u^{1}_{0} u^{1}_{5},u^{1}_{2} u^{1}_{1},
 u^{1}_{2} u^{1}_{3}, u^{1}_{2} u^{1}_{6}, u^{2}_{2} u^{2}_{4}, u^{2}_{2} u^{2}_{5},$\newline
$u^{2}_{2} u^{2}_{6}, u^{2}_{4} u^{2}_{0}, u^{2}_{4} u^{2}_{1}, u^{2}_{4} u^{2}_{3},
u^{0}_{0} u^{1}_{0}, u^{1}_{2} u^{2}_{2}, u^{0}_{2} u^{2}_{2}, 
u^{0}_{4} u^{2}_{4}\}$; \newline
$E(T_2)=\{u^{0}_{1} u^{0}_{2}, u^{0}_{1} u^{0}_{3},u^{0}_{1} u^{0}_{5}, u^{0}_{2} u^{0}_{0},
u^{0}_{2} u^{0}_{4}, u^{0}_{2}  u^{0}_{6},
u^{1}_{6}  u^{1}_{0}, u^{1}_{6} u^{1}_{3},u^{1}_{6} u^{1}_{4},
u^{1}_{6} u^{1}_{5}, u^{2}_{1} u^{2}_{0},  u^{2}_{1} u^{2}_{2},$\newline
$u^{2}_{1} u^{2}_{3},u^{2}_{1} u^{2}_{6}, u^{2}_{6} u^{2}_{4}, u^{2}_{6} u^{2}_{5},
u^{0}_{1} u^{1}_{1}, u^{0}_{2} u^{1}_{2}, u^{1}_{6} u^{2}_{6}, 
u^{0}_{1} u^{2}_{1}\}$; \newline
$E(T_3)=\{u^{0}_{3} u^{0}_{2},u^{0}_{3} u^{0}_{4}, u^{0}_{3} u^{0}_{5},
u^{0}_{4} u^{0}_{0}, u^{0}_{4}  u^{0}_{1}, u^{0}_{4}  u^{0}_{6},
u^{1}_{1}  u^{1}_{0}, u^{1}_{1} u^{1}_{3},u^{1}_{1} u^{1}_{4},
u^{1}_{1} u^{1}_{6}, u^{1}_{4} u^{1}_{2},  u^{1}_{4} u^{1}_{5},$\newline
$u^{2}_{3} u^{2}_{0},u^{2}_{3} u^{2}_{2}, u^{2}_{3} u^{2}_{5}, u^{2}_{3} u^{2}_{6},
u^{0}_{4} u^{1}_{4}, u^{1}_{1} u^{2}_{1}, u^{1}_{4} u^{2}_{4}, 
u^{0}_{3} u^{2}_{3}\}$; \newline
$E(T_4)=\{u^{0}_{5} u^{0}_{2},u^{0}_{5} u^{0}_{4}, u^{0}_{5} u^{0}_{6},
u^{0}_{6} u^{0}_{1}, u^{0}_{6}  u^{0}_{3},
u^{1}_{3}  u^{1}_{0}, u^{1}_{3} u^{1}_{4},u^{1}_{3} u^{1}_{5},
u^{1}_{5} u^{1}_{1}, u^{1}_{5} u^{1}_{2},  u^{2}_{0} u^{2}_{2},u^{2}_{0} u^{2}_{5},$\newline
$u^{2}_{0} u^{2}_{6}, u^{2}_{5} u^{2}_{1}, u^{2}_{5} u^{2}_{4},
u^{0}_{5} u^{1}_{5}, u^{0}_{6} u^{1}_{6}, u^{1}_{3} u^{2}_{3}, u^{1}_{5} u^{2}_{5},
u^{0}_{0} u^{2}_{0}\}$.
\subsection{Four completely independent spanning trees in $K_{7}\square C_4$}
\label{dek74}
$E(T_1)=\{u^{0}_{0} u^{0}_{1},u^{0}_{0} u^{0}_{3}, u^{0}_{0} u^{0}_{5},
u^{0}_{0} u^{0}_{6}, u^{1}_{0}  u^{1}_{1},
u^{1}_{0}  u^{1}_{1}, u^{1}_{0} u^{1}_{2},u^{1}_{0} u^{1}_{3},
 u^{1}_{2} u^{1}_{5}, u^{1}_{2} u^{1}_{6},  u^{2}_{2} u^{2}_{3} , u^{2}_{2} u^{2}_{5},u^{2}_{2} u^{2}_{6},$\newline
$ u^{2}_{5} u^{2}_{0}, u^{2}_{5} u^{2}_{1}, u^{2}_{5} u^{2}_{4},
u^{3}_{4} u^{3}_{3}, u^{3}_{4} u^{3}_{5}, u^{3}_{4} u^{3}_{6}, u^{3}_{5} u^{3}_{0},
u^{3}_{5} u^{3}_{1}, u^{0}_{0} u^{1}_{0}, u^{0}_{2} u^{1}_{2}, u^{1}_{2} u^{2}_{2}, 
u^{2}_{2} u^{3}_{2}, u^{2}_{5} u^{3}_{5}, u^{0}_{4} u^{3}_{4}\}$ \newline
$E(T_2)=\{u^{0}_{1} u^{0}_{2},u^{0}_{1} u^{0}_{4}, u^{0}_{1} u^{0}_{5},
u^{0}_{2} u^{0}_{0}, u^{0}_{2}  u^{0}_{3},
u^{0}_{2}  u^{0}_{6}, u^{1}_{1} u^{1}_{2},u^{1}_{1} u^{1}_{4},
u^{1}_{1} u^{1}_{6}, u^{1}_{6} u^{1}_{0},  u^{1}_{6} u^{1}_{3},u^{1}_{6} u^{1}_{5}, u^{2}_{4} u^{2}_{2},$\newline
$ u^{2}_{4} u^{2}_{3}, u^{2}_{4} u^{2}_{6},
u^{2}_{6} u^{2}_{0}, u^{2}_{6} u^{2}_{5}, u^{3}_{2} u^{3}_{0}, u^{3}_{2} u^{3}_{3},
u^{3}_{2} u^{3}_{4}, u^{3}_{2} u^{3}_{5}, u^{0}_{1} u^{1}_{1}, u^{1}_{1} u^{2}_{1}, 
u^{1}_{6} u^{2}_{6}, u^{2}_{6} u^{3}_{6}, u^{0}_{1} u^{3}_{1},
u^{0}_{2} u^{3}_{2} \}$ \newline
$E(T_3)=\{u^{0}_{3} u^{0}_{4},u^{0}_{3} u^{0}_{1}, u^{0}_{3} u^{0}_{5},
u^{0}_{4} u^{0}_{0}, u^{0}_{4}  u^{0}_{2},
u^{0}_{4}  u^{0}_{6}, u^{1}_{4} u^{1}_{2},u^{1}_{4} u^{1}_{5},
u^{1}_{4} u^{1}_{6}, u^{1}_{5} u^{1}_{0},  u^{1}_{5} u^{1}_{1},u^{2}_{1} u^{2}_{0}, u^{2}_{1} u^{2}_{2},$\newline
$ u^{2}_{1} u^{2}_{4}, u^{2}_{1} u^{2}_{6},
u^{3}_{1} u^{3}_{0}, u^{3}_{1} u^{3}_{2}, u^{3}_{1} u^{3}_{3}, u^{3}_{1} u^{3}_{4},
u^{3}_{3} u^{3}_{5}, u^{3}_{3} u^{3}_{7}, u^{0}_{3} u^{1}_{3}, u^{0}_{4} u^{1}_{4}, 
u^{1}_{5} u^{2}_{5}, u^{2}_{1} u^{3}_{1}, u^{2}_{3} u^{3}_{3},
u^{0}_{3} u^{3}_{3} \}$ \newline
$E(T_4)=\{u^{0}_{5} u^{0}_{2},u^{0}_{5} u^{0}_{4}, u^{0}_{5} u^{0}_{6},
u^{0}_{6} u^{0}_{1}, u^{0}_{6}  u^{0}_{3},
u^{1}_{3}  u^{1}_{1}, u^{1}_{3} u^{1}_{2},u^{1}_{3} u^{1}_{4},
u^{1}_{3} u^{1}_{5}, u^{2}_{0} u^{2}_{2},  u^{2}_{0} u^{2}_{3},u^{2}_{0} u^{2}_{4}, u^{2}_{3} u^{2}_{1},$\newline
$ u^{2}_{3} u^{2}_{5}, u^{2}_{3} u^{2}_{6},
u^{3}_{0} u^{3}_{3}, u^{3}_{0} u^{3}_{4}, u^{3}_{0} u^{3}_{6}, u^{3}_{6} u^{3}_{1},
u^{3}_{6} u^{3}_{2}, u^{3}_{6} u^{3}_{5}, u^{0}_{6} u^{1}_{6}, u^{1}_{0} u^{2}_{0}, 
u^{1}_{3} u^{2}_{3}, u^{2}_{0} u^{3}_{0}, u^{0}_{0} u^{3}_{0},
u^{0}_{6} u^{3}_{6} \}$.
\subsection{Five completely independent spanning trees in $K_{9}\square C_4$}
\label{dek94}
$E(T_1)=\{u^{0}_{0} u^{0}_{2},u^{0}_{0} u^{0}_{4}, u^{0}_{0} u^{0}_{5},
u^{0}_{0}  u^{0}_{8}, u^{0}_{5} u^{0}_{1}, u^{0}_{5} u^{0}_{3}, u^{0}_{5} u^{0}_{6}, u^{0}_{5} u^{0}_{7},
u^{1}_{0} u^{1}_{2}, u^{1}_{0} u^{1}_{4}, u^{1}_{0}  u^{1}_{6},
u^{1}_{0}  u^{1}_{7}, u^{1}_{4} u^{1}_{1},$\newline
$u^{1}_{4} u^{1}_{3}, u^{1}_{4} u^{1}_{5}, u^{1}_{4} u^{1}_{8},
 u^{2}_{0} u^{2}_{3} , u^{2}_{0} u^{2}_{4},u^{2}_{0} u^{2}_{6},
u^{2}_{4} u^{2}_{1}, u^{2}_{4} u^{2}_{8}, u^{2}_{8} u^{2}_{2},
u^{2}_{8} u^{2}_{5}, u^{2}_{8} u^{2}_{7},
u^{3}_{8} u^{3}_{1}, u^{3}_{8} u^{3}_{2}, u^{3}_{8} u^{3}_{3}, u^{3}_{8} u^{3}_{6},$\newline
$u^{3}_{8} u^{3}_{7}, u^{0}_{0} u^{1}_{0},u^{1}_{0} u^{2}_{0},
u^{2}_{4} u^{3}_{4},u^{2}_{8} u^{3}_{8},u^{0}_{0} u^{3}_{0},u^{0}_{5} u^{3}_{5}\};$ \newline
$E(T_2)=\{u^{0}_{3} u^{0}_{0},u^{0}_{3} u^{0}_{4}, u^{0}_{3} u^{0}_{7},
u^{0}_{3}  u^{0}_{8}, u^{0}_{4} u^{0}_{1}, u^{0}_{4} u^{0}_{2}, u^{0}_{4} u^{0}_{5}, u^{0}_{4} u^{0}_{6},
u^{1}_{1} u^{1}_{0}, u^{1}_{1} u^{1}_{2}, u^{1}_{1}  u^{1}_{6},
u^{1}_{1}  u^{1}_{7}, u^{1}_{1} u^{1}_{8},$\newline
$u^{2}_{1} u^{2}_{0}, u^{2}_{1} u^{2}_{2}, u^{2}_{1} u^{2}_{5},
 u^{2}_{1} u^{2}_{7} , u^{2}_{1} u^{2}_{8},u^{2}_{5} u^{2}_{4},
u^{2}_{5} u^{2}_{6}, u^{3}_{3} u^{3}_{2}, u^{3}_{3} u^{3}_{5},
u^{3}_{3} u^{3}_{6}, u^{3}_{3} u^{3}_{7},
u^{3}_{5} u^{3}_{0}, u^{3}_{5} u^{3}_{1}, u^{3}_{5} u^{3}_{4}, u^{3}_{5} u^{3}_{8},$\newline
$u^{0}_{3} u^{1}_{3}, u^{0}_{4} u^{1}_{4},u^{1}_{1} u^{2}_{1},
u^{1}_{5} u^{2}_{5},u^{2}_{3} u^{3}_{3},u^{2}_{5} u^{3}_{5},u^{0}_{3} u^{3}_{3}\};$ \newline
$E(T_3)=\{u^{0}_{2} u^{0}_{1},u^{0}_{2} u^{0}_{3}, u^{0}_{2} u^{0}_{5},
u^{0}_{2}  u^{0}_{7}, u^{0}_{7} u^{0}_{0}, u^{0}_{7} u^{0}_{4}, u^{0}_{7} u^{0}_{6}, u^{0}_{7} u^{0}_{8},
u^{1}_{2} u^{1}_{3}, u^{1}_{2} u^{1}_{4}, u^{1}_{2}  u^{1}_{5},
u^{1}_{2}  u^{1}_{6}, u^{1}_{3} u^{1}_{0},$\newline
$ u^{1}_{3} u^{1}_{1}, u^{1}_{3} u^{1}_{7},
 u^{1}_{3} u^{1}_{8} , u^{2}_{3} u^{2}_{1},u^{2}_{3} u^{2}_{4},
u^{2}_{3} u^{2}_{5}, u^{2}_{3} u^{2}_{6}, u^{2}_{6} u^{2}_{8},
u^{3}_{0} u^{3}_{2}, u^{3}_{0} u^{3}_{3},
u^{3}_{0} u^{3}_{4}, u^{3}_{0} u^{3}_{7}, u^{3}_{0} u^{3}_{8},u^{3}_{7} u^{3}_{1},u^{3}_{7} u^{3}_{5},$\newline
$ u^{3}_{7} u^{3}_{6},u^{0}_{2} u^{1}_{2},
u^{1}_{2} u^{2}_{2},u^{1}_{3} u^{2}_{3},u^{2}_{0} u^{3}_{0},u^{2}_{7} u^{3}_{7}, u^{0}_{7} u^{3}_{7} \};$ \newline
$E(T_4)=\{u^{0}_{6} u^{0}_{0},u^{0}_{6} u^{0}_{1}, u^{0}_{6} u^{0}_{2},
u^{0}_{6} u^{0}_{3}, u^{0}_{6}  u^{0}_{8},
u^{1}_{5} u^{1}_{0}, u^{1}_{5} u^{1}_{1}, u^{1}_{5}  u^{1}_{3},
u^{1}_{5}  u^{1}_{7}, u^{1}_{5} u^{1}_{8},u^{1}_{7} u^{1}_{2},
 u^{1}_{7} u^{1}_{4}, u^{1}_{7} u^{1}_{6},$\newline
$ u^{2}_{6} u^{2}_{1} , u^{2}_{6} u^{2}_{2},u^{2}_{6} u^{2}_{4},
u^{2}_{6} u^{2}_{7}, u^{2}_{6} u^{2}_{8}, u^{2}_{7} u^{2}_{0},
u^{2}_{7} u^{2}_{3}, u^{2}_{7} u^{2}_{5},
u^{3}_{4} u^{3}_{2}, u^{3}_{4} u^{3}_{3}, u^{3}_{4} u^{3}_{6}, u^{3}_{4} u^{3}_{7},
u^{3}_{4} u^{3}_{8}, u^{3}_{6} u^{3}_{0}, u^{3}_{6} u^{3}_{1},$\newline
$u^{3}_{6} u^{3}_{5},u^{0}_{5} u^{1}_{5}, u^{0}_{7} u^{1}_{7},u^{1}_{7} u^{2}_{7},
u^{2}_{6} u^{3}_{6},u^{0}_{4} u^{3}_{4}, u^{0}_{6} u^{3}_{6} \};$ \newline
$E(T_5)=\{u^{0}_{1} u^{0}_{0},u^{0}_{1} u^{0}_{3}, u^{0}_{1} u^{0}_{7},
u^{0}_{1} u^{0}_{8}, u^{0}_{8}  u^{0}_{2},
u^{0}_{8} u^{0}_{4}, u^{0}_{8} u^{0}_{5}, u^{1}_{6}  u^{1}_{3},
u^{1}_{6}  u^{1}_{4}, u^{1}_{6} u^{1}_{5},u^{1}_{6} u^{1}_{8},
 u^{1}_{8} u^{1}_{0}, u^{1}_{8} u^{1}_{2},$\newline
$ u^{1}_{8} u^{1}_{7} , u^{2}_{2} u^{2}_{0},u^{2}_{2} u^{2}_{3},
u^{2}_{2} u^{2}_{4}, u^{2}_{2} u^{2}_{5}, u^{2}_{2} u^{2}_{7},
u^{3}_{1} u^{3}_{0}, u^{3}_{1} u^{3}_{2},
u^{3}_{1} u^{3}_{3}, u^{3}_{1} u^{3}_{4}, u^{3}_{2} u^{3}_{5}, u^{3}_{2} u^{3}_{6},
u^{3}_{2} u^{3}_{7}, u^{0}_{1} u^{1}_{1},u^{0}_{6} u^{1}_{6},$\newline
$u^{0}_{8} u^{1}_{8},u^{1}_{6} u^{2}_{6},u^{1}_{8} u^{2}_{8}, u^{2}_{1} u^{3}_{1},u^{2}_{2} u^{3}_{2},
u^{0}_{1} u^{3}_{1},u^{0}_{8} u^{3}_{8} \}.$ 
\subsection{Five completely independent spanning trees in $K_{9}\square C_4$}\label{dek95}
$E(T_1)=\{u^{0}_{0} u^{0}_{2},u^{0}_{0} u^{0}_{3}, u^{0}_{0} u^{0}_{5},
u^{0}_{0} u^{0}_{6}, u^{0}_{0}  u^{0}_{7},
u^{0}_{6} u^{0}_{1}, u^{0}_{6} u^{0}_{4}, u^{0}_{6}  u^{0}_{8},
u^{1}_{6}  u^{1}_{1}, u^{1}_{6} u^{1}_{2},u^{1}_{6} u^{1}_{5},
 u^{1}_{6} u^{1}_{7}, u^{1}_{6} u^{1}_{8},$\newline
$ u^{2}_{3} u^{2}_{0} , u^{2}_{3} u^{2}_{1},u^{2}_{3} u^{2}_{4},
u^{2}_{3} u^{2}_{6}, u^{2}_{3} u^{2}_{8}, u^{2}_{4} u^{2}_{2},
u^{2}_{4} u^{2}_{5}, u^{2}_{4} u^{2}_{7},
u^{3}_{4} u^{3}_{0}, u^{3}_{4} u^{3}_{2}, u^{3}_{4} u^{3}_{3}, u^{3}_{4} u^{3}_{7},
u^{3}_{7} u^{3}_{1}, u^{3}_{7} u^{3}_{5},u^{3}_{7} u^{3}_{8},$\newline
$u^{4}_{4} u^{4}_{0},u^{4}_{4} u^{4}_{2},u^{4}_{4} u^{4}_{6},
u^{4}_{4} u^{4}_{8},u^{4}_{6} u^{4}_{1},u^{4}_{6} u^{4}_{3},u^{4}_{6} u^{4}_{5},
u^{0}_{0} u^{1}_{0}, u^{0}_{6} u^{1}_{6}, u^{1}_{3} u^{2}_{3}, u^{1}_{4} u^{2}_{4}, 
u^{2}_{4} u^{3}_{4}, u^{3}_{4} u^{4}_{4}, u^{3}_{6} u^{4}_{6}, u^{3}_{7} u^{4}_{7},  u^{0}_{6} u^{4}_{6}\};$ \newline
$E(T_2)=\{u^{0}_{7} u^{0}_{1},u^{0}_{7} u^{0}_{2}, u^{0}_{7} u^{0}_{4},
u^{0}_{7} u^{0}_{6}, u^{0}_{7}  u^{0}_{8},
u^{1}_{3} u^{1}_{1}, u^{1}_{3} u^{1}_{4}, u^{1}_{3}  u^{1}_{6},
u^{1}_{3}  u^{1}_{7}, u^{1}_{3} u^{1}_{8},u^{1}_{7} u^{1}_{0},
 u^{1}_{7} u^{1}_{2}, u^{1}_{7} u^{1}_{6},$\newline
$ u^{2}_{0} u^{2}_{1} , u^{2}_{0} u^{2}_{4},u^{2}_{0} u^{2}_{5},
u^{2}_{0} u^{2}_{7}, u^{2}_{0} u^{2}_{8}, u^{2}_{7} u^{2}_{2},
u^{2}_{7} u^{2}_{3}, u^{2}_{7} u^{2}_{6},
u^{3}_{0} u^{3}_{1}, u^{3}_{0} u^{3}_{5}, u^{3}_{0} u^{3}_{8}, u^{3}_{1} u^{3}_{2},
u^{3}_{1} u^{3}_{3}, u^{3}_{1} u^{3}_{4},u^{3}_{1} u^{3}_{6},$\newline
$u^{4}_{0} u^{4}_{1},u^{4}_{0} u^{4}_{5},u^{4}_{0} u^{4}_{6},
u^{4}_{0} u^{4}_{7},u^{4}_{5} u^{4}_{2},u^{4}_{5} u^{4}_{3},u^{4}_{5} u^{4}_{4},u^{4}_{5} u^{4}_{8},
u^{0}_{3} u^{1}_{3}, u^{0}_{7} u^{1}_{7}, u^{1}_{7} u^{2}_{7}, u^{2}_{0} u^{3}_{0}, 
u^{2}_{7} u^{3}_{7}, u^{3}_{0} u^{4}_{0}, u^{0}_{0} u^{4}_{0},  u^{0}_{5} u^{4}_{5}\};$ \newline
$E(T_3)=\{u^{0}_{4} u^{0}_{0},u^{0}_{4} u^{0}_{3}, u^{0}_{4} u^{0}_{7},
u^{0}_{4} u^{0}_{8}, u^{0}_{7}  u^{0}_{1},
u^{0}_{7} u^{0}_{2}, u^{0}_{7} u^{0}_{5}, u^{0}_{7}  u^{0}_{6},
u^{1}_{2}  u^{1}_{1}, u^{1}_{2} u^{1}_{3},u^{1}_{2} u^{1}_{4},
 u^{1}_{2} u^{1}_{5}, u^{1}_{2} u^{1}_{8},$\newline
$ u^{1}_{4} u^{1}_{0} , u^{1}_{4} u^{1}_{6},u^{1}_{4} u^{1}_{7},
u^{2}_{1} u^{2}_{2}, u^{2}_{1} u^{2}_{4}, u^{2}_{1} u^{2}_{5},
u^{2}_{1} u^{2}_{7}, u^{2}_{1} u^{2}_{8},
u^{2}_{2} u^{2}_{0}, u^{2}_{2} u^{2}_{3}, u^{2}_{2} u^{2}_{6}, u^{3}_{2} u^{3}_{0},
u^{3}_{2} u^{3}_{3}, u^{3}_{2} u^{3}_{5},u^{3}_{2} u^{3}_{7},$\newline
$u^{3}_{5} u^{3}_{1},u^{3}_{5} u^{3}_{4},u^{3}_{5} u^{3}_{6},
u^{3}_{5} u^{3}_{8},u^{4}_{2} u^{4}_{0},u^{4}_{2} u^{4}_{1},u^{4}_{2} u^{4}_{3},u^{4}_{2} u^{4}_{6},u^{4}_{2} u^{4}_{8},
u^{0}_{4} u^{1}_{4}, u^{1}_{2} u^{2}_{2}, u^{2}_{2} u^{3}_{2}, u^{3}_{2} u^{4}_{2}, 
u^{3}_{5} u^{4}_{5}, u^{0}_{4} u^{4}_{4},  u^{0}_{7} u^{4}_{7}\};$ \newline
$E(T_4)=\{u^{0}_{1} u^{0}_{0},u^{0}_{1} u^{0}_{3}, u^{0}_{1} u^{0}_{4},
u^{0}_{1} u^{0}_{8}, u^{0}_{3}  u^{0}_{2},
u^{0}_{3} u^{0}_{5}, u^{0}_{3} u^{0}_{6}, u^{0}_{3}  u^{0}_{7},
u^{1}_{0}  u^{1}_{1}, u^{1}_{0} u^{1}_{2},u^{1}_{0} u^{1}_{3},
 u^{1}_{0} u^{1}_{5}, u^{1}_{0} u^{1}_{6},$\newline
$ u^{1}_{1} u^{1}_{4} , u^{1}_{1} u^{1}_{7},u^{1}_{1} u^{1}_{8},
u^{2}_{8} u^{2}_{2}, u^{2}_{8} u^{2}_{4}, u^{2}_{8} u^{2}_{5},
u^{2}_{8} u^{2}_{6}, u^{2}_{8} u^{2}_{7},
u^{3}_{3} u^{3}_{0}, u^{3}_{3} u^{3}_{5}, u^{3}_{3} u^{3}_{7}, u^{3}_{3} u^{3}_{8},
u^{3}_{8} u^{3}_{1}, u^{3}_{8} u^{3}_{2},u^{3}_{8} u^{3}_{4},$\newline
$u^{3}_{8} u^{3}_{6},u^{4}_{3} u^{4}_{0},u^{4}_{3} u^{4}_{4},
u^{4}_{3} u^{4}_{7},u^{4}_{3} u^{4}_{8},u^{4}_{7} u^{4}_{2},u^{4}_{7} u^{4}_{5},u^{4}_{7} u^{4}_{6},u^{0}_{1} u^{1}_{1},
u^{1}_{0} u^{2}_{0}, u^{1}_{1} u^{2}_{1}, u^{2}_{3} u^{3}_{3}, u^{2}_{8} u^{3}_{8}, 
u^{3}_{3} u^{4}_{3}, u^{0}_{1} u^{4}_{1},  u^{0}_{3} u^{4}_{3}\};$ \newline
$E(T_5)=\{u^{0}_{2} u^{0}_{1},u^{0}_{2} u^{0}_{4}, u^{0}_{2} u^{0}_{6},
u^{0}_{2} u^{0}_{8}, u^{0}_{8}  u^{0}_{0},
u^{0}_{8} u^{0}_{3}, u^{0}_{8} u^{0}_{7}, u^{1}_{5}  u^{1}_{1},
u^{1}_{5}  u^{1}_{3}, u^{1}_{5} u^{1}_{4},u^{1}_{5} u^{1}_{8},
 u^{1}_{8} u^{1}_{0}, u^{1}_{8} u^{1}_{7},$\newline
$ u^{2}_{5} u^{2}_{2} , u^{2}_{5} u^{2}_{3},u^{2}_{5} u^{2}_{6},
u^{2}_{5} u^{2}_{7}, u^{2}_{6} u^{2}_{0}, u^{2}_{6} u^{2}_{1},
u^{2}_{6} u^{2}_{4}, u^{3}_{6} u^{3}_{0},
u^{3}_{6} u^{3}_{2}, u^{3}_{6} u^{3}_{3}, u^{3}_{6} u^{3}_{4}, u^{3}_{6} u^{3}_{7},
u^{4}_{1} u^{4}_{3}, u^{4}_{1} u^{4}_{4},u^{4}_{1} u^{4}_{5},$\newline
$u^{4}_{1} u^{4}_{8},u^{4}_{8} u^{4}_{0},u^{4}_{8} u^{4}_{6},
u^{4}_{8} u^{4}_{7},u^{0}_{2} u^{1}_{2},u^{0}_{5} u^{1}_{5},u^{0}_{8} u^{1}_{8},u^{1}_{5} u^{2}_{5},u^{1}_{6} u^{2}_{6},
u^{1}_{8} u^{2}_{8},u^{2}_{5} u^{3}_{5}, u^{2}_{6} u^{3}_{6}, u^{3}_{1} u^{4}_{1}, u^{3}_{8} u^{4}_{8}, 
u^{0}_{2} u^{4}_{2}, u^{0}_{8} u^{4}_{8}\}.$

 \section{Edge sets of the trees from Section \ref{sec:3Dgrid}}
\subsection{Three completely independent spanning trees in the last four levels of $TM(3,3,q)$}\label{ap:tor33}
$E(T_1)=\{(0,0,0)(1,0,0), (0,0,0) (0,1,0),(0,0,0)(0,2,0),(0,0,0) (0,0,1),$\newline
$(1,0,0)(1,2,0),(1,0,0) (2,0,0), (1,0,0) (1,0,1), (0,1,0) (1,1,0), (0,1,0) (0,1,1),$\newline
$(0,1,1)(1,1,1),(0,1,1) (0,2,1), (0,1,1) (0,1,2), (0,2,1) (1,2,1), (0,2,1) (2,2,1),$\newline
$(0,2,1)(0,2,2),(2,2,1) (2,0,1), (2,2,1) (2,1,1), (1,0,2) (0,0,2), (1,0,2) (2,0,2),$\newline
$(1,0,2)(1,2,2),(1,0,2) (1,0,3), (1,2,2) (2,2,2), (1,2,2) (1,1,2), (1,2,2) (1,2,3),$\newline
$(2,2,2)(2,1,2),(2,2,2) (2,2,3), (0,1,3) (1,1,3), (0,1,3) (0,0,3), (0,1,3) (2,1,3),$\newline
$(0,1,3)(0,1,4),(2,1,3) (2,0,3), (2,1,3) (2,2,3), (2,1,3) (2,1,4), (2,2,3) (0,2,3),$\newline
$(2,2,3) (2,2,4)\};$\newline
$E(T_2)=\{(1,1,0)(2,1,0), (1,1,0) (1,0,0),(1,1,0)(1,2,0),(1,1,0) (1,1,1),$\newline
$(2,1,0)(0,1,0),(2,1,0) (2,0,0), (2,1,0) (2,1,1), (1,2,0) (2,2,0), (1,2,0) (1,2,1),$\newline
$(0,0,1)(1,0,1),(0,0,1) (0,1,1), (0,0,1) (0,2,1), (1,0,1) (2,0,1), (1,0,1) (1,2,1),$\newline
$(1,0,1)(1,0,2),(1,2,1) (2,2,1), (1,2,1) (1,2,2), (0,0,2) (2,0,2), (0,0,2) (0,2,2),$\newline
$(0,0,2)(0,0,3),(2,0,2) (2,1,2), (2,0,2) (2,2,2), (2,0,2) (2,0,3), (2,1,2) (0,1,2),$\newline
$(2,1,2)(1,1,2),(2,1,2) (2,1,3), (0,0,3) (1,0,3), (0,0,3) (0,2,3), (0,0,3) (0,0,4),$\newline
$(0,2,3) (1,2,3), (0,2,3) (0,1,3), (0,2,3) (0,2,4), (1,2,3) (2,2,3),(1,2,3) (1,1,3),$\newline
$(1,2,3) (1,2,4)\};$\newline
$E(T_3)=\{(2,0,0)(0,0,0), (2,0,0) (2,2,0),(2,0,0)(2,0,1),(0,2,0) (1,2,0),$\newline
$(0,2,0)(2,2,0),(0,2,0) (0,1,0), (0,2,0) (0,2,1), (2,2,0) (2,1,0), (2,2,0) (2,2,1),$\newline
$(2,0,1)(0,0,1),(2,0,1) (2,1,1), (2,0,1) (2,0,2), (1,1,1) (2,1,1), (1,1,1) (1,0,1),$\newline
$(1,1,1)(1,2,1),(2,1,1) (0,1,1), (2,1,1) (2,1,2), (0,1,2) (1,1,2), (0,1,2) (0,0,2),$\newline
$(0,1,2)(0,2,2),(0,1,2) (0,1,3), (1,1,2) (1,0,2), (1,1,2) (1,1,3), (0,2,2) (1,2,2),$\newline
$(0,2,2)(2,2,2),(0,2,2) (0,2,3), (1,0,3) (2,0,3), (1,0,3) (1,1,3), (1,0,3) (1,2,3),$\newline
$(1,0,3) (1,0,4), (2,0,3) (0,0,3), (2,0,3) (2,2,3), (2,0,3) (2,0,4),(1,1,3) (2,1,3),$\newline
$(1,1,3) (1,1,4)\}.$
\subsection{Three completely independent spanning trees in the last five levels of $TM(3,3,q)$}\label{ap:tor34}
$E(T_1)=\{(0,0,0)(1,0,0), (0,0,0) (0,1,0),(0,0,0)(0,2,0),(0,0,0) (0,0,1),$\newline
$(1,0,0)(1,2,0),(1,0,0) (2,0,0), (1,0,0) (1,0,1), (0,1,0) (1,1,0), (0,1,0) (0,1,1),$\newline
$(1,0,1)(2,0,1),(1,0,1) (1,2,1), (1,2,1) (2,2,1), (1,2,1) (1,1,1), (1,2,1) (1,2,2),$\newline
$(2,2,1)(0,2,1),(2,2,1) (2,1,1), (2,2,1) (2,2,2), (0,0,2) (1,0,2), (0,0,2) (2,0,2),$\newline
$(0,0,2)(0,1,2),(0,0,2),(0,0,3), (1,0,2) (1,1,2), (1,0,2) (1,0,3), (0,1,2) (2,1,2),$\newline
$(0,1,2) (0,2,2),(0,1,2) (0,1,3), (1,0,3) (2,0,3),(1,0,3)(1,2,3),(1,0,3) (1,0,4), $\newline
$(1,2,3) (2,2,3), (1,2,3) (1,1,3), (1,2,3) (1,2,4),(2,2,3)(0,2,3),(2,2,3)(2,1,3),$\newline
$(2,2,3) (2,2,4), (0,1,4) (1,1,4), (0,1,4) (0,0,4), (0,1,4) (2,1,4), (0,1,4)(0,1,5),$\newline
$(2,1,4) (2,0,4), (2,1,4) (2,2,4), (2,1,4) (2,1,5), (2,2,4) (0,2,4),(2,2,4) (2,2,5)\};$\newline
$E(T_2)=\{(1,1,0)(2,1,0), (1,1,0) (1,0,0),(1,1,0)(1,2,0),(1,1,0) (1,1,1),$\newline
$(2,1,0)(0,1,0),(2,1,0) (2,0,0), (2,1,0) (2,1,1), (1,2,0) (2,2,0), (1,2,0) (1,2,1),$\newline
$(0,0,1)(1,0,1),(0,0,1) (2,0,1), (0,0,1) (0,2,1), (0,0,1) (0,0,2), (2,0,1) (2,1,1),$\newline
$(2,0,1)(2,2,1),(2,0,1) (2,0,2), (2,1,1) (0,1,1), (1,1,2) (0,1,2), (1,1,2) (2,1,2),$\newline
$(1,1,2) (1,2,2),(1,1,2)(1,1,3), (2,1,2) (2,2,2), (2,1,2) (2,1,3), (1,2,2) (0,2,2),$\newline
$(1,2,2) (1,0,2),(1,2,2) (1,2,3),(0,0,3) (1,0,3),(0,0,3) (2,0,3), (0,0,3) (0,2,3),$\newline
$(0,0,3)(0,0,4),(2,0,3) (2,1,3), (2,0,3) (2,2,3), (2,0,3) (2,0,4), (2,1,3) (0,1,3),$\newline
$(2,1,3) (2,1,4), (0,0,4) (1,0,4), (0,0,4) (0,2,4), (0,0,4) (0,0,5),(0,2,4) (1,2,4), $\newline
$(0,2,4) (0,1,4), (0,2,4) (0,2,5), (1,2,4) (2,2,4),(1,2,4) (1,1,4),(1,2,4) (1,2,5)\};$\newline
$E(T_3)=\{(2,0,0)(0,0,0), (2,0,0) (2,2,0),(2,0,0)(2,0,1),(0,2,0) (1,2,0),$\newline
$(0,2,0)(2,2,0),(0,2,0) (0,1,0), (0,2,0) (0,2,1), (2,2,0) (2,1,0), (2,2,0) (2,2,1),$\newline
$(0,1,1) (1,1,1),(0,1,1)(0,0,1), (0,1,1) (0,2,1),(0,1,1) (0,1,2), (1,1,1) (2,1,1),$\newline
$(1,1,1) (1,0,1),(1,1,1)(1,1,2), (0,2,1) (1,2,1), (2,0,2) (1,0,2),(2,0,2) (2,1,2), $\newline
$(2,0,2) (2,2,2),(2,0,2) (2,0,3),(0,2,2) (2,2,2), (0,2,2) (0,0,2),(0,2,2) (0,2,3),$\newline
$(2,2,2) (1,2,2),(2,2,2) (2,2,3),(0,1,3) (1,1,3), (0,1,3) (0,0,3),(0,1,3)(0,2,3),$\newline
$(0,1,3) (0,1,4),(1,1,3)(2,1,3),(1,1,3) (1,0,3), (1,1,3) (1,1,4), (0,2,3) (1,2,3),$\newline
$(0,2,3) (0,2,4), (1,0,4) (2,0,4), (1,0,4) (1,1,4), (1,0,4) (1,2,4),(1,0,4) (1,0,5),$\newline
$(2,0,4) (0,0,4), (2,0,4) (2,2,4), (2,0,4) (2,0,5),(1,1,4) (2,1,4),(1,1,4) (1,1,5)\}.$
\end{appendices}

\end{document}